\documentclass[jcs]{iosart1c}

\usepackage[T1]{fontenc}
\usepackage{times}%

\usepackage[numbers]{natbib}
\usepackage{dcolumn}
\usepackage{graphics}

\newcolumntype{d}[1]{D{.}{.}{#1}}

\firstpage{1} \lastpage{7} \volume{1} \pubyear{2014}

\usepackage{proof}
\usepackage{amssymb}
\usepackage{amsmath}
\usepackage{amsthm}
\usepackage{stmaryrd}
\usepackage[all]{xy}
\RequirePackage{txfonts}
\DeclareMathAlphabet{\mathsl}{OT1}{cmr}{m}{sl}

\usepackage{color}
\usepackage{url}
\usepackage{ifpdf}
\usepackage{lineno}
\usepackage{balance}
\usepackage{subcaption}
\usepackage{lipsum}
\usepackage{wrapfig}
\usepackage{mathtools}
\ifpdf
\usepackage[colorlinks=true,linkcolor=blue,citecolor=blue]{hyperref}
\fi

\usepackage[multiple]{footmisc}

\long\def\comment#1{}
\newcommand{\eg}{{\em e.g.}}
\newcommand{\ie}{{\em i.e.}}
\newcommand{\etal}{\emph{et al.}}

\newcommand\uSize{\ensuremath{k}}
\newcommand\iSize{\ensuremath{m}}

\newcommand{\Ascr}{\mathcal{A}}
\newcommand{\Pscr}{\mathcal{P}}
\newcommand{\Uscr}{\mathcal{U}}

\newcommand{\Cscr}{\mathcal{C}}

\newcommand{\Sscr}{\mathcal{S}}
\newcommand{\Rscr}{\mathcal{R}}
\newcommand{\Tscr}{\mathcal{T}}

\newcommand{\Zscr}{\mathcal{Z}}
\newcommand{\Gscr}{\mathcal{G}}

\newcommand{\bluend}{ {\color{blue} $\blacktriangleleft$ }}
\newcommand{\blue}{ {\color{blue} $\blacktriangleright$ }}

\newcommand{\tup}[1]{\langle#1\rangle}

\newcommand{\arr}{\rotatebox[origin=c]{-90}{$\leftrightsquigarrow$}}

 \newtheorem{theorem}{\bf Theorem}[section]
\newtheorem{definition}[theorem]{\bf Definition}
\newtheorem{lemma}[theorem]{\bf Lemma}
 \newtheorem{proposition}[theorem]{\bf Proposition}
\newtheorem{remark}[theorem]{\bf Remark}
\newtheorem{corollary}[theorem]{\bf Corollary}

\renewenvironment{paragraph}[1]{\smallskip\noindent{\it #1}~}

\usepackage{tweaklist}
\renewcommand{\enumhook}{
	\setlength{\topsep}{0em}
	\setlength{\itemsep}{0em}
	\setlength{\partopsep}{0em}
	\setlength{\parskip}{0em}
	\setlength{\parsep}{0em}
	\addtolength{\leftmargin}{-0.75em}
}
\renewcommand{\itemhook}{
	\setlength{\topsep}{0em}
	\setlength{\itemsep}{0em}
	\setlength{\partopsep}{0em}
	\setlength{\parskip}{0em}
	\setlength{\parsep}{0em}
	\addtolength{\leftmargin}{-0.75em}
}

\newcommand\lra{\longrightarrow}

\newcommand\Int[1]{int\,(#1)}
\newcommand\dec[1]{dec\,(#1)}
\newcommand\Dmax{D_{max}}

\begin{document}
\begin{frontmatter}                           

%


\title{Time, Computational Complexity, and Probability in the Analysis of Distance-Bounding Protocols}

\runningtitle{Time, Computational Complexity, and Probability in the Analysis of Distance-Bounding Protocols}

\author[B,C]{\fnms{Max} \snm{Kanovich}},
\author[D]{\fnms{Tajana} \snm{Ban Kirigin}},
\author[E,H]{\fnms{Vivek} \snm{Nigam}},
\author[F,C]{\fnms{Andre} \snm{Scedrov}}
and
\author[G]{\fnms{Carolyn} \snm{Talcott}}
\runningauthor{M. Kanovich et al.}
\address[B]{Department of Computer Science (UCL-CS),  University College London, London, UK\\ E-mail: m.kanovich@ucl.ac.uk}
\address[C]{Faculty of Computer Science, National Research University Higher School of Economics, Moscow,  Russian Federation}
\address[D]{Department of Mathematics, University of Rijeka, Rijeka, Croatia \\ E-mail:bank@math.uniri.hr}
\address[E]{Computer Science Department, Federal University of Para\'iba, Jo\~ao Pessoa, Brazil \\ E-mail: vivek@ci.ufpb.br}
\address[F]{Department of Mathematics, University of Pennsylvania, Philadelphia, PA, USA \\ E-mail: scedrov@math.upenn.edu}
\address[G]{Computer Science Laboratory, SRI International,  Menlo Park, CA, USA \\ E-mail: clt@csl.sri.com}
\address[H]{fortiss, Munich, Germany}

\begin{abstract}
Many security protocols rely on the assumptions on the physical properties in which its protocol sessions will be carried out. For instance, Distance Bounding Protocols take into account the round trip time of messages and the transmission velocity to infer an upper bound of the distance between two agents. We classify such security protocols as Cyber-Physical. 
Time plays a key role in design and analysis of many of these protocols. This paper investigates the foundational differences and the impacts on the analysis when using models with discrete time and models with dense time. We show that there are attacks that can be found by models using dense time, but not when using discrete time. We illustrate this with an attack that can be carried out on most Distance Bounding Protocols. In this attack, one exploits the execution delay of instructions during one clock cycle to convince a verifier that he is in a location different from his actual position. We additionally present a probabilistic analysis of this novel attack. As a formal model for representing and analyzing Cyber-Physical properties, we propose a Multiset Rewriting model with dense time suitable for specifying cyber-physical security protocols. 
We introduce Circle-Configurations and show that they can be used to symbolically solve the reachability problem for our model, and  show that for the important class of balanced theories the reachability problem is PSPACE-complete.
We also  show how our model can be implemented using the computational rewriting tool Maude,  the machinery that automatically searches for such attacks. 
\end{abstract}


\end{frontmatter}





\section{Introduction}
\label{sec:intro}

With the development of pervasive cyber-physical systems and consequent security
issues, it is often necessary to specify protocols that not only make use of cryptographic keys and nonces,\footnote{ In protocol security literature~\cite{cervesato99csfw,durgin04jcs} fresh values are usually called \emph{nonces. }} but also take into account the physical properties of the environment where its protocol sessions are carried out. We call such protocols \emph{Cyber-Physical Security Protocols}. For instance, Distance Bounding Protocols~\cite{brands93eurocrypt} is a class of cyber-physical security protocols which infers an upper bound on the distance between two agents from the round trip time of messages.
In a distance bounding protocol session, the verifier ($V$) and the prover ($P$) exchange messages:
\begin{equation} \label{eq: DBP}
 \begin{array}{l}
  V \lra P : m \\
  P \lra V : m'
 \end{array}
\end{equation}
where $m$ is a challenge and $m'$ is a response message (constructed using $m$'s components). To infer the distance to the prover, the verifier remembers the time, $t_0$, when the message $m$ was sent, and  the time, $t_1$, when the message $m'$ returns. From the difference $t_1 - t_0$ and the assumptions on the speed of the transmission medium, $v$, the verifier can compute an upper bound on the distance to the prover, namely $(t_1 - t_0) \times v$. 
Typically, the verifier grants the access to the prover if the inferred upper bound on the distance  does not exceed some
pre-established fixed \emph{time response bound}, R, given by the protocol specification.

Distance bounding protocol sessions are used in a number of cyber-physical security protocols to infer an upper-bound on the distance of participants. Examples include Secure Neighbor Discovery, Secure Localization Protocols~\cite{tippenhauer09esorics, capkun06sacom,shmatikov07asian}, and Secure Time Synchronization Protocols~\cite{sun06ccs,ganeriwal08iss}. The common feature in most cyber-physical security protocols is that they mention cryptographic keys, nonces and time. (For more examples, see \cite{basin11iss,meadows07booktitle,cremers12oakland} and references therein.)

A major problem of using the traditional protocol notation for the description of distance bounding protocols,  as in (\ref{eq: DBP}), is that many assumptions about time, such as the time requirements for the fulfillment of a protocol session, are not formally specified. 
It is only informally described that the verifier remembers the time $t_0$ and $t_1$ and which exact moments these correspond to. Moreover, from the above description, it is not clear which assumptions about the network are used, such as the transmission medium used by the participants. Furthermore, it is not formally specified which properties does the above protocol ensure, and in which conditions and against which intruders.

It is easy to check that the above protocol is not safe against the standard Dolev-Yao intruder~\cite{DY83} who is capable of intercepting and sending messages anywhere at anytime. The Dolev-Yao intruder can easily convince $V$ that $P$ is closer than he actually is. The intruder first intercepts the message $m$ and with zero transmission time sends it to $P$. Then he intercepts the message $m'$ and instantaneously sends it to $V$, reducing the round-trip-time ($t_1 - t_0$).
 Thus, $V$ will believe that $P$ is much closer than he actually is. Such an attack does not occur in practice as messages take time to travel from one point to another. Indeed, the standard Dolev-Yao intruder model is not a suitable model for the 
analysis of cyber-physical protocols. Since such an intruder is able to intercept and send messages anywhere at anytime, he appears faster than the speed of light. In fact, a major difference between cyber-physical protocols and traditional security protocols is that there is not necessarily a network in the traditional sense, as any transmission medium used is the part of the network. 

Existing works have proposed and used models with time for the analysis of distance bounding protocols where the attacker is constrained by some physical properties of the system. Some models have considered dense time~\cite{basin11iss}, while others have used discrete time~\cite{boureanu13iacr}. However, although these models have included time, the foundational differences between  these models and the impacts to analysis has not been investigated 
in more detail. For example, these models have not investigated the fact that provers, verifiers, and attackers may have different clock rates, \ie, processing speeds, which can affect security. 
 Indeed, already in the original paper on Distance Bounding Protocols, Brands and Chaum~\cite{brands93eurocrypt} suspected that a verifier may be subject to attack as it uses discrete clock ticking, \ie, a processor, and thus measures time in discrete units, while the environment and the powerful attacker is not limited by a particular clock. 
However, until now, no careful analysis of such attacks has been carried out.
This paper addresses this gap.

A key observation of this paper is that models with dense time abstract the fact that attacker clocks may tick at any rate.
The attacker can mask his location by exploiting the fact that a message may be sent at any point between two clock ticks of the verifier's clock, while the verifier believes that it was sent at a particular time. Depending on the speed of the verifier, \ie, its clock rate, the attacker can \emph{in principle} convince the verifier that he is very close to the verifier (less than a meter) even though he is very far away (many meters away). We call this attack \emph{attack in-between-ticks}.

Interestingly, from a foundational point of view, in our formalization there is no complexity increase when using a model with dense time when compared to a model with discrete time. In our previous work~\cite{kanovich12rta, kanovich15mscs}, we proposed a rewriting framework which assumed discrete time. We showed that the reachability problem is PSPACE-complete. Here we show how to formalize systems with dense time, and that, if we extend the model with dense time, the reachability problem is still PSPACE-complete. For this result we introduce a novel machinery called Circle-Configurations. Moreover, we show that it is possible to automate verification of whether a protocol written in our model is subject to an instance of an attack in-between-ticks. 

However, our symbolic analysis only provides an yes or no answer, that is, it is only possible to infer whether a system is or not subject to an attack in-between-ticks. For a more accurate analysis, we construct a probabilistic model for analysing Distance Bounding Protocols. This model provides precise measure of how probable a system is subject to an attack. A difference, however, to the proposed symbolic model is that it seems less likely to carry out analysis in an automated fashion. We leave this investigation to future work.

\vspace{2mm}
The paper is organised as follows.
Section~\ref{sec:motivation} contains the novel attack in-between-ticks. In Section~\ref{sec:msr} we introduce a formal model based on Multiset Rewriting (MSR) which includes dense time. We also show how to specify distance bounding protocols in this language. 
In Section \ref{s-prob-main}  we provide a probabilistic analysis related to distance bounding protocols and the attack in-between-ticks. 
In Section \ref{sec:Maude} we show how our model can be implemented using the computational rewriting tool Maude. In particular, Maude  is able to automatically find the  novel attack in-between-ticks.
Section~\ref{sec:circle} introduces a novel machinery, called circle-configurations, that allows one to symbolically represent configurations that mention dense time.  In Section~\ref{sec:complex} we prove that the reachability problem for timed bounded memory protocols~\cite{kanovich13esorics} is PSPACE-complete. Finally, in Section~\ref{sec:related}, we comment on related and future work. 

This paper considerably extends the conference paper~\cite{kanovich15post}. Besides adding the proofs of our complexity results, Section~\ref{s-prob-main} and \ref{sec:Maude} contain new material.

\section{Motivating Examples}
\label{sec:motivation}

In this section we point to some subtleties of cyber-physical protocol analysis and verification. 
We  present examples of protocols that serve to illustrate the differences between models with discrete and dense time.

One such an example is a timed version of the classical  \emph{Needham-Schroeder} protocol~\cite{needham78cacm}  which we already presented  in our conference paper \cite{kanovich15post}.
We have shown that the timed Needham-Schroeder protocol
may be safe in the discrete time model, but an attack (similar to Lowe attack~\cite{lowe96tacas}) is possible when using models with dense time. 
This phenomena shows that some attacks on  cyber-physical security protocols  may only be found when using models with dense time.

\subsection{Time-Bounding Needham-Schroeder Protocol}
We first show some  subtleties of cyber-physical security protocol analysis by re-examining the original  \emph{Needham-Schroeder} public key protocol \cite{needham78cacm} (NS), presented in Figure~\ref{fig: NS}. Although this protocol is well known to be insecure \cite{lowe96tacas}, we look at it from another dimension, the dimension of time.
We check whether Needham and Schroeder were right after all, in the sense that their protocol can be considered secure under some time requirements. In other words, we investigate whether NS can be fixed by means of time.

We timestamp each event in the protocol execution, that is, we explicitly mark the time of sending and receiving messages by a participant. We then propose a timed version of this protocol, called {\em Time-Bounding Needham-Schroeder Protocol} (Timed-NS), as depicted in Figure \ref{fig: timed NS}.
The protocol exchanges the same messages as in the original version, but the last protocol message, \ie,~the confirmation message 
\mbox{$ \{N_B\}_{K_{B}}$},  is sent by $A$ only if the time difference  \ $t_3 - t_0$ is smaller or equal to the given {\em response bounding time} $R$.

\begin{figure}[t]
\centering
\begin{subfigure}{.48\textwidth}
  \centering
\includegraphics[width={ 0.40\textwidth}]{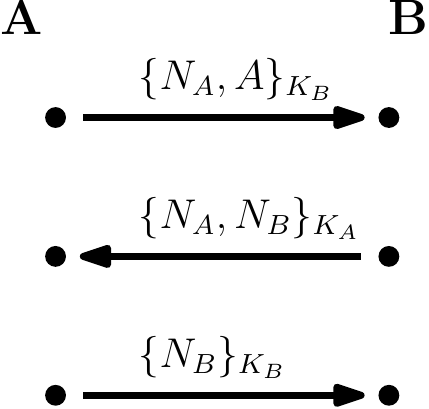}\medskip \bigskip
\caption{Needham-Schroeder protocol }
\label{fig: NS}
\end{subfigure}%
\quad
\begin{subfigure}{.48\textwidth}
  \centering
\includegraphics[width={ 0.40\textwidth}]{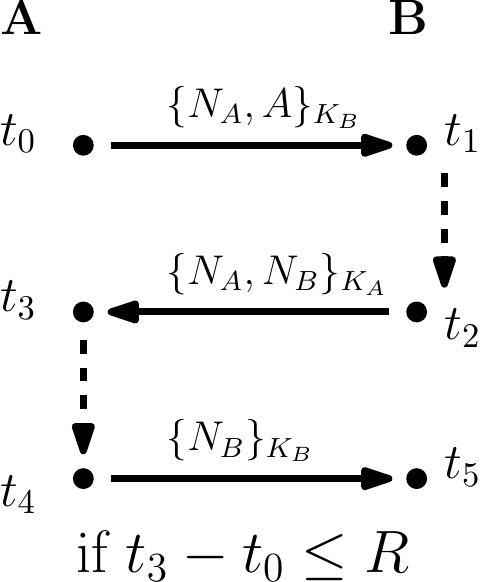} 
\caption{Timed Needham-Schroeder Protocol}
\label{fig: timed NS}
\end{subfigure}
\caption{Adding time to  Needham-Schroeder Protocol}
\label{fig:both NS}
\bigskip \medskip
\begin{subfigure}{.47\textwidth}
  \centering
\includegraphics[width={ 0,7\textwidth}]{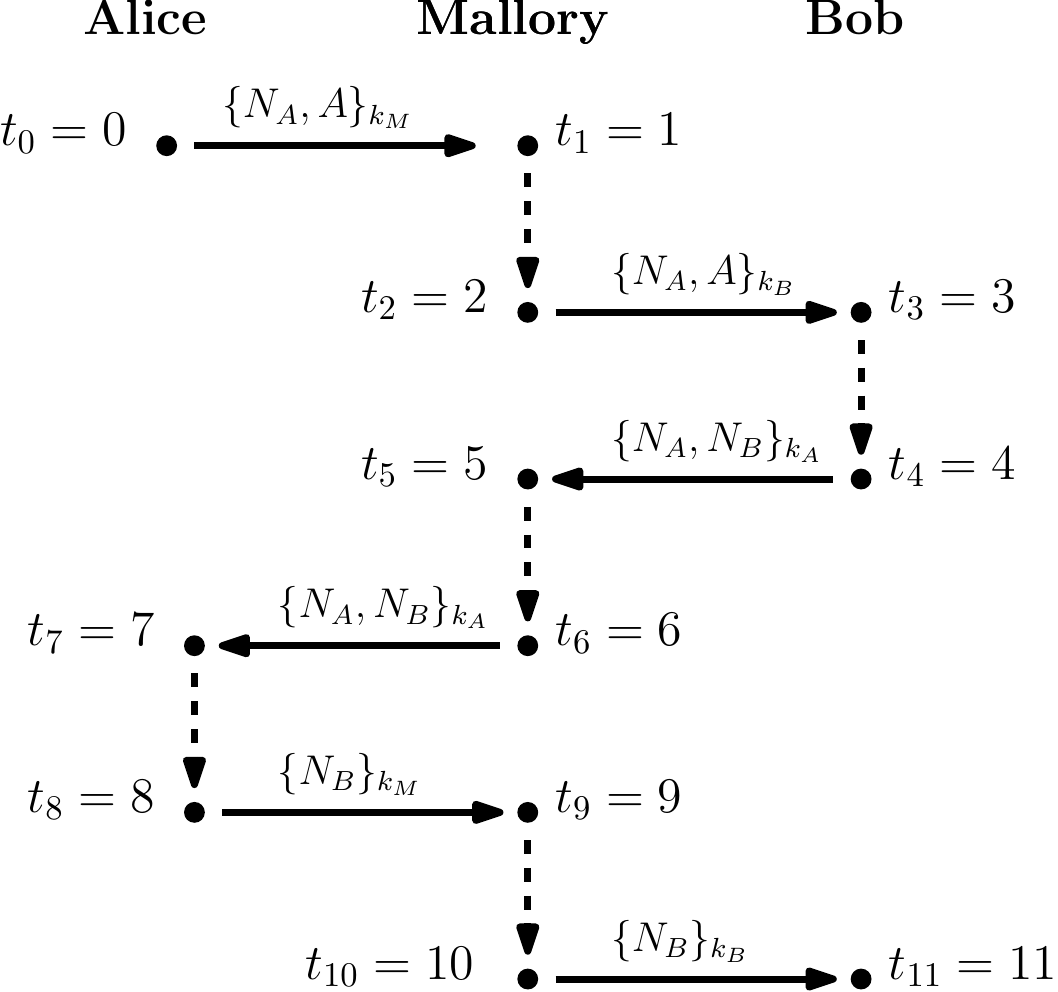}
\medskip
\caption{Discrete Time Model}
\label{fig: Lowe discrete}
\end{subfigure}%
\quad
\begin{subfigure}{.47\textwidth}
  \centering
\includegraphics[width={ 0,7\textwidth}]{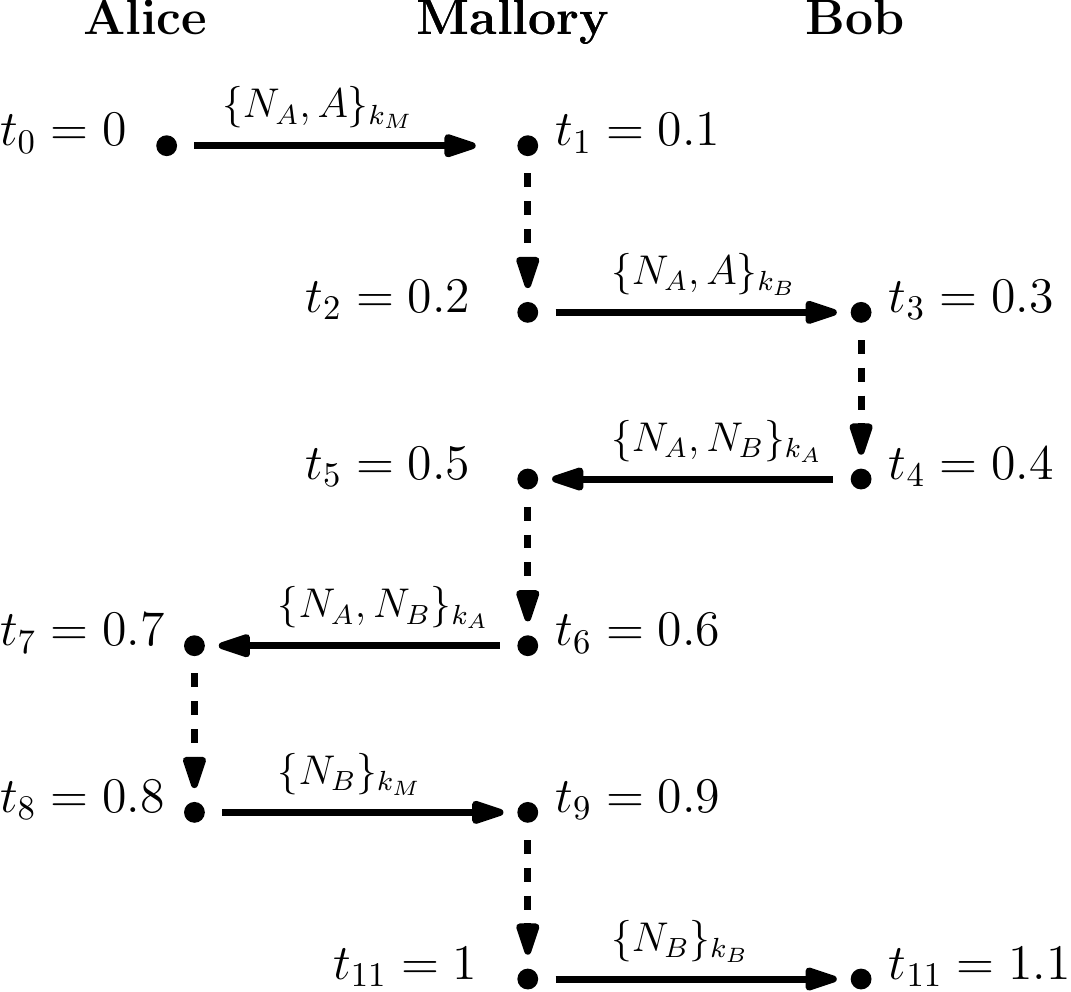} 
\medskip
\caption{Dense Time Model}
\label{fig: Lowe dense}
\end{subfigure}
\caption{Timed Version of Lowe Attack}
\vspace{2mm}
\label{fig:both Lowe}
\end{figure}

 The protocol is considered \emph{secure} in the standard way, that is, if  the ``accepted'' $N_A$ and $N_B$  may never be revealed to anybody else except Alice and Bob. Recall that the well known Lowe attack on NS~\cite{lowe96tacas} involves a third party, Mallory
who is able to learn Bob's nonce. At the same time Bob believes that he communicated with Alice and that only Alice learned his nonce.  

\vspace{2mm}
 The intriguing result of the analysis of Timed-NS is that one may not  find an attack in the discrete time model, but can find one in the dense time model:
 Figure \ref{fig:both Lowe} depicts the Lowe attack scenario in Timed-NS. 
In particular, the attack requires that  events marked with $t_0, \dots, t_7$ take place and that the round trip time of messages, that is  $t_7 - t_0$,  does not exceed the given response bounding time $R$. Assuming that both network delay and processing time are non-zero,
in the discrete time model the
 attack could be modeled only for response bounding time $R \geq 7$, see Figure \ref{fig: Lowe discrete}. 
In the discrete model, the protocol would seem safe for response bounding time $R < 7$.
However, in the dense time model the attack is possible for any positive response bounding time $R$,  see Figure \ref{fig: Lowe dense}. Indeed, assuming continuous times, the attacker can, in principle, be as fast as needed to satisfy the given response bounding time. Notice additionally, that if $R$ is set to be too low, then it may also turn the protocol unusable as legitimate players may not be able to satisfy the response time requirement. The same is true with Distance Bounding Protocols. If the distance bound is set to be too low, legitimate players may not be able to get access to resources. Therefore, response bounds should be set in such a way that makes the protocol secure and still grant resource to legitimate users. 

\vspace{2mm}
This simple example already illustrates the challenges of timed models for 
cyber-physical security protocol analysis and verification. No rescaling of discrete time units removes the presented difference between the models. For any discretization of time, such as seconds or any other infinitesimal time unit, there is a protocol for which there is an
attack with continuous time and no attack is possible in the discrete case.

We further illustrate  the challenges of timed models for cyber-physical security protocol analysis and verification 
 by a more realistic example of a distance bounding protocol.
The novel attack in-between-ticks appears and  illustrates that for the analysis of distance bounding protocols it is necessary to consider time assumptions of the players involved.

\subsection{Attack In-Between-Ticks}
\label{sec: attack-in-b-t}

Regardless of the design details of a specific distance bounding protocol, we identify a new type of anomaly. We call it {\em Attack In-Between-Ticks}. This attack is particularly harmful when the verifier and the prover exchange messages using radio-frequency (RF), where the speed of transmission is the speed of light. In this case an error of a $1$ nanosecond ($ns$) already results in a distance error of $30cm$. This is the case with a number of cyber-physical security protocols which try to establish the property that participants are \emph{physically close} to each other, \eg, Secure Neighbor Discovery, Secure Localization Protocols~\cite{tippenhauer09esorics, capkun06sacom,shmatikov07asian}, Secure Time Synchronization Protocols~\cite{sun06ccs,ganeriwal08iss}, Protocols used by RFID payment devices~\cite{chothia15fcds,chothia10fcds}.

\begin{figure}[t]
\centering
\vspace{1mm}
\begin{subfigure}{.48\textwidth}
  \centering
\includegraphics[width={ 1\textwidth}]{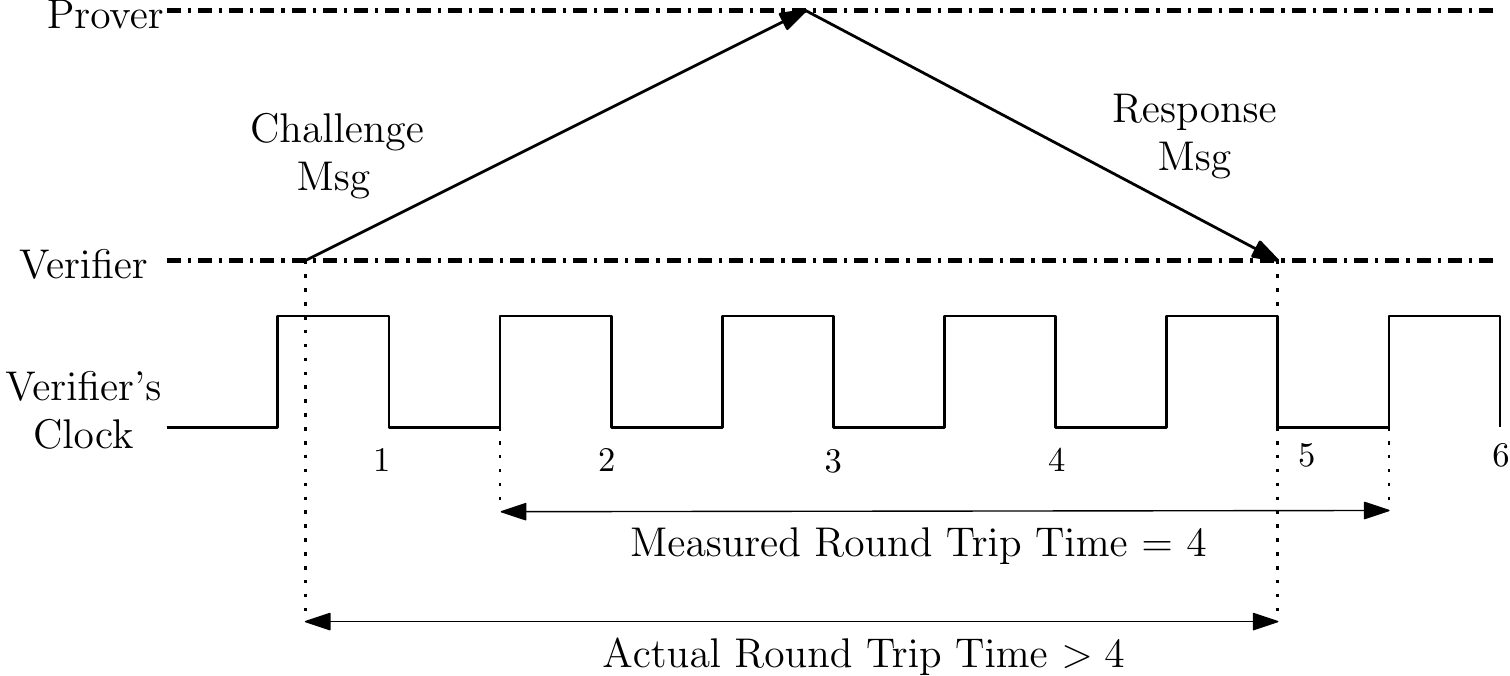}
\vspace{0,5mm}
\caption{In different ticks (Sequential Execution)}
\label{fig:in-btw-clock-attack1}
\end{subfigure}%
\quad
\begin{subfigure}{.47\textwidth}
  \centering
\includegraphics[width={ 1\textwidth}]{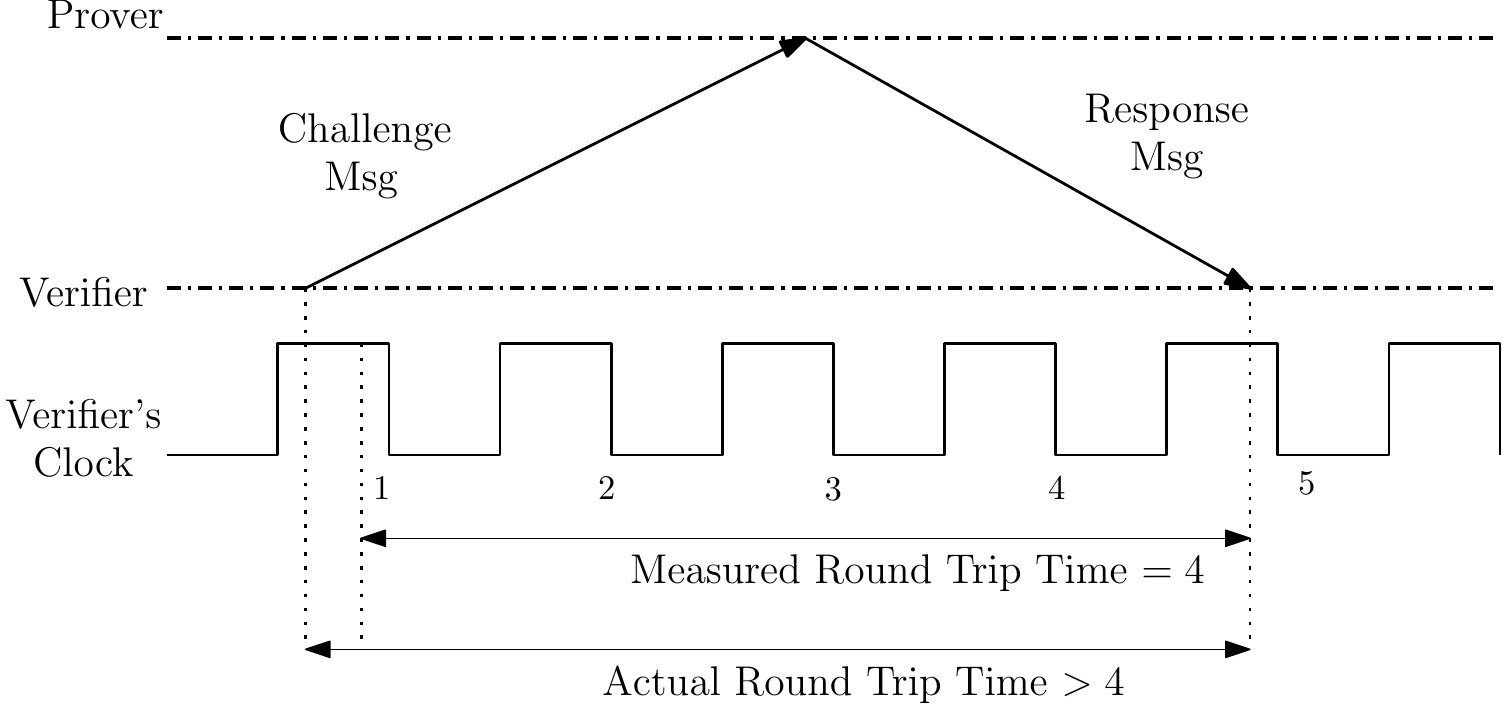}
\vspace{0,5mm}
\caption{In the same tick (Parallel Execution)}
\label{fig:in-btw-clock-attack2}
\end{subfigure}
\vspace{1mm}
\caption{Attack In-Between-Ticks. Here time response bound $R = 4$ ticks. 
{\color{blue}  }
}
\label{fig: in-btw-clock-attack}
\vspace{3mm}
\end{figure}

Consider the illustrations in Figure \ref{fig: in-btw-clock-attack}. They depict the execution of instructions by the verifier in a distance bounding protocol with the time response bound $R=4$. The verifier has to execute two instructions: (1) the instruction that sends the signal to the prover and (2) the instruction that measures the time when this message is sent. Figure~\ref{fig:in-btw-clock-attack1} illustrates the case when the verifier is running a sequential machine (that is, a single processor), which is the typical case as the verifier is usually a not very powerful device, \eg, door opening device. Here we optimistically assume that an instruction can be executed in one cycle. When the first instruction is executed, it means that the signal is sent somewhere when the clock is up, say at time $0.6$. In the following clock cycle, the verifier remembers the time when the message is sent. Say that this was already done at time $1.5$. If the response message is received at time $5$  it triggers an interruption so that the verifier measures the response time in the following cycle, \ie, at time $5.5$. Thus the measured round time is $4 = 5.5 - 1.5 = R$ ticks. Therefore, the verifier grants access to the prover although the actual round trip is $5 - 0.6 = 4.4 > R = 4$ ticks. This means that the verifier is granting access to the prover although the prover's distance to the verifier may not satisfy the distance bound and thus this is a security flaw.\footnote{Notice that inverting the order of the instructions, \ie, first collecting the time and then sending the signal, would imply errors of measurement but in the opposite direction turning the system impractical.} 

\vspace{2mm}
Depending on the speed of verifier's processor, the difference of $0.4$ tick results in a huge error. Many of these devices use very weak processors.\footnote{High precision verifiers are expensive and normally used only in high security applications.} The one proposed in~\cite{srinivasa13iitee}, for example, executes at a frequency of at most $24MHz$. This means a tick is equal to $41ns$ (in the best case). Thus, an error of $0.4$ tick corresponds to an error of $16ns$ or an error of \emph{$2 \times 2.4$ meters} when using RF. In the worst case, the error can be of $1.0$ tick when the signal is sent at the beginning of the cycle, \ie, at time $0.5$ tick, and the measurements at the end of the corresponding cycles, \ie, at times $2$ and $6$ ticks. An error of $1.0$ tick ($41ns$) corresponds to an error greater than \emph{$2 \times 6.15$ meters}.    

Consider now the case when the verifier can execute both instructions in the same cycle. Even in this case there might be errors in measurement as illustrated in Figure~\ref{fig:in-btw-clock-attack2}. It may happen that the signal is sent before the measurement is taken thus leading to errors of at most $0.5$ ticks (not as great as in the sequential case). (Here we are again optimistically assuming that an instruction can be executed in one cycle.)

\subsection{Vulnerability of Cyber-Physical Security Protocols to Attack In-Between-Ticks}
\label{sec: DB-example}

The attack in-between-ticks is based on the foundational difference between real-time in nature and time management by discrete time processors, irrespective of their physical-layer implementations and distance bounding protocol design.
The  attack  in-between-ticks appears because of the discrepancy between actual real-time distance and the upper bound on  time distance that is calculated by a discrete verifier from the time taken as the time of sending and receiving challenge and response messages, respectively.

Distance bounding protocol sessions are used in a number of cyber-physical security protocols, such as Secure Neighbor Discovery, Secure Localization Protocols~\cite{tippenhauer09esorics, capkun06sacom,shmatikov07asian}, and Secure Time Synchronization Protocols~\cite{sun06ccs,ganeriwal08iss}. Therefore, all these protocols can be vulnerable to the attack in-between-ticks.

Such a discrepancy has been suspected in Brands and Chaum's original paper introducing distance bounding protocols~\cite{brands93eurocrypt}. However, no study on the vulnerabilities of this discrepancy has been investigated until now. The attack in-between-ticks represents a \emph{new category of  attacks that needs to be considered in the analysis of cyber-physical protocols.} 

From the well-known types of attacks, such as Distance Fraud, Mafia Fraud, Terrorist Fraud and Distance Hijacking \cite{cremers12oakland}, the distance fraud, or the corresponding classification of Lone Distance Fraud from~\cite{cremers12oakland} is the closest to the attack in-between-ticks. It also leads the verifier to believe that a prover is closer than he actually is, and it does not involve additional participants, only a single prover and a verifier.
However, in (lone) distance fraud, the prover is assumed to be dishonest, trying to change his distance and appear closer to the verifier than he actually is. This can be achieved in some cases by, \eg, prover sending the response message(s) too soon, before receiving the challenge message(s), and can be fixed by changes in the  protocol design, so that the response messages are dependant of challenge messages sent by the verifier.
On the other hand, in the attack in-between-ticks, even an honest prover 
can result closer than he actually is, not by his intent, but because of the discrepancy between real time and discrete processor clocks.

\vspace{0.5em}
In order to obtain high-resolution timing information about the arrival of individual data bits, distance bounding protocol, as well as secure-positioning protocols,  are tightly integrated into the physical layer of the communication protocol, providing  sub-microsecond timing information. In this way these protocols rely directly on the laws of physics, \ie, on the assumption that the communication is bounded by the speed of light. Nevertheless, as we have shown above, substantial errors in measurement, up to several meters, may appear.
Verifier may allow the access to a prover that is considerably  outside the distance bound specified by the protocol.
This would equally appear in any distance bounding protocol, \eg, in Brands and Chaum protocol~\cite{brands93eurocrypt} or Hancke-Kuhn protocol~\cite{hancke2005rfid}.

\vspace{0.5em}
Formal frameworks for reasoning about distance bounding protocols that are based on discrete time models, such as \cite{boureanu13iacr}, clearly cannot capture the attack in-between-ticks. This is also the case with the real-time formalisms that do not take into account the way verifiers with discrete clocks operate. None of the real-time formalisms that we are aware of formalizes this treatment of time that is typical of processors.

\vspace{0.5em}
Furthermore, we observe that these security flaws may happen \emph{in principle}. In practice, distance bounding protocols carry out a large number of challenge and response rounds. This is generally  believed to  mitigate the chances of  attacks occurring. In the future, we intend to investigate the challenge-response approach in providing security of distance bounding protocols. 
In particular, we are planning to analyze  whether the effects of the attack  in-between-ticks
 can be reduced by repeated challenge-response rounds of protocols. 

Finally, we also point out that these attacks have been inspired by similar issues in the analysis of digital circuits~\cite{alur04sfm}.

\section{A Multiset Rewriting Framework with Dense Time}
\label{sec:msr}

For our  multiset rewriting framework we assume a finite first-order typed alphabet, $\Sigma$, with variables, constants, function and predicate symbols. Terms and facts are constructed as usual (see~\cite{enderton}) by applying symbols with correct type (or sort). For instance,  if $P$ is a predicate of type $\tau_1 \times \tau_2 \times \cdots \times \tau_n \rightarrow o$, where $o$ is the type for propositions, and $u_1, \ldots, u_n$ are terms of types $\tau_1, \ldots, \tau_n$, respectively, then $P(u_1, \ldots, u_n)$ is a \emph{fact}. A fact is grounded if it does not contain any variables. 

In order to specify systems that explicitly mention time, we use \emph{timestamped facts} of the form $F@T$, 
where $F$ is a fact and $T$ is its timestamp. In our previous work~\cite{kanovich15mscs}, timestamps 
were only allowed to be natural numbers. Here, on the other hand, in order to express dense time, timestamps are allowed to be non-negative real numbers. 

We assume that there is a special predicate symbol $Time$ with arity zero, which will be used to represent
the global time. A \emph{configuration} is a multiset of ground timestamped facts, 
$$\{~Time@t, ~F_1@t_1, \ldots, ~F_n@t_n~\},$$
with a single occurrence of a $Time$ fact. 
Configurations are to be interpreted as states of the system. For example, the following configuration
\begin{equation}\label{eq: conf1}
 \{~Time@7.5, ~Deadline@10.3, ~Task~(1, \textrm{done})@5.3, ~Task~(2, \textrm{pending})@2.13~\}
\end{equation}
specifies that the current global time is $7.5$, the Task 1 was performed at time 5.3, Task 2  was issued at time 2.13 and is still pending, and that the deadline to perform all tasks is 10.3.

We may sometimes denote the timestamp of a fact $F$ in a given configuration as $T_F$.

\subsection{Actions and Constraints} 
\label{subsec:actions}
Actions are multiset rewrite rules and are either the
time advancement action or instantaneous actions.
The action representing the advancement of time,
called \emph{Tick Action}, is the following:
\begin{equation}
\label{eq:time-advancement}
 Time@T \lra Time@(T + \varepsilon)
\end{equation}
Here $\varepsilon$ can be instantiated by any positive real number specifying that the global 
time of a configuration can advance by any positive number. For example,
if we apply this action with $\varepsilon = 0.6$ to the configuration (\ref{eq: conf1}) we obtain the  configuration 
\begin{equation}\label{eq: conf2}
\{~Time@8.1, ~Deadline@10.3, ~Task~(1, \textrm{done})@5.3, ~Task~(2, \textrm{pending})@2.13~\}
\end{equation}
where the global time advanced from $7.5$ to $8.1$.

Clearly such an action is a source of unboundedness as time can always advance by any positive real number. 
In particular we will need to deal with issues such as Zeno Paradoxes~\cite{alur15CPS} when considering how time should advance.

The remaining actions are the Instantaneous Actions, which do not affect the global time, 
but may rewrite the remaining facts. 

\vspace{2mm}
\begin{definition}
\emph{ Instantaneous Actions} are actions of the form:
\begin{equation}
\label{eq:instantaneous}
\begin{array}{l}
 Time@T, ~W_1@T_1, \ldots, ~W_k@T_k, ~F_1@T_1', \ldots, ~F_n@T_n' ~\mid ~\Cscr ~\lra \\
 \quad \exists~\vec{X}.~[~Time@T, ~W_1@T_1, \ldots, ~W_k@T_k, ~Q_1@(T + D_1), \ldots, ~Q_m@(T + D_m)~]
\end{array}
\end{equation}
where $D_1, \ldots, D_m$ are natural numbers,   existentially quantified variables $\vec{X}$ denote fresh values that are created by the rule,  and  $\Cscr$ is the guard of the action which is a set of constraints
involving the time variables appearing in the pre-condition, \ie,~the variables $T, T_1, \ldots, T_k, T_1', \ldots, T_n'$. 
\\
\emph{Time Constraints} are expressions of the form:
\begin{equation}
\label{eq:constraints}
  T' > T'' \pm D \quad \textrm{ and } \quad  T' = T'' \pm D
\end{equation}
where $T'$ and $T''$ are time variables, and $D$ is a natural number. \footnote{Here, and in the rest of the paper,  the symbol $\pm$ stands for either $+$ or $-$, that is,  constraints may involve addition or subtraction.} 
\end{definition}

\vspace{2mm}
We use $T' \geq T'' \pm D$ to denote the disjunction of $T' > T'' \pm D$ and $T' = T'' \pm D$.

An instantaneous action can only be applied if all 
the constraints in its guard  are satisfied. 

We say that facts $F_i@T_i'$ are \emph{consumed} \ and that the facts $Q_i@(T+D_i)$ are \emph{created }by the rule (\ref{eq:instantaneous}).

\vspace{3mm}
Notice that we allow only natural numbers for constants $D$s and $D_i$s that appear in time constraints and timestamps of created facts. We impose such conditions because of the relevant computational complexity issues. However, this is not as restrictive w.r.t. expressivness of the model as one might think. We will address this issue in more detail later on in Section~\ref{sec:complex} when we investigate the complexity of related computational problems.

\vspace{3mm}
Also, notice that the global time does not change when applying an instantaneous action. Moreover, 
the timestamps of the facts that are created  by the action, 
namely the facts $Q_1, \ldots, Q_m$, are of the form 
$T + D_i$, where $D_i$ is a natural number and $T$ is the global time. That is, their timestamps are in the present or the future. 
For example, the following is an instantaneous action
$$
\begin{array}{l}
 Time@T, ~Task~(1, \textrm{done})@T_1, ~Deadline@T_2, ~Task~(2, \textrm{pending})@T_3 ~\mid ~\{~T_2 \geq T + 2~\}  \\
 \qquad \lra ~Time@T, ~Task~(1, \textrm{done})@T_1, ~Deadline@T_2, ~Task~(2, \textrm{done})@(T + 1)
\end{array}
$$
\noindent
which specifies that one should complete Task 2, if Task 1 is already completed, and moreover, if the Deadline is 
at least 2 units ahead of the current time. If these conditions are satisfied, then the Task 2 will be completed
in one time unit. Applying this action to the configuration (\ref{eq: conf2}) yields
$$
\{ ~Time@8.1, ~Deadline@10.3, ~Task~(1, \textrm{done})@5.3, ~Task~(2, \textrm{done})@9.1 ~\}
$$
  \noindent
where Task 2 will be completed by the time $9.1$.

Finally, the variables $\vec{X}$ that are existentially quantified in (\ref{eq:instantaneous}) 
are to be replaced by fresh values, also called \emph{nonces} in protocol security literature~\cite{cervesato99csfw,durgin04jcs}. 
For example, the following action specifies the creation of a new task with a fresh identifier $Id$, 
which should be completed by time $T + D$:
$$
\hspace{-3mm} Time@T  \lra  \exists~Id.~[~Time@T, ~Task~(Id, \textrm{pending})@(T + D)~]
$$
\noindent
Whenever this action is applied to a configuration, the variable $Id$ is instantiated by a 
fresh value. In this way we are able to specify that the identifier assigned to the new task is different to the 
identifiers of all other existing tasks. In the same way it is possible to specify the 
use of nonces in Protocol Security~\cite{cervesato99csfw,durgin04jcs}.

\vspace{2mm}
Formally, a rule $W \mid \Cscr \lra~\exists~{\vec{X}}.~W'$ can be applied to a configuration $\Sscr$ if there is a ground substitution $\sigma$, where the variables in $\vec{X}$ are fresh, such that $W\sigma \subseteq \Sscr$ and $\Cscr\sigma$ is true. The resulting configuration is $(\Sscr \setminus W) \cup W'\sigma$.

\vspace{2mm}
Notice that by the nature of multiset rewriting there are various aspects of 
non-determinism in the model. For example, different actions and even different instantiations of the 
same rule may be applicable to the same configuration $\Sscr$, which may 
lead to different resulting configurations $\Sscr'$.

More precisely, an instance of an action is obtained by substituting all variables appearing in the pre- and post-condition of the
action with constants. That applies to variables appearing in terms inside facts, variables representing fresh values, as well as time variables used in specifying timestamps of facts. For example, consider the following action
\begin{equation}
\label{eq:example-instance}
\begin{array}{l}
 Time@T,  ~Finaltask~(x,\textrm{done})@T, ~Deadline@ T_1 \ \mid \ \{~ T_1 \geq T + 3~\}\\
 \qquad   \lra  \exists~n.~[~Time@T, ~File~(n,x, \textrm{pending})@(T + 2), ~Deadline@T_1~]
\end{array}
\end{equation}
\noindent
specifying that a case record for a completed process should be filed under a unique record label within $2$ time units.
An instance of above action is obtained by substituting constants for variables $x$, $T$, $T_1$ and $n$. 
For example, 
$$
\begin{array}{l}
 Time@3.2,  ~Finaltask~(\textrm{C},\textrm{done})@3.2, ~Deadline@10\ \mid \ \{~10 \geq 3.2+ 3~\}\\
 \qquad   \lra  ~Time@3.2, ~File~(\textrm{N},\textrm{C}, \textrm{pending})@5.2, ~Deadline@10
\end{array}
$$
is an instance of action (\ref{eq:example-instance}).
Recall that an instance of an  action can only be applied to a configuration containing all the facts from its precondition if all the corresponding time constraints, \ie,~all the time constraints in its guard, are satisfied. For example, since $ 15.3 \geq 8.1+3$,  action (\ref{eq:example-instance}) is applicable to configuration
$$
\{\ Time@8.1, ~  Finaltask~(\textrm{C}_0,\textrm{done})@8.1, ~Deadline@15.3 \ \} \ , 
$$
resulting in configuration \ 
$$
\{\ Time@8.1, ~ File~(\textrm{N}_1,\textrm{C}_0, \textrm{pending})@10.1, ~Deadline@15.3 \ \} \ , 
$$
but it is not applicable to 
 the following configuration
$$
\{\ Time@13.1, ~  Finaltask~(\textrm{C}_0,\textrm{done})@13.1, ~Deadline@15.3 \ \} \ .
$$


\subsection{Initial, Goal Configurations and The Reachability Problem}

We write \ $\Sscr \lra_r \Sscr_1$ \ for the one-step relation where configuration $\Sscr$ is 
rewritten to $\Sscr_1$ using an instance of action $r$. For a set of actions $\Rscr$, we 
define \ $\Sscr \lra_\Rscr^* \Sscr_1$ \ as the transitive reflexive closure of the one-step relation on all actions in $\Rscr$.
We elide the subscript \ $\Rscr$, when it is clear from the context. 

\vspace{2mm}
\begin{definition}
A \emph{goal} $\Sscr_G$ is a pair of a multiset of facts and a set of constraints, written
$$
 \{~F_1@T_1, \ldots, ~F_n@T_n~\}~ \mid ~\Cscr
$$
\noindent
where \ $T_1, \ldots, T_n$ \ are time variables, \ $F_1, \ldots, F_n$ \ are facts and 
$\Cscr$ is a set of constraints involving only \ $T_1, \ldots, T_n$. 
We call a configuration $\Sscr_1$ a \emph{goal configuration w.r.t. goal} $\Sscr_G$  if there is a grounding substitution $\sigma$ replacing term variables by ground terms and time variables \ by real numbers such that \ $\Sscr_G \sigma \subseteq \Sscr_1$
and all the constraints in \ $\Cscr\sigma$ \ are satisfied. 
\end{definition}

For simplicity, since goal will usually be clear from the context, we will use terminology {\em goal configuration} eliding the goal w.r.t.  which it is defined.

\vspace{2mm}
The reachability problem  is then defined for a given initial configuration,
a goal and a set of actions.

\vspace{2mm}
\begin{definition}
Given an initial configuration $\Sscr_I$, a goal $\Sscr_G$ and a set of actions $\Rscr$, the \emph{reachability problem } $\Tscr$ is the problem of establishing whether there is a goal configuration $\Sscr_1$, 
such that ~$\Sscr_I \lra_\Rscr^* \Sscr_1$. 
Such a sequence of actions leading from an initial to a goal configuration is called a \emph{plan}.
\end{definition}

We assume that goals are invariant to  renaming of fresh values, that is, a goal configuration ~$\Sscr_G$ is equivalent to the goal configuration ~$\Sscr_G'$ if  they only differ in the nonce names (see \cite{kanovich13ic} for more discussion on this).

Although the reachability  problem is stated as a
decision problem, we follow \cite{kanovich11jar} and are able to prove more than just  existence of a plan.
Namely,  by using the notion of ``scheduling'' a plan
 we are also able to generate a plan when there is a solution.
This is useful for the complexity of the plan generation,
since the number of actions in the plan
may be very large (see \cite{kanovich13ic} for more details on this).

An algorithm is said to \emph{schedule a plan} if it finds a plan if one exists, and 
on input $i$, if the plan contains at least $i$ actions,
then it outputs the $i^{th}$ action of the plan, otherwise
it outputs \emph{no}.

\vspace{2mm}
\begin{remark}
Notice that the feature of time constraints being attached to the rules increases expressivitiy of the model. Constraints may or may not be attached to a rule or to a configuration. With no constraints attached to  rules and configurations, we deal with the reachability problem that does not make use of the dimension of time to its potential.
In other words, adding constraints to rules and  configurations is not a restriction of the model. Quite the contrary,
using constraints, we are able to express time properties both for application of rules, and states of the system. In that way we are able to formalize time-sensitive actions and system configurations, and hence we can consider problems that involve explicit time requirements. 
\end{remark}

\vspace{2mm}

For the complexity results of the reachability problem we will consider actions that are \emph{balanced}, \ie,~actions that have the same number of facts in the pre-condition as in the post-condition. 
Balanced systems that contain only balanced actions have the special property that  all configurations in their plans have the same number of facts, given by the number of facts in the initial configuration.
This is because when applying a balanced action consumed facts from an enabling configuration are  replaced with the same number of created facts  in the resulting configuration.

Indeed, balanced systems can be conceived as systems with bounded memory~\cite{kanovich13ic}. More precisely, if we additionally impose a bound on the number of symbols that can be contained in a fact, such balanced systems involve configurations with a fixed number of facts of a bounded storage capacity, \ie,~they can only store a bounded number of symbols at a time.

\begin{remark}
Notice that any un-balanced rule can be made balanced by adding additional facts where needed, either in the pre-condition or in the post-conditions of the rule. For that purpose we use so called \emph{empty facts}, $E@T$,
that simply denote available memory slots.
Notice that  in such a way we do not get the equivalent system w.r.t. the reachability problem, as some rules may not be applied in the obtained balanced system unless there is a sufficient number of empty facts in the configuration, even in cases when the corresponding rule is applicable in the original system.
For example, the following unbalanced rule \ 
$$Time@ T, ~F_1@T' ~\lra ~Time@T, ~F_1@T', ~F_1@T, ~F_2@T $$ 
 is applicable to configuration \ $Time@0, F_1@0, E@0$, \ but its corresponding balanced rule  \ 
$$Time@ T, ~F_1@T', ~E@T, ~E@T ~\lra ~Time@T, ~F_1@T', ~F_1@T, ~F_2@T \ $$ 
 is not. Hence, the obtained balanced system may not have a solution to the given reachability problem, although the original system does.
However, there is no a priori bound on the number of facts in the initial configuration, \ie, we could add any number of empty facts to the given initial configuration.
\end{remark}

\vspace{2mm}

For our complexity results (Section~\ref{sec:complex})  we necessarily consider only systems with balanced actions,  since it has been shown in~\cite{kanovich09csf} 
that the reachability problem is undecidable if actions are allowed to be un-balanced. 
Multiset rewriting models considered in~\cite{kanovich09csf} were untimed, and the undecidability of the reachability problem for those models implies undecidability of the reachability problem for timed models considered in this paper. 

Additionally, we will impose a bound on the size of facts, that is, a bound on a total number of symbols contained in a fact.
Namely, the reachability problem is undecidable even for (un-timed) balanced systems when the size of facts is unbounded~\cite{cervesato99csfw,durgin04jcs}.

Furthermore, in~\cite{kanovich15mscs} we show that by  relaxing any
of the main conditions on instantaneous rules, 
same as conditions given in Equation (\ref{eq:instantaneous}), leads to the undecidability of the reachability
problem for multiset rewriting models with discrete time, and thus lead as well to the undecidability of the reachability problems considered in this paper.
For example,   we get undecidability  for systems with  time constraints that involve three or more time variables.

\subsection{Equivalence Between Configurations}

\vspace{2mm}
Extending the model to include dense time leads to additional challenges with respect to the complexity of the corresponding problems such as the reachability problem we address in this paper. Our solution proposed in \cite{kanovich15mscs} for the model with discrete time is not suitable for the dense time case. In particular it does not address the unboundedness caused by the Tick action which allows time to advance for any positive real number.
In order to tackle this source of unboundedness, we define an equivalence relation among configurations defined below.

\vspace{2mm}
 Many formal definitions and results in this paper mention $\Dmax$, an
upper bound on the numeric values of a reachability problem. This value is 
computed from the given problem: we set $\Dmax$ to be a
 natural number such that $\Dmax > n + 1$ for any number $n$ (both real or natural) appearing in the timestamps of the 
initial configuration, or the $D$s and $D_i$s in constraints or actions of the reachability problem. 

The following definition establishes the equivalence of 
configurations.

\vspace{2mm}
\begin{definition}
\label{def:equivalence}
 Given a reachability problem $\Tscr$,
 let $\Dmax$ be an upper bound on the numeric values
 appearing in $\Tscr$.
Let
\begin{equation}
\label{eq:two-configurations}
\begin{array}{l}
\Sscr = \{~Q_1@t_1, ~Q_2@t_2, \ldots, ~Q_n@t_n~\} \qquad \textrm{and} \qquad
\widetilde{\Sscr} = \{~\widetilde{Q}_1@\widetilde{t}_1, ~\widetilde{Q}_2@\widetilde{t}_2, \ldots,
~\widetilde{Q}_n@\widetilde{t}_n ~\}
\end{array} 
\end{equation}
be two configurations written in canonical way where the two sequences of
timestamps \ $t_1, \ldots, t_n$ \ and \ $\widetilde{t}_1, \ldots, \widetilde{t}_n$ \ are non-decreasing. 
 (For the case of equal timestamps, we sort the facts in alphabetical order, if necessary.) 
We say that  configurations  $\Sscr$
and $\widetilde{\Sscr}$ are \emph{equivalent configurations} if the following conditions hold:
\begin{itemize}
\item[(i)]
there is a bijection $\sigma$ that maps the set of all nonce names appearing in 
configuration $\Sscr$ to the set of all nonce names appearing in
configuration $\widetilde{\Sscr}$,
such that \ $Q_i \sigma = \widetilde{Q}_i$, \ for each $i \in \{1,\dots, n\}$;  and
\item[(ii)]  configurations ${\Sscr}$ and $\widetilde{\Sscr}$ satisfy the same constraints, that is:
$$
\begin{array}{l}
t_i > t_j \pm D \quad  \textrm{ iff } \quad \widetilde{t_i} > \widetilde{t_j} \pm D \quad  \ \text{and} \\
 t_i = t_j \pm D \quad  \textrm{ iff } \quad \widetilde{t_i} = \widetilde{t_j} \pm D ,
\end{array}
$$
for all \ $1\leq i \leq n$, $1\leq j \leq n$  \  and \ $D \leq \Dmax$.
\end{itemize}
When    $\Sscr$ and $\widetilde{\Sscr}$ are equivalent we  write\  $\Sscr \sim_{\Dmax} \widetilde{\Sscr}$, or simply  $\Sscr \sim \widetilde{\Sscr}$.
\end{definition}

\vspace{2mm}
Notice that equivalent configurations contain the same (untimed) facts up to renaming of fresh values.
Notice as well that by increasing or decreasing all the timestamps of a configuration $\Sscr$ by the same value $\Delta$
the obtained configuration  $\widetilde{\Sscr}$ satisfies the same constraints as  ${\Sscr}$. That is because time constraints are relative, \ie, involve exactly two time variables, and hence for configurations ${\Sscr}$ and $\widetilde{\Sscr}$ given in  (\ref{eq:two-configurations}) the following holds:
$$ 
\begin{array}{l}
t_i > t_j \pm D \ \ \textrm{ iff } \ \ (\widetilde{t_i}+\Delta) > (\widetilde{t_j}+\Delta) \pm D \quad  \ \text{and} \\
t_i = t_j \pm D \ \ \textrm{ iff } \ \ (\widetilde{t_i}+\Delta) = (\widetilde{t_j}+\Delta) \pm D 
\ .
\end{array}
 $$

When compared, facts of equivalent configurations satisfy the same order. 
Although facts follow the same order they need not lay at the exact same points of the real line.

For example, with $\Dmax = 3$,  configurations:
\[
\begin{array}{l}
	\Sscr_1 = \{~Time@0.2, ~F_1(n_1)@1.5,~F_2(n_1,b)@2.3, ~F_3(b)@2.5 ~\} \quad \textrm{and}
	\\[5pt]
	 \Sscr_2 =\{~Time@1.1,~F_1(n_2)@2.7,~F_2(n_2,b)@3.2, ~F_3(b)@3.7~\}
\end{array}
\]
are equivalent.
However, above configurations are not equivalent to configuration:
\[
\Sscr_3=  \{~Time@0.2,~F_1(n_1)@1.5,~F_2(n_1,b)@2.3, ~F_3(b)@2.52~\} \ .
\]
since both  $\Sscr_1$ and $\Sscr_2$ satisfy the constraint \ $T_3 - T_1 = 1$, where $T_1$ is the timestamp of the fact $F_1$ and 
 $T_3$ is the timestamp of the fact $F_3$, 
while $\Sscr_3$ does not.

As shown by the above example, same order of facts by itself  is not enough for the equivalence of configurations because all constraints of the type $T_i = T_j \pm D $ \ need to be satisfied simultaneously by both configurations.
Therefore,   the  facts corresponding to $T_i$ and $T_j$ from constraints   $T_i = T_j \pm D $, when placed on the real line,  
 need to lay at the points with the same decimal parts. 
Moreover,  matching of corresponding facts at integer points needs to hold  when placing any of the corresponding pairs of facts at point 0 on the real line.

Moreover,  because of  constraints of the type $T_i > T_j \pm D $, when placing any pair of corresponding facts at 0, all remaining  pairs of   corresponding facts need to lay either at the same integer point or  inbetween same consecutive integers.

In Section \ref{sec:circle} we introduce another, more illustrative, representation of the above equivalence relation.

The following proposition captures our intuition that the notion of equivalence defined above 
is coarse enough to identify applicable actions and thus the reachability problem.

\begin{proposition}
\label{th:equiv_action}
Let $\Sscr$ and $\Sscr'$ be two equivalent configurations for a given reachability problem $\Tscr$ and the
 upper bound $\Dmax$.
Let an  action $r$ transform $\Sscr$ into $\Sscr_1$. Then there is an instance of the action $r$ such that \  $\Sscr'  \lra_r \Sscr_1'$ \ and  that configurations $\Sscr_1$ and  $\Sscr_1'$ are also equivalent.
\end{proposition}
\begin{proof}
Let $\Sscr$ and  $\Sscr'$  be two equivalent configurations, namely
\[
\begin{array}{lc}
\, \Sscr \ = \{~Time@t, \  Q_1@t_1, \, Q_2@t_2, \ldots, ~Q_n@t_n ~\} \qquad \textrm{and} \\
{\Sscr'} = \{~Time@t',  Q_1'@{t}_1', \,Q_2'@{t}_2', \ldots, ~Q_n'@{t}_n' ~\}\ . 
\end{array}		
\]
Assume that $\Sscr$ is transformed into $\Sscr_1$ by means of
 an action~$r$.
 By definition of equivalence between configurations, \mbox{Definition~\ref{def:equivalence}},  configuration~ $\Sscr'$ contains the same (untimed) facts as $\Sscr$ up to nonce renaming, \ie,~there is a bijection $\sigma$ such that \ $Q_i \sigma = Q_i'$,  for all $i=1,\dots,n$.
Also, per \mbox{Definition~\ref{def:equivalence}}, configuration $\Sscr'$  satisfies  all the time constraints corresponding to the action~$r$, if any. 
Hence the action~$r$ is indeed applicable to  configuration~$\Sscr'$ and  will transform 
 $\Sscr'$  into  some  $\Sscr_1'$, as depicted in the following diagram:
\[
 \begin{array}{cccc}
\Sscr & \to_{r}&   \Sscr_1\\[5pt]
\wr &     &  \\[5pt]
\Sscr' & \to_{r} \ & \ \Sscr_1'
\end{array}
\]
Since configurations $\Sscr$ and $\Sscr'$ may differ in the actual values of the timestamps attached to their facts, possibly different instances of the action $r$ may be applied to each of the configurations. We consider both cases for the action $r$, namely  time advance action \ie,~action of type (\ref{eq:time-advancement}) and  instantaneous actions \ie,~action of type (\ref{eq:instantaneous}), and we need to show that in both cases $\Sscr_1$  is equivalent to $\Sscr_1'$. 

\vspace{5pt}
Assume that $r$ is an instantaneous action. Recall that the action $r$ does not change the global time, so $t$ and $t'$  denote the global time both in $\Sscr$ and  $\Sscr_1$, and $\Sscr'$ and  $\Sscr_1'$, respectively. 
\\
Hence, in showing that $\Sscr_1$  is equivalent to $\Sscr_1'$, only constraints involving facts created by the action $r$ are interesting, \ie,~facts of the form
~$F_1@(t+D_1)$ \ and \ $F_2'@(t'+D_2)$, respectively. 
Notice that the fact $P@(t+D)$, crated by an instance of $r$, appears in $\Sscr_1$ \ iff \ the fact $P'@(t'+D)$, created by an instance of $r$, appears in $\Sscr_1'$, where ~$P'=P\sigma $ ~for some nonce renaming bijection $\sigma$.
\\
Hence, the relative time difference to the global time is exactly $D$, both for \ $P@(t+D)$ \ and for \ $P'@(t'+D)$.
This implies that any given time constraint attached to $r$ involving a created fact and the global time is satisfied in  $\Sscr_1$ iff it is satisfied in $\Sscr_1'$.\\
Similarly, time constraints involving two created facts, $P_1@(t+D_1)$, $P_2@(t+D_2)$ and  $P_1'@(t'+D_1)$, $P_2'@(t'+D_2)$, respectively, involve the same relative difference $D_1-D_2$. Hence, they are concurrently satisfied.
\\
Finally, a time constraint $c$ involving a created fact $P_1@(t+D_1)$ and a fact $Q_i@t_i$ that appears  in $\Sscr$,
and, respectively, facts  $P_1'@(t'+D_1)$ and  $Q_i'@t'_i$ in $\Sscr_1'$, 
 can be associated to a time constraint involving the global time.
Clearly,  for all $D$
$$
\begin{array}{rcl}
t_i > t \pm D \ & \textrm{ iff } & \ {t_i}' > t' \pm D \quad  \textrm{and} \\ 
t_i = t \pm D \ & \textrm{ iff } & \ {t_i}' = t' \pm D
\end{array}
$$
 is equivalent to 
$$
\begin{array}{rcl}
t_i > (t+D_1) \pm D \ & \textrm{ iff } & \ {t_i}' > (t' + D_1) \pm D \quad \textrm{and} \\
t_i = (t+D_1) \pm D \ & \textrm{ iff } & \ {t_i}' = (t' + D_1) \pm D \ .
\end{array}
$$
Above equivalence states that $\Sscr_1$ and $\Sscr'_1$ concurrently satisfy  constraint $c$ iff they concurrently satisfy some constraint 
 involving a created fact and the global time.
Since the later has already  been shown, we can conclude that both 
 $\Sscr$ and $\Sscr'$ satisfy  constraint $c$. 

Having considered all the relevant types of constraints, we can conclude that, since $\Sscr \sim \Sscr'$,  it follows that ~${\Sscr_1} \sim \Sscr_1'$. 

\vspace{5pt}
Now, assume  $\Sscr$  is transformed into $\Sscr_1$  by means of a time advancing action~$r$:
$$
\begin{array}{lc}
\Sscr_1  = \, \{~ Time@(t+\varepsilon), \  Q_1@t_1, \, Q_2@t_2, \ldots, ~Q_n@t_n  ~\} \ .
\end{array}
$$
Depending on the  actual value $\varepsilon$ in  $r$ we will  find the value $\varepsilon'$ for the instance of time advancement action that will transform $\Sscr'$ into $\Sscr_1'$ in such a way that equivalence  $\Sscr_1 \sim \Sscr_1'$  holds.
Recall that time advancement action is  applicable to any configuration. 

With time advancement only the timestamp denoting the global time is increased while the rest of the configuration remains unchanged. Therefore, only time constraints involving the global time are to be considered for equivalence of $\Sscr_1$ and  $\Sscr_1'$. 

Firstly, assume that the new global time \ $t+\varepsilon$ \ in $\Sscr_1$  is equal to some timestamp $t_j \pm B$, where $B$ is an integer. Then we set \ 
$$\varepsilon'=t_j' \pm B - t' \ .
$$
 Then the new global time in $\Sscr'_1$ is \ $t' + \varepsilon' = t_j' \pm B $. Clearly, for any integer $D$ it holds that
$$
t_i \lesseqqgtr  (t+\varepsilon) \pm D = t_j \pm B \pm D\ \ \ \textrm{ iff } \ \ \ {t_i}' \lesseqqgtr (t' + \varepsilon) \pm D= t_j' \pm B \pm D
$$
since $t_i, t_i'$ and $t_j, t_j'$ are corresponding  timestamps from $\Sscr$ and $\Sscr'$ which are equivalent, \ie,~satisfy the same constraints.

Next, we consider the remaining case when the  decimal part of the new global time \ $t+ \varepsilon$ \ is different from the decimal part of any $t_i$s in $\Sscr_1$. More precisely, if we arrange the facts in $\Sscr_1$ according only to their decimal parts, then the decimal part of \ $ t+ \varepsilon$ \ either lays directly in between decimal parts of  $t_i$ and $t_j$, or it is greater than the decimal part of any $t_i$.
In order to get the configuration $\Sscr_1'$ that is equivalent to $\Sscr_1$, we need to achieve the same ordering of facts.
We therefore set \
$$\varepsilon' = \Int{ \varepsilon } + \delta, \ $$
where $\delta \in \langle 0,1 \rangle$ is any number such that 
\ 
$$
\begin{array}{rcl}
& \quad \dec{ t_i'}< \dec{t' + \varepsilon'} < \dec{ t_j'}  & \quad \text{in case when } \  \dec{ t_i}< \dec{t + \varepsilon} < \dec{ t_j}, \ 
\ \text{for some } \ i,j 
\\
   \text{ or } \quad & \quad  \dec{ t_i'}< \dec{t' + \varepsilon'}, & \quad   \text{in case that  } \ \ \ \dec{ t_i}< \dec{t + \varepsilon} , \ 
\ \forall  \  i=1, \dots, n \ .
\end{array}
$$
This way we obtain the same ordering of $\dec{t' + \varepsilon'}$ in $\Sscr_1'$ as for \ $\dec{t + \varepsilon}$ in $\Sscr_1$.\\
Here \
$\Int{x}$ \  is the integer part of $x$ and \ $\dec x$ \ is the decimal part of $x$. Advancing time in $\Sscr'$ for such an $\varepsilon'$  results in configuration  $\Sscr_1'$ that is equivalent to $\Sscr_1$.

Indeed, none of the time constraints involving global time and equality, such as constraint  \ \mbox{$t_i = t \pm D $} ~\ie,~  \mbox{$t_i' = t' \pm D $}, is satisfied since $D$ is an integer and the decimal parts of both global times 
$t$ and $t'$ are different from the decimal part of any other fact in the configuration. Therefore, both \ $t_i - t$~and~ $t_i' - t'$ \ are  not  integers.

For the  time constraints involving the global time and inequality, we consider the constraint of type ~$t_k > t \pm D_i $. 
The proof for the case of constraint of type ~$t_k < t \pm D_i $ \ is analogous. 
\\
From \ $\Sscr \sim \Sscr'$  \ for all integers \ $B\leq \Dmax$ \ we know that for all $k \in \{1,\dots,n\}$
\begin{equation}
\label{eq: constr-eq}
t_k > t \pm B \ \ \textrm{ iff } \ \ {t_k}' > t' \pm B \ .
\end{equation}
We need to  prove that
\begin{equation}
\label{eq: constr-eq-claim}
t_k > (t+ \varepsilon) \pm D \ \ \textrm{ iff } \ \ {t_k}' > (t' + \varepsilon') \pm D 
\end{equation}
 for all integers $D<\Dmax$. Notice that  ~$t+ \varepsilon \pm D = t + \dec{\varepsilon} + \Int{\varepsilon} \pm D$
\ involves possibly an integer $\Int{\varepsilon} + D$~ greater than $\Dmax$, for which (\ref{eq: constr-eq}) may not hold. Similarly, it is the case for integer
$\Int{\varepsilon'} + D$.

We will, therefore, assume that $\varepsilon <1$, \ie, ~ $\Int{\varepsilon} =0$. Equivalently, we could split an arbitrary $\varepsilon$ into a finite number of ~$\varepsilon_i$~ such that ~$\sum \varepsilon_i= \varepsilon$, and prove the result for $\varepsilon$ by induction.
\\
From  (\ref{eq: constr-eq}) and from 
$$
\Int{t} \leq \Int{t+ \varepsilon} \leq \Int{t} + 
1
$$
we can conclude that 
\begin{equation}
\label{eq: constr-eq-int}
\Int{t_k} \geq  \Int{t+ \varepsilon} \pm D= \Int{t+ \varepsilon \pm D} \ \ \ \textrm{ iff } \ \ \ \Int{t_k}' \geq \Int{t' + \varepsilon'} \pm D =  \Int{t'+ \varepsilon' \pm D } \ .
\end{equation}
Additionally,
\begin{equation}
\label{eq: constr-eq-dec}
 \dec{t_k} > \dec{t+\varepsilon} \ \ \textrm{ iff } \ \  \dec{t_k'}> \dec{t' + \varepsilon'}  
\end{equation}
 holds because $\varepsilon'$ is chosen so that the ordering of decimal parts of timestamps in  $\Sscr_1$ and $\Sscr_1'$ is the same. That is, either 
$$
 \dec {t_i} < \dec{t+\varepsilon}< \dec{ t_j}  \ \ \ \textrm{ and } \ \ \ \dec{ t_i'}< \dec{t' + \varepsilon'} < \dec{ t_j'}   
$$
for some \ $Q_i@t_i,Q_j@t_j$ \ that immediately precede and follow the fact \  $Time@(t+\varepsilon)$ \ when sorting the facts in $\Sscr$ only by the decimal parts of their timestamps, or it is the case that
$$
 \dec {t_i} < \dec{t+\varepsilon}  \ \ \ \textrm{ and } \ \ \ \dec{ t_i'}< \dec{t' + \varepsilon'} \ , \quad  \forall i \in \{1,\dots,n \}  \ .
$$
From  (\ref{eq: constr-eq-int}) and  (\ref{eq: constr-eq-dec}) we obtain the claim  (\ref{eq: constr-eq-claim}).
\ \end{proof}

Following the above Proposition, we now relate  the equivalence of configurations given by Definition~\ref{def:equivalence} and the reachability problem. 

\vspace{2mm}
\begin{theorem}
\label{th:equiv_reach}
 Let $\Sscr_I$ and $\Sscr_I'$ be two equivalent initial configurations, $\Sscr_G$ be a goal 
 and $\Rscr$ a set of actions. Let $\Dmax$ be an 
 upper bound on the numbers in $\Rscr$, $\Sscr_I$, $\Sscr_I'$ and $\Sscr_G$. Then the reachability problem with $\Sscr_I, \Sscr_G$ and $\Rscr$ is solvable if and only if 
 the reachability problem with $\Sscr_I', \Sscr_G$ and $\Rscr$ is solvable.
\end{theorem}
\begin{proof} \blue
Suppose  the reachability problem with $\Sscr_I, \Sscr_G$ has a solution.
We prove the existence of the solution to the reachability problem with $\Sscr_I', \Sscr_G$  by induction on the length of the given plan from $\Sscr_I$ to $\Sscr_G$. 
From  Proposition \ref{th:equiv_action} it follows  that  the plan from $\Sscr_I $ to $\Sscr_G$ and the plan from $\Sscr_I'$ to $\Sscr_G'$ contain exactly the same actions, possibly using different instances.
It also follows that  $\Sscr_i \sim \Sscr_i'$  for each configuration $\Sscr_i$ along the plan. This is depicted in the following diagram:
\[
\begin{array}{cccccccc}
\Sscr_I & \to_{r_1} \dots \to_{r_{i-1}} &  \Sscr_{i-1} & \to_{r_i} \  & \ \Sscr_{i}  & \to_{r_{i+1}}  \dots \to_{r_{n}} & \Sscr_G
\\
\wr & &   \wr &     &  \wr &     &  \wr
\\
\Sscr_I' &\to_{r_1} \dots \to_{r_{i-1}} & \Sscr_{i-1}' & \to_{r_i}&   \Sscr_i' & \to_{r_{i+1}} \dots \to_{r_{n}} & \Sscr_G'
\end{array}
\]
Since $\Sscr_G \sim \Sscr_G'$, these configurations both satisfy the same set of constraints, therefore $\Sscr_G$ is a goal configuration iff $\Sscr_G'$ is a goal configuration. Hence,  the reachability problem with $\Sscr_I, \Sscr_G$ and $\Rscr$ has a solution  if and only if 
 the reachability problem with $\Sscr_I', \Sscr_G$ and $\Rscr$ has a solution. 
\bluend
\end{proof}

\vspace{-3mm}
\subsection{Distance Bounding Protocol Formalization}
\label{sec: formal}

To demonstrate how our model can capture the attack in-between-ticks, consider the following protocol, called DB.
 This protocol captures the time challenge of distance bounding protocols.\footnote{Another specification that includes an intruder model, keys, and the specification of the attack described in \cite{basin11iss} can be found in our workshop paper~\cite{kanovich14fccfcs}.}  Verifier should allow the access to his resources only if the measured round trip time of messages in the distance-bounding phase of the protocol does not exceed the given bounding time $R$. We assume that the verifier and the prover have already exchanged nonces $n_P$ and $n_V$:
\[
 \begin{array}{l}
  V \lra P : n_P      \qquad  \text{at time \ } t_0 \\
  P \lra V : n_V  \qquad \text{at time \ } t_1\\
 V \lra P : OK(P)  \quad  \ \, %
\ \text{iff} \ t_1-t_0 \leq R
 \end{array}
\]

\paragraph{Encoding of verifier's clock}
The fact $Clock_V@T$ denotes the  local clock of the verifier \ie,~the discrete time clock that verifier uses to measure the response time in the distance bounding phase of the protocol. 

We encode ticking of verifier's clock  in \emph{discrete units of time}.
Action (\ref{eq:time-verifier}) represents the ticking of verifier's clock:
\begin{equation}
\label{eq:time-verifier}
 Time@T ,\  Clock_V@T_1 \ | \ \{~T=T_1+1~\} \ \lra Time@T,\  Clock_V@T 
\end{equation}
Notice that if this action is not executed and $T$ advances too much, \ie, $T > T_1$, it means that the verifier clock stopped as it no longer advances.

\paragraph{Network}
Let $D(X,Y)=D(Y,X)$ be the integer representing the minimum time needed for a message to reach $Y$ from $X$. We also assume that participants do not move. Rule (\ref{eq-network-XY-not consumed})
models network transmission from $X$ to some $Y$:
\begin{equation}
\begin{array}{@{}l@{}}%
 Time@T,\, \mathcal{N}_X^S~(m)@T_1, \, E@T_2 \, \mid \, \{~T\geq T_1 + D(X,Y)~\}  ~\lra~ 
 Time@T,\, \mathcal{N}_X^S~(m)@T_1, \, \mathcal{N}_Y^R~(m)@T
\end{array}
           \label{eq-network-XY-not consumed}
\end{equation}%
Facts ${\cal N}_X^S(m) $ and ${\cal N}_X^R(m) $ specify that the participant $X$ 
has sent and may receive the message $m$, respectively.  
Once $X$ has sent the message $m$, that message can only be received by $Y$ once it traveled from $X$ to $Y$. The fact $E$ is an empty fact which can be interpreted as a slot of resource. 
This is a technical device used to turn a theory balanced. 
See~\cite{kanovich13esorics} for more details.

\begin{remark}
Notice that in the rule (\ref{eq-network-XY-not consumed}) the fact ~$\mathcal{N}_X^S(m)@T_1$  is not consumed. 
This models the transmission media, such as radio frequency, where messages are not consumed by recipients. In such media, and as modelled by the rule (\ref{eq-network-XY-not consumed}), 
it is possible for multiple participants to receive the same message $m$. Alternatively, as in the classical (wire) network communication, the messages are removed from the network. In our formal model we are able to represent such transmission media as well, \eg, using the following rule
$$
\begin{array}{@{}l@{}}%
 Time@T,\, \mathcal{N}_X^S~(m)@T_1  \, \mid \, \{~T\geq T_1 + D(X,Y)~\} ~\lra~ 
 Time@T,\,  \mathcal{N}_Y^R~(m)@T \ .
\end{array}
        $$
\end{remark}

\paragraph{Measuring the round trip time of messages}
A protocol run creates facts denoting times when messages of the distance bounding phase are sent and received by the verifier.
Predicates  $\mathit{Start}$ and $\mathit{Stop}$  denote the actual (real) time of these events so that the round trip time of messages is
$T_2 - T_1$  for timestamps $T_1, T_2$ in
 $\mathit{Start}(m)@T_1$,  $\mathit{Stop}(m)@T_2$.
On the other hand  predicates $\mathit{Start}_V$ and $\mathit{Stop}_V$  model the verifier's view of time: $ T_2 - T_1$, for $T_1, T_2$ in
$\mathit{Start}_V(m)@T_1$, $\mathit{Stop}_V(m)@T_2
$.

\begin{figure*}[t]
\begin{equation*}
\begin{array}{ll@{\qquad\qquad}ll}
& Time@T, ~V_0~(P,N_P,N_V)@T_1, ~E@T_2, ~E@T_3 \lra \\
& \qquad \qquad \qquad \qquad
         Time@T, ~V_1~(\textrm{pending},P,N_P,N_V)@T, ~\mathcal{N}_V^S~(N_P)@T, ~\mathit{Start}~(P,N_P,N_V)@T\\[3pt]
      
& Time@T, ~V_1~(\textrm{pending},P,N_P,N_V)@T_1, ~Clock_V@T, ~P@T_2 \ ~  |~ \ ~\{ ~T \geq T_1~\} ~ \lra \\
& \qquad \qquad \qquad \qquad
 Time@T, ~V_1~(start,P,N_P,N_V)@T, ~Clock_V@T, ~\mathit{Start}_V~(P,N_P,N_V)@T\\[3pt]

& Time@T, ~P_0~(V,N_V,N_P)@T_1, ~\mathcal{N}_P^R~(N_P)@T_2 ~ \ | ~  \ \{~ T  \geq T_2 ~\} ~ \lra \\
&  \qquad \qquad \qquad \qquad Time@T, ~P_1~(V,N_V,N_P)@T, ~\mathcal{N}_P^S~(N_V)@T \\[2pt]

&Time@T, ~V_1~(\textrm{start},P,N_P,N_V)@T_1,~\mathcal{N}_V^R~(N_V)@T_2 \lra \\  
& \qquad \qquad \qquad \qquad
Time@T, ~V_2~(\textrm{pending},P,N_P,N_V)@T, ~\mathit{Stop}~(P,N_P,N_V)@T\\[3pt]

& Time@T, ~V_2~(\textrm{pending},P,N_P,N_V)@T_1, ~Clock_V@T, ~E@T_2 \ \ | \ \ \{~T \geq T_1 ~\} ~ \lra \\
& \qquad \qquad \qquad \qquad
Time@T, ~V_2~(\textrm{stop},P,N_P,N_V)@T, ~Clock_V@T, ~\mathit{Stop}_V~(P,N_P,N_V)@T\\[3pt]

& Time@T, ~\mathit{Start}_V~(P,N_P,N_V)@T_1, ~\mathit{Stop}_V~(P,N_P,N_V)@T_2, ~V_2~(\textrm{stop},P,N_P,N_V)@T_3 \ \ |\\
&  \qquad \qquad \qquad \qquad \ \ \{~ T_2-T_1 \leq R, ~ T \geq T_3 ~\} ~ \lra ~ 
Time@T, ~V_3~(P)@T, ~\mathcal{N}_V^S~(Ok(P))@T, ~E@T       
\end{array}
\end{equation*}
\caption{Protocol Rules for DB protocol}
\label{fig:DB}
 \vspace{2mm}
\end{figure*}

\paragraph{Protocol Theory} Our example protocol $DB$ is formalized in Figure \ref{fig:DB}.
The first rule specifies that the verifier has sent a nonce and still needs to mark the time, specified by the fact ~$V_1(\textrm{pending},P,N_P,N_V)@T$. The second rule specifies verifier's instruction of remembering the current time. The third rule specifies prover's response to the verifier's challenge. The fourth and fifth rules are similar to the first two, specifying when verifier actually receives  prover's response and when he executes the instruction to remember the time. Finally, the sixth rule specifies that the verifier grants access to the prover if he believes that the distance to the prover is under the given bound.

\vspace{2mm}
\paragraph{Attack In-Between-Ticks}

We now show how attack in-between-ticks is detected in our formalization.

The initial configuration contains facts $Time@0$, $Clock_V@0$ denoting that global time and time on verifier's discrete time are initially set to $0$.

Given the protocol specification in Figure \ref{fig:DB}, attack in-between-ticks is  represented with the following configuration:
$$
\{~ \mathit{Start}~(P,N_P,N_V)@T_1, ~\mathit{Stop}~(P,N_P,N_V)@T_2,
~{\cal N}_V^S~(Ok(P))@T_3 ~\} \
 \ | \ \ \{~ T_2-T_1 > R~\}
$$
It denotes that in the session involving nonces $N_P,N_V$ the verifier $V$ has allowed the access to prover $P$ although the distance requirement has been violated.

Notice that such an anomaly is really  possible in this specification. Consider the following example: verifiers actually sends the first message at time 1.7 while the prover  responds at time 4.9. 
Between moments 1.7 and 4.9, there would be 3 ticks on the verifier's clock.  The verifier would consider starting time of 2 and finishing time of 5, and confirm with  the time bound $R=3$. Actually,  the real round trip time is greater than the time bound, namely $4.9-1.7=3.2$.
Following facts would appear in the configuration: 
$$\mathit{Start}_V~(n)@2, ~\mathit{Stop}_V~(n)@5, ~\mathit{Start}~(n)@1.7, ~\mathit{Stop}~(n)@4.9 \ .$$
Since ~$5 - 2 = 3$~ the last rule from Figure \ref{fig:DB}, the accepting rule, would apply resulting in the configuration containing the facts:
$$
 \ \mathit{Start}~(p,n_P,n_V)@1.7, \ \mathit{Stop}~(p,n_P,n_V)@4.9, \ {\cal N}_V^S~(Ok(p))@5 \ .
$$
Since ~$4.9-1.7=3.2$ is greater than ~$R=3$, this configuration constitutes an attack.

\medskip
Finally, notice as well that in our  formalization  in Figure \ref{fig:DB}, prover immediately responds to the received challenge message, see the third rule in  Figure \ref{fig:DB}.
In reality there is some non-zero time of processing of messages.
This could cause additional discrepancy between the actual and the measured round trip time. 
Namely, once the prover receives the challenge message he needs to check whether the nonce received is the agreed value, previously exchanged with the verifier. Then, he needs to compose the response message for sending, using the stored nonce value. This process would take some additional time which may be calculated in the design of the distance bounding protocol itself, \ie,~in the established bounding time of the protocol, $R$. The difference between expected and actual processing time can  be the source of further inaccuracy.

\long\def\comment#1{}

\newcount \WIDTH  
\newcount \HEIGHT 
\newcount \Xcur   
\newcount \Ycur   
\newcount \HALF   
\newcount \DBL    
\newcount \QUA    
\newcount \TreeH
\newcount \TreeW

\def\mtrue{\mbox{{\bf I}}}
\def\une{\mbox{${\mathchoice{\rm 1\mskip-4mu l}{\rm 1\mskip-4mu l}%
               {\rm 1\mskip-4.5mu l}{\rm 1\mskip-5mu l}}$}}
\def\mtrue{\une}

\long\def\comment#1{}

\newcount \WIDTH  
\newcount \HEIGHT 
\newcount \Xcur   
\newcount \Ycur   
\newcount \HALF   
\newcount \DBL    
\newcount \QUA    
\newcount \TreeH
\newcount \TreeW

\def\mtrue{\mbox{{\bf I}}}
\def\une{\mbox{${\mathchoice{\rm 1\mskip-4mu l}{\rm 1\mskip-4mu l}%
               {\rm 1\mskip-4.5mu l}{\rm 1\mskip-5mu l}}$}}
\def\mtrue{\une}

\def\hbrace#1#2%
{\begin{picture}(0,0)\thicklines
 \put(0,0){\makebox(0,0)[c]{$\overbrace{\hspace*{#1 pt}}^{#2}$}}
\end{picture}}

\long\def\probOne#1#2#3#4#5#6#7#8%
{\setlength{\unitlength}{1pt}\HEIGHT=#1\multiply\HEIGHT by 2%
\advance\HEIGHT by 28%
\begin{picture}(0,\HEIGHT)%
  \put(0,\HEIGHT){%
\begin{picture}(0,0)\thicklines\linethickness{2pt}%
\Ycur=0
\Xcur=#1
\Ycur=14
 \put(-\Xcur,-\Ycur){{}{\line(-1,-2){#1}}}
 \put(-\Xcur,-\Ycur){{}{\line( 1,-2){#1}}}
 \put(-\Xcur,-\Ycur){\makebox(0,0)[bl]  {#3}}

\advance\Ycur by #1
\advance\Ycur by #1

\Xcur=#1\multiply\Xcur by 2%
 \put(\Xcur,-\Ycur){{}{\vector(1,0){#1}}}
 \put(0,-\Ycur){{}{\line(1,0){\Xcur}}}
 \put(-\Xcur,-\Ycur){{}{\line(-1,0){#1}}}
\linethickness{0.3pt}%
\Xcur=#1\multiply\Xcur by 2%
 \put(0,-\Ycur){\line(-1,0){\Xcur}}
 \put(0,-\Ycur){\line( 1,0){\Xcur}}
\thicklines
 \put(-\Xcur,-\Ycur){\makebox(0,0)[c]{$\bullet$}}
 \put(-#1,-\Ycur){\makebox(0,0)[c]{$\bullet$}}
 \put(0,-\Ycur){\makebox(0,0)[c]{$\bullet$}}
 \put( #1,-\Ycur){\makebox(0,0)[c]{$\bullet$}}
 \put( \Xcur,-\Ycur){\makebox(0,0)[c]{$\bullet$}}
\advance\Ycur by 5
 \put(-\Xcur,-\Ycur){\makebox(0,0)[t]{#4}}
 \put(-#1,-\Ycur){\makebox(0,0)[t]{#5}}
 \put(0,-\Ycur){\makebox(0,0)[t]{#6}}
 \put( #1,-\Ycur){\makebox(0,0)[t]{#7}}
 \put( \Xcur,-\Ycur){\makebox(0,0)[t]{#8}}

{}{%
\Xcur=#1\divide\Xcur by 4%
 \put(-\Xcur,-\Ycur){\makebox(0,0)[t]{$\boldmath{\ell}$}}
 \put(-\Xcur,-\Ycur){\makebox(0,0)[b]{$\raisebox{2pt}{\line(0,1){6}}$}}
}
\end{picture}%
}\end{picture}}



\long\def\probOneh#1#2#3#4#5#6#7#8%
{\setlength{\unitlength}{1pt}\HEIGHT=#1\multiply\HEIGHT by 2%
\advance\HEIGHT by 28%
\begin{picture}(0,\HEIGHT)%
  \put(0,\HEIGHT){%
\begin{picture}(0,0)\thicklines\linethickness{2pt}%
\Ycur=0
\Xcur=#1
\Ycur=14
 \put(-\Xcur,-\Ycur){{}{\line(-1,-2){#1}}}
 \put(-\Xcur,-\Ycur){{}{\line( 1,-2){#1}}}
 \put(-\Xcur,-\Ycur){\makebox(0,0)[bl]  {#3}}

\advance\Ycur by #1
\advance\Ycur by #1

\Xcur=#1\multiply\Xcur by 2%
 \put(\Xcur,-\Ycur){{}{\vector(1,0){#1}}}
 \put(0,-\Ycur){{}{\line(1,0){\Xcur}}}
 \put(-\Xcur,-\Ycur){{}{\line(-1,0){#1}}}
\linethickness{0.3pt}%
\Xcur=#1\multiply\Xcur by 2%
 \put(0,-\Ycur){\line(-1,0){\Xcur}}
 \put(0,-\Ycur){\line( 1,0){\Xcur}}
\thicklines
 \put(-\Xcur,-\Ycur){\makebox(0,0)[c]{$\bullet$}}
 \put(-#1,-\Ycur){\makebox(0,0)[c]{$\bullet$}}
 \put(0,-\Ycur){\makebox(0,0)[c]{$\bullet$}}
 \put( #1,-\Ycur){\makebox(0,0)[c]{$\bullet$}}
 \put( \Xcur,-\Ycur){\makebox(0,0)[c]{$\bullet$}}
\advance\Ycur by 5
 \put(-\Xcur,-\Ycur){\makebox(0,0)[t]{#4}}
 \put(-#1,-\Ycur){\makebox(0,0)[t]{#5}}
 \put(0,-\Ycur){\makebox(0,0)[t]{#6}}
 \put( #1,-\Ycur){\makebox(0,0)[t]{#7}}
 \put( \Xcur,-\Ycur){\makebox(0,0)[t]{#8}}

{}{%
 \put(-#1,-\Ycur){\makebox(0,0)[b]{\raisebox{2ex}{$\boldmath{R}$}}}
 \put(-#1,-\Ycur){\makebox(0,0)[b]{$\raisebox{2pt}{\line(0,1){6}}$}}
\Xcur=#1\divide\Xcur by 4%
 \put(-\Xcur,-\Ycur){%
     \makebox(0,0)[b]{\raisebox{2ex}{$\boldmath{\ell}$}}}
 \put(-\Xcur,-\Ycur){\makebox(0,0)[b]{$\raisebox{2pt}{\line(0,1){6}}$}}
\advance\Ycur by -27 
\Xcur=#1\multiply\Xcur by 5\divide\Xcur by 8%
\HALF=#1\multiply\HALF by 3\divide\HALF by 4%
 \put(-\Xcur,-\Ycur){\hbrace{\HALF}{h}}
}
\end{picture}%
}\end{picture}}

\long\def\prob#1#2#3#4#5#6#7#8%
{\setlength{\unitlength}{1pt}\HEIGHT=#1\multiply\HEIGHT by 2%
\advance\HEIGHT by 28%
\begin{picture}(0,\HEIGHT)%
  \put(0,\HEIGHT){%
\begin{picture}(0,0)\thicklines\linethickness{2pt}%
\Ycur=0
\Xcur=#1
\Ycur=14
 \put( \Xcur,-\Ycur){{}{\line(-1,-2){#1}}}
 \put( \Xcur,-\Ycur){{}{\line( 1,-2){#1}}}
 \put( \Xcur,-\Ycur){\makebox(0,0)[bl]  {#3}}

\advance\Ycur by #1
 \put(-\Xcur,-\Ycur){{}{\line(-1,-1){#1}}}
 \put(-\Xcur,-\Ycur){{}{\line( 1,-1){#1}}}
 \put(-\Xcur,-\Ycur){\makebox(0,0)[bl]  {#2}}
\advance\Ycur by #1
\Xcur=#1\multiply\Xcur by 2%
 \put(-\Xcur,-\Ycur){{}{\line(-1,0){#1}}}
 \put( \Xcur,-\Ycur){{}{\vector(1,0){#1}}}
\linethickness{0.3pt}%
\Xcur=#1\multiply\Xcur by 2%
 \put(0,-\Ycur){\line(-1,0){\Xcur}}
 \put(0,-\Ycur){\line( 1,0){\Xcur}}
\thicklines
 \put(-\Xcur,-\Ycur){\makebox(0,0)[c]{$\bullet$}}
 \put(-#1,-\Ycur){\makebox(0,0)[c]{$\bullet$}}
 \put(0,-\Ycur){\makebox(0,0)[c]{$\bullet$}}
 \put( #1,-\Ycur){\makebox(0,0)[c]{$\bullet$}}
 \put( \Xcur,-\Ycur){\makebox(0,0)[c]{$\bullet$}}
\advance\Ycur by 5
 \put(-\Xcur,-\Ycur){\makebox(0,0)[t]{#4}}
 \put(-#1,-\Ycur){\makebox(0,0)[t]{#5}}
 \put(0,-\Ycur){\makebox(0,0)[t]{#6}}
 \put( #1,-\Ycur){\makebox(0,0)[t]{#7}}
 \put( \Xcur,-\Ycur){\makebox(0,0)[t]{#8}}

{}{%
\Xcur=#1\multiply\Xcur by 5\divide\Xcur by 8%
 \put(\Xcur,-\Ycur){\makebox(0,0)[t]{$\boldmath{\ell}$}}
 \put(\Xcur,-\Ycur){\makebox(0,0)[b]{\raisebox{2pt}{\line(0,1){6}}}}
}
\end{picture}%
}\end{picture}}

\long\def\probh#1#2#3#4#5#6#7#8%
{\setlength{\unitlength}{1pt}\HEIGHT=#1\multiply\HEIGHT by 2%
\advance\HEIGHT by 28%
\begin{picture}(0,\HEIGHT)%
  \put(0,\HEIGHT){%
\begin{picture}(0,0)\thicklines\linethickness{2pt}%
\Ycur=0
\Xcur=#1
\Ycur=14
 \put( \Xcur,-\Ycur){{}{\line(-1,-2){#1}}}
 \put( \Xcur,-\Ycur){{}{\line( 1,-2){#1}}}
 \put( \Xcur,-\Ycur){\makebox(0,0)[bl]  {#3}}

\advance\Ycur by #1
 \put(-\Xcur,-\Ycur){{}{\line(-1,-1){#1}}}
 \put(-\Xcur,-\Ycur){{}{\line( 1,-1){#1}}}
 \put(-\Xcur,-\Ycur){\makebox(0,0)[bl]  {#2}}
\advance\Ycur by #1
\Xcur=#1\multiply\Xcur by 2%
 \put(-\Xcur,-\Ycur){{}{\line(-1,0){#1}}}
 \put( \Xcur,-\Ycur){{}{\vector(1,0){#1}}}
\linethickness{0.3pt}%
\Xcur=#1\multiply\Xcur by 2%
 \put(0,-\Ycur){\line(-1,0){\Xcur}}
 \put(0,-\Ycur){\line( 1,0){\Xcur}}
\thicklines
 \put(-\Xcur,-\Ycur){\makebox(0,0)[c]{$\bullet$}}
 \put(-#1,-\Ycur){\makebox(0,0)[c]{$\bullet$}}
 \put(0,-\Ycur){\makebox(0,0)[c]{$\bullet$}}
 \put( #1,-\Ycur){\makebox(0,0)[c]{$\bullet$}}
 \put( \Xcur,-\Ycur){\makebox(0,0)[c]{$\bullet$}}
\advance\Ycur by 5
 \put(-\Xcur,-\Ycur){\makebox(0,0)[t]{#4}}
 \put(-#1,-\Ycur){\makebox(0,0)[t]{#5}}
 \put(0,-\Ycur){\makebox(0,0)[t]{#6}}
 \put( #1,-\Ycur){\makebox(0,0)[t]{#7}}
 \put( \Xcur,-\Ycur){\makebox(0,0)[t]{#8}}

{}{%
\Xcur=#1\multiply\Xcur by 5\divide\Xcur by 8%
 \put(\Xcur,-\Ycur){%
     \makebox(0,0)[b]{\raisebox{2ex}{$\boldmath{\ell}$}}}
 \put(\Xcur,-\Ycur){\makebox(0,0)[b]{\raisebox{2pt}{\line(0,1){6}}}}

 \put(-#1,-\Ycur){\makebox(0,0)[b]{\raisebox{2ex}{$\boldmath{R}$}}}
 \put(-#1,-\Ycur){\makebox(0,0)[b]{$\raisebox{2pt}{\line(0,1){6}}$}}
\advance\Ycur by -27 
\Xcur=#1\multiply\Xcur by 3\divide\Xcur by 16%
\HALF=#1\multiply\HALF by 13\divide\HALF by 8%
 \put(-\Xcur,-\Ycur){\hbrace{\HALF}{h}}
}
\end{picture}%
}\end{picture}}

\newcommand{\CLOCK}[9]{
\begin{picture}(0,0)\thicklines%
 \put(0,0){\makebox(0,0)[c]{$\bullet$}}%
 \DBL=#1\multiply\DBL by 7\divide\DBL by 10
 \HALF=#1\multiply\HALF by 2\divide\HALF by 5
 $\bezier{#1}(0,#1)       (\HALF,#1)  (\DBL,\DBL)$%
 $\bezier{#1}(\DBL,\DBL)  (#1,\HALF)  (#1,0)$%
 $\bezier{#1}(#1,0)       (#1,-\HALF) (\DBL,-\DBL)$%
 $\bezier{#1}(\DBL,-\DBL) (\HALF,-#1) ( 0,-#1)$%
 $\bezier{#1}(0,-#1)      (-\HALF,-#1)(-\DBL,-\DBL)$%
 $\bezier{#1}(-\DBL,-\DBL)(-#1,-\HALF)(-#1,0)$%
 $\bezier{#1}(-#1,0)      (-#1, \HALF)(-\DBL,\DBL)$%
 $\bezier{#1}(-\DBL,\DBL) (-\HALF, #1)(0,#1)$%
 \put(0,#1){\line(0,-1){4}}
 \put(#1,0){\line(-1,0){4}}%
 \put(0,-#1){\line(0,1){4}}%
 \put(-#1,0){\line(1,0){4}}%
%
%
 \HALF=#1 \advance\HALF by 5%
 \DBL=\HALF\multiply\DBL by 7\divide\DBL by 10
 \put(0,\HALF){\makebox(0,0)[b]{#2}}%
 \put(\DBL,\DBL){\makebox(0,0)[bl]{#3}}%
 \put(\HALF,0){\makebox(0,0)[l]{#4}}%
 \put(\DBL,-\DBL){\makebox(0,0)[tl]{#5}}%
 \put(0,-\HALF){\makebox(0,0)[t]{#6}}%
 \put(-\DBL,-\DBL){\makebox(0,0)[tr]{#7}}%
 \put(-\HALF,0){\makebox(0,0)[r]{#8}}%
 \put(-\DBL,\DBL){\makebox(0,0)[br]{#9}}%
\end{picture}}

\long\def\zone#1#2#3#4#5#6#7#8%
{\setlength{\unitlength}{1pt}\HEIGHT=#1\multiply\HEIGHT by 6%
\begin{picture}(0,\HEIGHT)%
  \put(0,\HEIGHT){%
\begin{picture}(0,0)\thicklines
\Ycur=0
\advance\Ycur by #1
\advance\Ycur by #1
\advance\Ycur by #1
\Xcur=#1\multiply\Xcur by 2%
 \put(0,-\Ycur){{}{%
    \CLOCK{\Xcur}{{}{``{}{$p_{error}(R,h)>0$}''}} {} {}{}
  {$R/2$} {}
{}{}{} {} {}{} {} {}
\put(0,0){\vector(0,-1){\Xcur}}
}}%
\Xcur=#1\multiply\Xcur by 3%
 \put(0,-\Ycur){{}{%
    \CLOCK{\Xcur}{}
{{}{``\mbox{$p_{error}(R,h)=0$}''}}{} {$\ell/2$} {} {}{}{}{}{}{}
\Xcur=#1\multiply\Xcur by 2\advance\Xcur by 6%
 \put(0,0){\vector(1,-1){\Xcur}}
}}%
 \put(0,-\Ycur){\makebox(0,0)[c]{\qquad \qquad Verifier}}
\advance\Ycur by -7 
\Xcur=#1\multiply\Xcur by 5\divide\Xcur by 2%
 \put(\Xcur,-\Ycur){\hbrace{#1}{1/2}}
\end{picture}%
}\end{picture}}


\section{Attack in Between Ticks - 
 A Full Probabilistic Analysis}
\label{s-prob-main}

\def\theenumi{{\rm\bf(\arabic{enumi})}}

In Section \ref{sec: attack-in-b-t} we presented the novel attack which we believe can be carried out on most distance bounding protocols. Our symbolic model described in Section~\ref{sec:msr}  formally demonstrates that the attack can, in principle, happen. In this section we investigate how likely it is for an attacker to carry out such an attack. We explicitly calculate the probability of such an erroneous    ``acceptance event''  happening based on a single challenge/response time measurement. Such a  ``measurement phase''  is an indispensable part of any distance bounding protocol and a basis for the ultimate verifier's decision of whether to grant the access or not. 

Our main result of this section (Theorem~\ref{t-h}) is to show that the attack is not so unlikely when the prover is beyond the established perimeter  up to a distance corresponding to a half-tick of the verifier, \ie, 3 meters in the scenario discussed in Section~\ref{sec:motivation}. The probability of distance measurement error is of 1/2. This probability, however, reduces to zero once the prover is further away.  That is, probability of error is  zero when the prover exceeds the perimeter  at the distance corresponding to one tick or more. 

\vspace{0,5em}
The attack in-between-ticks  is based on the discrepancy between the
 {\em observable\/} time interval\/\ \mbox{$t_1-t_0$},\ (between the moment $t_0$ when the time of sending the challenge message is recorded, and the moment $t_1$ when the time of receiving the response message is recorded)
 and the {\em actual\/} time interval\/\ \mbox{$s_1-s_0$}.
Here,  $s_0$, $s_1$, $t_0$ and $t_1$ respectively  denote the actual time when the challenge message is sent, and the actual time of receiving the response, the recorded time of sending the challenge message, and  the recorded  time of receiving the response message.

\vspace{0,5em}
For our probabilistic analysis we consider the {\em challenge-response protocol\/}
 in which the verifier reacts as quick as possible
 but within the time constraints that
 {\em only one operation can be executed in one clock cycle}.%
\footnote{%
 From the performance point of view, the difference between
 discrete time and dense time is that, in contrast with dense time,
 only {\em a fixed finite number of events may occur
 within a bounded time interval\/} in the case of discrete time.
 Without loss of generality,
 we allow here no more than one action
 be performed in one clock cycle.
} 

In a round of the challenge-response protocol the verifier performs the  following actions:

\begin{enumerate}
\item[(1)] 
 At a moment\/~$s_0$ within an initial clock cycle~$1$, say
  \  \mbox{$s_0=1 + X$},\ 
 verifier sends a challenge message\/~$m$.
 Here $X$ is a random variable
 distributed on the interval\/~\mbox{$[0,\frac{1}{2}]$}
 with its probability density\/~$f_{X}$.
\item[(2)] 
 {\em Just after that\/} - that is,
 at a moment\/~$t_0$ within the next clock cycle~$2$,
 say\ \mbox{$t_0=2 + Y$},\ 
 verifier records the fact that $m$ has been sent.
 Here $Y$ is a random variable
 distributed on the interval\/~\mbox{$[0,\frac{1}{2}]$}
 with its probability density\/~$f_{Y}$.
\item[(3)] 
 At a moment\/~$s_1$ within the corresponding
 clock cycle~\mbox{$\lfloor s_1\rfloor$},
 say\ \mbox{$s_1= s_0+\ell$},\
 verifier receives a response message\/~$m'$.
\item[(4)] 
 {\em Just after that\/}, that is,
 at a moment\/~$t_1$ within the next
 clock cycle~\mbox{$\lfloor s_1\rfloor+1$},
 say\ \mbox{$t_1 = (\lfloor s_1 \rfloor +1)+Z$},\ 
 verifier records the fact that $m'$ has been received.
 Here $Z$ is a random variable
 distributed on the interval\/~\mbox{$[0,\frac{1}{2}]$}
 with its probability density\/~$f_{Z}$.
\end{enumerate}

\vspace{1em}
 For the sake of perspicuity, we assume 
  $X$, $Y$, and\/ $Z$ be independent random variables
 uniformly distributed on the interval\/~\mbox{$[0,\frac{1}{2}]$}.

\noindent
 Thus, we are dealing with the model given by the system as illustrated in Figure~\ref{fig:in-btw-clock-attack3}:
{}{%
\begin{equation}%
 \begin{array}{lcllcl}
 s_0 & = & 1+X, &\qquad  s_1 & = & s_0+\ell,
\\[1ex]{}%
 t_0 & = & 2+Y, &\qquad  t_1 & = & \lfloor s_1\rfloor +1+Z.
   \end{array}%
          \label{eq-main}
\end{equation}%
}%

\begin{figure*}[t]
\centering
\vspace{2mm}
\includegraphics[width={ 0.6\textwidth}]{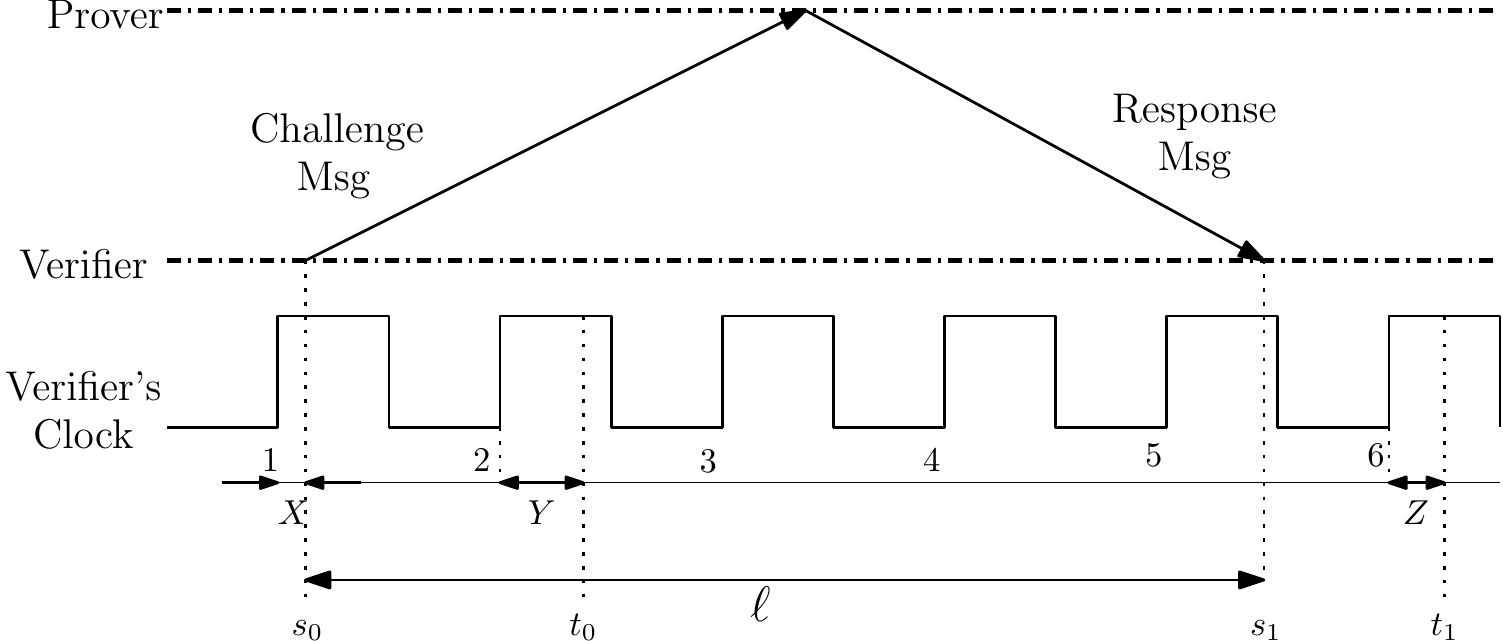}
\vspace{1mm}
\caption{In different ticks (Sequential Execution)}
\vspace{3mm}
\label{fig:in-btw-clock-attack3}
\end{figure*} 

\noindent
 The decision rule applied by the verifier is described bellow.

\vspace{1em}
\begin{definition}\label{d-yes}
 For a fixed time response bound, an integer\/~$R$,
 {}{verifier decides to grant the access to its resources}
 if and only if 
 the following holds
 for the\/ {\em measured\/} time interval\/\ \ \mbox{$t_1-t_0$}:
         $$\mbox{$ t_1-t_0 \leq R $}.$$%
\end{definition}

 Thus, the ``Yes'' decision taken by the verifier is erroneous if in reality
 the actual distance between verifier and prover,  \ \mbox{$s_1-s_0$}, \ turns out to be larger
 than\/~$R$, say by some positive value~$h$. 

\vspace{2em}
We now investigate the probability of such an event actually occurring. Firstly, we define the required probability.

\vspace{1em}
\begin{definition}\label{d-error}
 For a fixed time response bound, an integer\/~$R$,
 and an extra, a positive\/~$h$, we define
 the {\em probability of the erroneous decision  to grant the access \/},
\ \ \mbox{$p_{error}(R,h)$},
\begin{equation}%
 p_{error}(R,h) = Prob\,\{\ t_1-t_0\leq R\ /\ s_1-s_0 = R+h\ \}
      \label{eq-error-00}
\end{equation}%
 as the conditional probability of an ``acceptance event''
 of the form \
      $ t_1-t_0 \leq R,$%
\  given that
   \  $ \mbox{$s_1-s_0 = R+h $}.$%
\end{definition}

\vspace{1em}
 We calculate the probability of the {\em erroneous decision\/},
\ \mbox{$p_{error}(R,h)$}, and obtain the explicit values of probabilities as stated in the theorem bellow.
Recall that the protocol time response bound $R$ and the actual challenge-response time $\ell$ denote the respective round trip time of messages, and that in case of an erroneous ``Yes'' decision taken by the verifier, the value of ~$h$,  $~h= \ell - R $, is  positive.
In such a case prover that is outside of the perimeter specified by the value $\frac{R}{2}$ appears within the perimeter from the point of view of the verifier.

\vspace{01em}
The obtained  results are visualized in Figure~\ref{f-zone} which  shows how the conditional probability of erroneous decision, \ $ p_{error}(R,h)$, is classified w.r.t. time distances  between the verifier and the prover.
This probability is non-zero when the prover is in the zone close to the perimeter, on the outskirts up to $\frac{1}{2}$ tick time distance which is $1$ tick in the round trip time.

\vspace{1em}
\begin{theorem}\label{t-h}
 Let\/ $X$, $Y$ and $Z$ be independent random variables
 uniformly distributed on\/~\mbox{$[0,\frac{1}{2}]$}.
 Then, for a fixed time response bound, an integer\/~$R$,
 and an extra, a positive\/~$h$,
the probability of the erroneous decision  to grant the access, $p_{error}(R,h) $,  is given by
\begin{equation}%
  p_{error}(R,h) =
\left\{\begin{array}{ll}{}%
       \ \frac{1}{2}, &\ \ \mbox{if\/\ \ $0 < h \leq \frac{1}{2}$},
\\[1ex]%
    \    1-h,
 &\ \ \mbox{if\/\ \ $\frac{1}{2}<h<1$},
\\[1ex]%
  \      0, &\ \ \mbox{if\/\ \ $h \geq 1$}.
\end{array}\right.
               \label{eq-h-2}
\end{equation}%

\end{theorem}

\begin{figure*}[b]
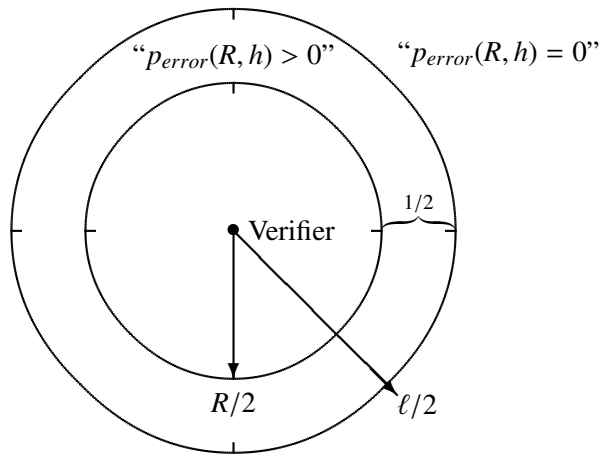

\vspace{1em}
\begin{center}%
 \zone{28} 
 {$\raisebox{1.1ex}{$2-4\widetilde{\ell}$}$}
 {$\raisebox{1.1ex}{$4\widetilde{\ell}$}$}
 {$\lfloor\ell\rfloor-\frac{1}{2}\ \ \ $} {$\lfloor\ell\rfloor$}
 {$\lfloor\ell\rfloor+\frac{1}{2}$} {$\lfloor\ell\rfloor+1$}
 {$\lfloor\ell\rfloor+\frac{3}{2}$}
\end{center}
\caption{Towards Theorem~\ref{t-h}.
\ Conditional probability of erroneous decision \ $ p_{error}(R,h)$
\ classified w.r.t.  round trip time distances, $R$ and $\ell$,  between verifier and  prover. The Verifier is at the center of the circles. The inner circle represents the distance bounding area. The outer circle is the actual area where Prover can be granted access with probability greater than zero.
}
\label{f-zone}
\end{figure*}

\vspace{1em}
Theorem~\ref{t-h}  provides the explicit probability of verifier making an erroneous decision to grant access to a  prover that is located outside of the perimeter specified by the protocol distance bound, $R$.

In particular, for  $\frac{1}{2}<h<1$, \ie, when the prover is between ``half a tick''  and  ``a single tick'' further away form the specified perimeter,  the probability of  the erroneous  decision  decreases with $h$.
In the case that the distance of the prover exceeds the upper bound by a ``tick''  or more, the probability of verifier making the erroneous  decision is zero. This is not surprising, given the nature of the attack in-between-ticks.

However, notice that, {}{{contrary to our expectations}},
 the probability of the {\em erroneous decision\/}
 {}{turns out to be  ~50\%~  
for any ~\mbox{$0 < h \leq \frac{1}{2}$}.   
That is,  the probability of error is rather high when the prover is close to the bound,  under ``half a tick'' distance form the perimeter. This extra time distance would amount to up to 3~meters in our example given earlier in the paper, which is not negligible.

\vspace{0,5em}
The remainder of this section contains the proof of Theorem~\ref{t-h}. 
A reader that is not that interested in this more technical part of this section, can  jump to Section~\ref{sec:Maude}.

\subsection{Proof of ~Theorem~\ref{t-h}}

\vspace{0,5em}
In order to prove  Theorem~\ref{t-h} we introduce some auxiliary machinery.

\vspace{0,5em}
\begin{definition}\label{d-F-ell}
 To investigate \mbox{$p_{error}(R,h)$},
 we introduce the following
 distribution function \mbox{$F_{\ell}(x)$}
\begin{equation}%
 F_{\ell}(x) = Prob\,\{\ t_1-t_0\leq x\ /\ s_1-s_0 =\ell\ \}
      \label{eq-F-ell}
\end{equation}%
 defined as the conditional probability of the event
      \ $ t_1-t_0 \leq x,$ \ 
 given  the actual time interval  \ \     $ \mbox{$s_1-s_0 = \ell$} .$%
\end{definition}

\vspace{0,5em}
 In~Figures~\ref{f-prob-1}~and\/~\ref{f-prob-2} we illustrate the two cases of
 the graph of the conditional probability density,
 the derivative\ \mbox{$F_{\ell}'(x)$},\ 
for the distribution function \ \mbox{$F_{\ell}(x)$}.

 Notice that, with\  \ \mbox{$\ell = R+h$},\ \ we have:\ \ \mbox{%
 $p_{error}(R,h) = F_{\ell}(R) = \int_{-\infty}^{R}F_{\ell}'(x)\,dx$}.  

 \vspace{0,5em}

\begin{figure*}
{
\begin{center}%
 \probOne{48}
 {\raisebox{1.1ex}{$2-4\widetilde{\ell}$}}
 {\raisebox{1.1ex}{$2$}}
 {$\lfloor\ell\rfloor-\frac{1}{2}\ \ \ $} {$\lfloor\ell\rfloor$}
 {\hspace*{5ex}$\lfloor\ell\rfloor\!+\!\frac{1}{2}$}
 {\hspace*{3ex}$\lfloor\ell\rfloor+1$}
 {$\lfloor\ell\rfloor+\frac{3}{2}$}
\end{center}
}
\caption{The graph of the probability
 density, {$F_{\ell}'(x){}$},
 for the distribution\ 
 \mbox{$F_{\ell}(x){}
$}
-
The single-humped (``Dromedary camel'') case:
\
\mbox{$\widetilde{\ell} = \ell-\lfloor\ell\rfloor<\frac{1}{2}$}.
}
\label{f-prob-1}

{
\begin{center}%
 \prob{48}
 {\raisebox{1.1ex}{$4-4\widetilde{\ell}$}}
 {\raisebox{1.1ex}{$4\widetilde{\ell}-2$}}
 {$\lfloor\ell\rfloor-\frac{1}{2}\ \ \ $}
 {$\lfloor\ell\rfloor$}
 {$\lfloor\ell\rfloor+\frac{1}{2}$}
 {\hspace*{3ex}$\lfloor\ell\rfloor+1$}
 {$\lfloor\ell\rfloor+\frac{3}{2}$}
\end{center}
}
\caption{
 The graph of the  probability density,
 \mbox{$F_{\ell}'(x)$},
 for the distribution function \ \mbox{$F_{\ell}(x)$}\ - 
The 2-humped (``Bactrian camel'')
 case of bimodal distribution:
\quad
\mbox{$\widetilde{\ell} = \ell-\lfloor\ell\rfloor>\frac{1}{2}$}.
}
\vspace{1em}
\label{f-prob-2}
\end{figure*}

\vspace{1em}
The following  lemmas provide 
 an explicit expression for
 the distribution function \mbox{$F_{\ell}(x)$}
 and its density \mbox{$F_{\ell}'(x)$}.
Let  here, and henceforth, \ \/\ $\widetilde{\ell}$\ \
 denote the decimal 
part of\/~$\ell$:\quad
 \mbox{$\widetilde{\ell}= \ell-\lfloor\ell\rfloor$}.

\vspace{1em}
\begin{lemma}\label{l-t1-t0}   
 In the model\/~(\ref{eq-main})
 we are dealing with the observable period of time
 \ \mbox{$t_1-t_0$}\ \ that is calculated as: 
\begin{equation}%
   t_1-t_0
  = \lfloor X+\widetilde{\ell}\rfloor+\lfloor\ell\rfloor+Z-Y
  =
\left\{\begin{array}{ll}{}%
   \lfloor\ell\rfloor+Z-Y, &
 \mbox{\ \ if\/\ \ $\widetilde{\ell}<\frac{1}{2}$},
\\[2ex]
   \lfloor\ell\rfloor+Z-Y, &
 \mbox{\ \ if\/\ \ $\widetilde{\ell}\geq\frac{1}{2}$\ \
      but\ \ $X+\widetilde{\ell} < 1$},
\\[2ex]
   1+\lfloor\ell\rfloor+Z-Y, &
 \mbox{\ \ if\/\ \ $\widetilde{\ell}\geq\frac{1}{2}$\ \
      and\ \ $X+\widetilde{\ell} \geq 1$}.
\end{array}\right.
         \label{eq-t1-t0}
\end{equation}
\end{lemma}
\begin{proof}
 By simple calculation,
$$ t_1-t_0
  = \lfloor 1+ X+\ell\rfloor+1+Z-(2+Y)
  = \lfloor X+\widetilde{\ell}\rfloor+\lfloor\ell\rfloor+Z-Y 
$$
\end{proof}

\vspace{0,5em}

\begin{lemma}\label{l-F-ell}\ 
 Let\/ $X$, $Y$ and $Z$ be independent random variables  uniformly distributed on\/~\mbox{$[0,\frac{1}{2}]$}. Then for the distribution function $F_{\ell}(x)$ given by
~$F_{\ell}(x) = Prob\,\{\ t_1-t_0\leq x\ /\ s_1-s_0=\ell\ \}  $, \ the following holds:

\begin{enumerate}
\item[(i)]
 In the case of\/\ \
 \mbox{$\widetilde{\ell}=\ell-\lfloor\ell\rfloor<\frac{1}{2}$},\ \
\begin{equation}%
 F_{\ell}(x) 
  = Prob\,\{\ Z-Y \leq x - \lfloor\ell\rfloor\ \}
                   \label{eq-F-ell-1}
\end{equation}%

\item[(ii)]
 In the case of\/ \
 \mbox{$\widetilde{\ell}=\ell-\lfloor\ell\rfloor\geq\frac{1}{2}$},
\begin{equation}%
  F_{\ell}(x) = (2-2\widetilde{\ell})
   \cdot Prob\,\{\ Z-Y \leq x-\lfloor\ell\rfloor\ \}\ +\
              (2\widetilde{\ell}-1)
   \cdot Prob\,\{\ Z-Y \leq x-\lfloor\ell\rfloor-1\ \}
                   \label{eq-F-ell-2-X}
\end{equation}%
\end{enumerate}
\end{lemma}

\begin{proof}
 Given the condition \ \mbox{$s_1-s_0 = \ell $},\
 we deal the following two cases.
\begin{enumerate}
\item[(i)]
 In the case of\/\ \
 \mbox{$\widetilde{\ell}=\ell-\lfloor\ell\rfloor<\frac{1}{2}$},\ \
 by~Lemma~\ref{l-t1-t0},
     \ $t_1-t_0 = \lfloor\ell\rfloor+Z-Y,$ \ 
 and, respectively,
$$
 F_{\ell}(x) = Prob\,\{\ \lfloor\ell\rfloor+Z-Y \leq x\ \}
  = Prob\,\{\ Z-Y \leq x - \lfloor\ell\rfloor\ \}
$$ 

 In~Figure~\ref{f-prob-1} we draw
 the graph of the conditional probability density,
\ \mbox{$F_{\ell}'(x)$},\ 
for the distribution function \ \mbox{$F_{\ell}(x)$} \
 in the case of the uniformly distributed $Z$ and\/~$Y$.
 The height of the triangle there is~$2$.

\vspace{0,5em}

\item[(ii)]
 In the case of\/
 \mbox{$\widetilde{\ell}=\ell-\lfloor\ell\rfloor\geq\frac{1}{2}$},
 by~Lemma~\ref{l-t1-t0} we have:
$$
 F_{\ell}(x) = Prob\{X+\widetilde{\ell}<1\}
   \cdot Prob\,\{\ \lfloor\ell\rfloor+Z-Y \leq x\ \}\ +\
              Prob\{X+\widetilde{\ell}\geq 1\}
   \cdot Prob\,\{\ 1+\lfloor\ell\rfloor+Z-Y \leq x\ \}.
                   \label{eq-F-ell-2-0}
$$
 Notice that for the $X$
 uniformly distributed on\/~\mbox{$[0,\frac{1}{2}]$}:
$$
    Prob\{X+\widetilde{\ell}<1\}=
      Prob\{X<1-\widetilde{\ell}\}=2(1-\widetilde{\ell}),%
                    \label{eq-X+ell}
$$
  resulting in
$$
  F_{\ell}(x) = (2-2\widetilde{\ell})
   \cdot Prob\,\{\ Z-Y \leq x-\lfloor\ell\rfloor\ \}\ +\
              (2\widetilde{\ell}-1)
   \cdot Prob\,\{\ Z-Y \leq x-\lfloor\ell\rfloor-1\ \}
$$

 In~Figure~\ref{f-prob-2} we draw
 the graph of the conditional probability density,
 the derivative\ \mbox{$F_{\ell}'(x)$},\
for the distribution function \ \mbox{$F_{\ell}(x)$} \
 in the case of the uniformly distributed $X$, $Z$ and\/~$Y$.
 The height of the left triangle in~Figure~\ref{f-prob-2}
  is ~\mbox{$4-4\widetilde{\ell}$}, and the height of
  the right triangle is ~\mbox{$4\widetilde{\ell}-2$}.
\end{enumerate}
\end{proof}    

\begin{figure*}
{
\begin{center}%
 \probOneh{48}
 {\raisebox{1.1ex}{$2-4\widetilde{\ell}$}}
 {\raisebox{1.1ex}{$2$}}
 {$\lfloor\ell\rfloor-\frac{1}{2}\ \ \ $} {$\lfloor\ell\rfloor$}
 {\hspace*{5ex}$\lfloor\ell\rfloor\!+\!\frac{1}{2}$}
 {\hspace*{3ex}$\lfloor\ell\rfloor+1$}
 {$\lfloor\ell\rfloor+\frac{3}{2}$}
\end{center}
}
\caption{
Probability of the erroneous decision  to grant the access, $p_{error}(R,h) $ 
- The single-humped case.}
\label{f-prob-1-h}
\vspace{0,5em}
{
\begin{center}%
 \probh{48}
 {\raisebox{1.1ex}{$4-4\widetilde{\ell}$}}
 {\raisebox{1.1ex}{$4\widetilde{\ell}-2$}}
 {$\lfloor\ell\rfloor-\frac{1}{2}\ \ \ $}
 {$\lfloor\ell\rfloor$}
 {$\lfloor\ell\rfloor+\frac{1}{2}$}
 {\hspace*{3ex}$\lfloor\ell\rfloor+1$}
 {$\lfloor\ell\rfloor+\frac{3}{2}$}
\end{center}
}
\caption{
Probability of the erroneous decision  to grant the access, $p_{error}(R,h) $ 
-  The 2-humped  (``Bactrian camel'') case.
}
\vspace{1,5em}
\label{f-prob-2-h}
\end{figure*}

\vspace{1em}
We are now ready to prove the main theorem.

\begin{proof}[{\bf  Proof  of  Theorem~\ref{t-h}}]
\ \\
Recall that $X$, $Y$ and $Z$ are independent random variables
 uniformly distributed on\/~\mbox{$[0,\frac{1}{2}]$}.
\ \\
 Given an integer\/~$R$,
 here  \ \mbox{$\ell = R+h$}
\ and\/\ \
 \mbox{$h = \widetilde{\ell} = \ell-\lfloor\ell\rfloor$}.
\ 
\begin{enumerate}
\item[(i)]
 In the case of\/\ \ \mbox{$0 < h \leq \frac{1}{2}$},
 we have \ \ \mbox{$\lfloor\ell\rfloor = R$},\ \ 
 $$h=\ell-R=\ell-\lfloor\ell\rfloor=\widetilde{\ell}\leq\frac{1}{2},$$%
 and, by Lemma~\ref{l-F-ell} (see~Figure~\ref{f-prob-1-h})
  $$ p_{error}(R,h)=\int_{-\infty}^{\lfloor\ell\rfloor}F_{\ell}'(x)\,dx
   = \frac{1}{2} $$%
\item[(ii)]
 In the case of\/\ \ \mbox{$\frac{1}{2}<h<1$},
 we have \ \ \mbox{$\lfloor\ell\rfloor = R$},\ \ 
 $$h=\ell-R=\ell-\lfloor\ell\rfloor=\widetilde{\ell}>\frac{1}{2},$$%
 and, by Lemma~\ref{l-F-ell} (see~Figure~\ref{f-prob-2-h})
  $$ p_{error}(R,h)=\int_{-\infty}^{\lfloor\ell\rfloor}F_{\ell}'(x)\,dx
 = \frac{1}{2}\cdot Prob\{X+\widetilde{\ell}<1\}
 = \frac{1}{2}\cdot 2(1-\widetilde{\ell}) = 1-h $$
\item[(iii)]
 Lastly,
 in the case of\/\ \ \mbox{$h>1$},
 we have \ \ \mbox{$R \leq \lfloor\ell\rfloor-1$},\ \
 and (see~Figures~\ref{f-prob-1-h}~and~\ref{f-prob-2-h})
 $$ p_{error}(R,h)\leq
         \int_{-\infty}^{\lfloor\ell\rfloor-1}F_{\ell}'(x)\,dx=0.$$%
\end{enumerate}
 which completes the proof of  Theorem~\ref{t-h}.
\end{proof}

\vspace{2mm}
\section{Implementation in Maude}
\label{sec:Maude}

\begin{figure}
\begin{small}
	\begin{verbatim}
crl [Network]: 
  { S (Time @ T) (Ns(PV,M) @ T1) } => 
  { S (Time @ T2) (Nr(VP,M) @ T2) }
 if (T2 >= T1 + dvp and (T2 >= T)) = true /\ VP := recSend(PV) [nonexec] .	

crl [Tick]: 
  { S (Time @ T) (vTime @ T1) } => 
  { S (Time @ T2) (vTime @ T1) }
 if (T2 > T and (T2 < T1 + 1/1))  = true [nonexec] .

crl [Tick-Vclock]: 
  { S (Time @ T) (vTime @ T1) } => 
  { S (Time @ T1 + 1/1) (vTime @ T1 + 1/1) }
 if (T1 + 1/1 >= T)  = true [nonexec] .

crl [Real-V-Send]: 
  { S (Time @ T) (V0(P,NP,NV) @ T1)  } => 
  { S (Time @ T) (V1(pending,P,NP,NV) @ T)
    (Start(P,NP,NV) @ T) (Ns(p,NP) @ T ) }
 if (T1 + 1/1 > T and T >= T1 )  = true [nonexec] . 

crl [Real-P-Rcv]: 
  { S (Time @ T) (P0(V,NV,NP) @ T1) (Nr(p,NP) @ T2) } => 
  { S (Time @ T3) (P1(V,NV,NP) @ T3) (Ns(p,NV) @ T3)}
 if ((T2 >= T) and (T3 >= T2)) = true [nonexec] . 

crl [Real-V-Rec]: 
{ S (Time @ T) (V1(start,P,NP,NV) @ T1) (Nr(v,NV) @ T2 ) } => 
{ S (Time @ T2) (V2(pending,P,NP,NV) @ T2) (Stop(P,NP,NV) @ T2)  }
if (T2 >= T ) = true [nonexec] .

crl [Dis-V-Rec]: 
  { S (Time @ T) (vTime @ T1) (V2(pending,P,NP,NV) @ T2) } => 
  {S (Time @ toReal(toInteger(T2)) + 1/1) (vTime @ toReal(toInteger(T2)) + 1/1) 
   (V2(stop,P,NP,NV) @ toReal(toInteger(T2)) + 1/1) 
   (StopV(P,NP,NV) @ toReal(toInteger(T2)) + 1/1) }
 if ( (T2 >= T1)) = true [nonexec] .

crl [ok]: 
  { S (Time @ T) (StartV(P,NP,NV) @ T1) (StopV(P,NP,NV) @ T2) (V2(stop,P,NP,NV) @ T3)} => 
  {S (Time @ T) (Ok(P) @ T) }
if (T2 - T1 <= 2/1 * 3/1)  = true [nonexec] .
\end{verbatim}
\end{small}
\caption{Rewrite Rules Specification of a Distance Bounding Protocol in Maude}
\label{fig:maude-rules}
\end{figure}
We have formalized the scenario of the attack in-between-ticks in an extension with SMT-solver of the rewriting logic tool Maude. The tool was able to automatically find this attack. We considered a scenario with two players, a verifier (\texttt{v}) and a prover (\texttt{p}) such that messages take \texttt{dvp} time units to navigate from one another. The function \texttt{recSend(PV)} returns \texttt{v} if \texttt{PV} is \texttt{p} and vice-versa.

Although the specification detailed in Section~\ref{sec: formal} could be specified in Maude, it would be impractical to use it to search for an attack as the state space is infinite. Moreover, using the machinery of circle-configurations described in Section~\ref{sec:circle}, although resulting in a finite search space, is still intractable. 

Instead, we show in this section that by using constrained variables and relying on SMT solvers allows us to verify distance bounding protocols in practice. During verification (search), the variables in the time constraints are not instantiated, but are accumulated. The SMT solver is, then, used to determine whether a set of accumulated constraints (in a branch of search) is consistent. If it is consistent, then search may proceed; otherwise, Maude backtracks and continues with the verification.

\vspace{2mm}
The modifications to the theory in Section~\ref{sec: formal} are minor simply those involving the constrained variables.
The Maude rewrite rules are depicted in Figure~\ref{fig:maude-rules}. To illustrate the use of SMT, consider the \texttt{Network} rule. It specifies that a message sent $\texttt{Ns(P,M)}$ at time \texttt{T1} can be received at any time \texttt{T2} such that \texttt{T2 >= T1 + dvp} and \texttt{T2 >= T}. Notice that Maude does not instantiate these time variables with concrete values, but only with time symbols accumulating time constraints. Maude backtracks whenever the collection of accumulated constraints is unsatisfiable.

Our formalization encodes the Tick rule in the MSR theory which advances global time by any real number using two rewrite rules \texttt{Tick} and \texttt{Tick-Vclock}. The former rewrite rule advances time to any value \texttt{T2} within the verifier's current clock cycle. This is specified by the constraint \texttt{T2 < T1 + 1/1}. The second rewrite advances time to the beginning of the next verifier's clock cycle (\texttt{T1 + 1/1}). Intuitively, we consider a verifier's clock cycle as an event thus controlling how time advances and avoiding state space explosion.

A second difference to the the theory in Section~\ref{sec: formal} is that we assume a powerful verifier which measures the time of sending and receiving a message exactly at the beginning of clock cycle following the time when message is actually sent and received. This is specified, for example, by the rule \texttt{Dis-V-Rec}, in particular, by using the timestamp \texttt{toReal(toInteger(T2)) + 1/1} where \texttt{1/1} represents the real number one. If there is an attack with this more powerful verifier, then there is an attack in a less powerful verifier, \ie, it is sound. 

\vspace{2mm}
Finally, the third difference is on the way time advances. 
While it is convenient to separate the tick rules from the instantaneous rules in our theoretical framework as described in Section~\ref{sec:msr}, this distinction increases considerably search space. Instead, we advance time according to the events processed. This is similar to the behavior of Real-Time Maude~\cite{lisp/OlveczkyM07}, but here we use time symbols. Thus, whenever a message is received or sent or whenever the verifier's clock ticks, global time advances to the time of the corresponding event. For example, in rule \texttt{Real-P-Rcv}, the global time advances to a time \texttt{T3} greater than the time, \texttt{T2}, when a message is received. This time sampling is sound as no events are processed before they should, \eg, a message is not received before the corresponding time to travel elapses. 

\vspace{2mm}
For completeness, it seems possible to apply results from the literature, \eg, the completeness results of the time sampling used by Real-Time Maude~\cite{olveczky-meseguer-completeness}.  More recently, Nigam~\etal~\cite{nigam16esorics} have proved the completeness of time intruders using symbolic time constraints.

We used the \texttt{smt-search} to find a symbolic representation of a family of attacks. Using Maude SMT the potentially infinite search space becomes finite, by treating the distance between the verifier and the prover as a constrained variable. 

As an example, we considered the distance bound to be \texttt{3} and set \texttt{(dva > 3/1) = true} which means that the prover should not succeed the distance bound challenge as shown below:

\vspace{1mm}
{\small
\begin{verbatim}
Maude> smt-search [1] 
{ (Time @ 0/1) (vTime @ 0/1) (V0(p,n(0),n(10)) @ 0/1) (P0(v,n(10),n(0)) @ 0/1) (dist(dvp)) } 
 =>+ 
{ (Ok(p) @ T:Real) S:Soup } 
such that (dva > 3/1) = true .
\end{verbatim}
}

\vspace{1mm}
\noindent The following (simplified) solution is obtained after few seconds (ca 10 seconds) of computation:
 
\vspace{1mm}
{\small
\begin{verbatim}
S --> ... 
(Time @ toReal(toInteger(#3-T2:Real)) + 1/1)
(Start(p, n(0), n(10)) @ 0/1) (Stop(p, n(0), n(10)) @ #3-T2:Real) 
V2(stop,p,n(0),n(10)) @ toReal(toInteger(#3-T2:Real)) + 1/1
where dvp > 3/1 
...
and (#1-T2:Real >= 0/1 + dvp) and (#1-T2:Real >= #1-T2:Real and #2-T3:Real >= #1-T2:Real) 
and (#3-T2:Real >= #2-T3:Real + dvp) and #3-T2:Real >= #3-T2:Real
and toReal(toInteger(#3-T2:Real)) + 1/1 - (0/1 + 1/1) <= 2/1 * 3/1 
\end{verbatim}
}

It states that the verifier grants the resource to the prover, as the message \texttt{Ok} has been sent, although the prover is further away than the distance bound of \texttt{3}.

\section{Circle-Configurations }
\label{sec:circle}

This Section introduces the machinery, called Circle-Configurations, that can symbolically represent configurations and plans that mention dense time. 
Dealing with dense time leads to some difficulties, which have puzzled us for some time now, in particular, means to handle Zeno paradoxes. When we use discrete domains to represent time, such as the natural numbers, time  always advances by one, specified by the rule:
$$
 Time@T \lra Time@(T+1)
$$
\noindent
There is no other choice.\footnote{However, as time can always advance, a plan may use an unbounded number of natural numbers.  This source of unboundedness was handled in our previous work~\cite{kanovich15mscs}. This solution, however, does not scale
to dense time.} On the other hand, when considering systems with dense time, the problem is much more involved, as the non-determinism is much harder to deal with: the value that the time advances, the $\varepsilon$ in 
$$Time@T \lra Time@(T + \varepsilon)$$ 
can be instantiated by any positive real number.

Our claim is that we can symbolically represent any plan involving dense time by using a canonical form called 
circle-configurations. We show that circle-configurations provide a sound and complete representation of plans with dense time  
(Theorem~\ref{th:circle}). 

Recall Definition \ref{def:equivalence} and Theorem \ref{th:equiv_reach} where we show that the introduced equivalence between configurations corresponds to the reachability problem in the sense that a solution of a reachability problem is independent of the choice of  equivalent (initial) configurations.
Indeed, we will show that using circle-configurations we can symbolically represent the entire class of equivalent configurations
and, moreover, we can consider reachability problem over circle-configurations.

\medskip
A circle-configuration consists of two components: a \mbox{\emph{$\delta$-Configuration}}, $\Delta$, and a \emph{Unit Circle}, $\Uscr$, written $\tup{\Delta, \Uscr}$. Intuitively, the former
accounts for the integer part of the timestamps of facts in the configuration, while the latter deals with the decimal part of the timestamps.  

In order to define these components, however, we need some additional machinery. For a  real-number, $r$,
$\Int{r}$ denotes the integer part of $r$ and $\dec{r}$ its decimal part. For example, $\Int{2.12}$ is 2
and $\dec{2.12}$ is $0.12$. Given a natural number $\Dmax$, the \emph{truncated time difference (w.r.t. $\Dmax$)} between two facts 
$P@t_P$ and $Q@t_Q$ such that $t_Q \geq t_P$ is defined as follows
$$
 \delta_{P,Q} = 
\left\{\begin{array}{l}
 \Int{t_Q} - \Int{t_P}, \textrm{ if } \Int{t_Q} - \Int{t_P} \leq \Dmax\\
 \infty, \textrm{ otherwise }
\end{array}\right.
$$
For example, if $\Dmax = 3$ for the facts $F@3.12, ~G@1.01, ~H@5.05$ we have  $\delta_{F,H} = 2$ and $\delta_{G, H} = \infty$.
Notice that whenever we have $\delta_{P,Q} = \infty$ for two timestamped facts, $P@t_P$ and $Q@t_Q$, 
we can infer that \ $t_Q > t_P + D$ \ for any natural number $D$ in the theory. Thus, we can truncate time difference without sacrificing soundness and completeness. This was pretty much the idea used in~\cite{kanovich15mscs} to handle systems with discrete-time.

\paragraph{$\delta$-Configuration}
\label{subsec:delta}
\\
We now explain the first component, $\Delta$, of circle-configurations, $\tup{\Delta,\Uscr}$, namely the $\delta$-configuration, to only later enter into the details of the second component.

Given a configuration 
\mbox{$\Sscr = \{F_1@t_1, \ldots, F_n@t_n, Time@t\}$}, we construct its
\mbox{$\delta$-configuration} as follows: We first sort the facts using the integer part of their timestamps, 
obtaining the sequence of timestamped facts $Q_1@t_1', \ldots, Q_{n+1}@t_{n+1}'$, where $t_{i}' \leq t_{i+1}'$ 
for $1 \leq i \leq n+1$ and $\{Q_1, \ldots, Q_{n+1}\} = \{F_1, \ldots, F_n, Time\}$. We then aggregate in classes facts with the same  integer part of the timestamps 
obtaining a sequence of classes 
$$\{Q_{1}^1, \ldots, Q_{m_1}^1\}, \{Q_{1}^2, \ldots, Q_{m_2}^2\}, \ldots, 
\{Q_{1}^j, \ldots, Q_{m_j}^j\} \ ,$$
 where \ $\delta_{Q_i^k,Q_j^k} = 0$  \ for any $1 \leq i\leq m_k$ and $1 \leq k \leq j$, and \ $\{  Q_{1}^1, \ldots,  Q_{m_j}^j \}  = \{ Q_1, \dots, Q_{n+1} \}$.

\vspace{1mm}
The $\delta$-configuration of $\Sscr$ is then:
\[
 \Delta = \left\langle
 \begin{array}{l}
 ~\{Q_{1}^1, \ldots, Q_{m_1}^1\},~\delta_{1, 2}, ~\{Q_{1}^2, \ldots, Q_{m_2}^2\}, \ldots 
  ,~\{Q_{1}^{j-1}, \ldots, ~Q_{m_{j-1}}^{j-1}\}, ~\delta_{j-1, j}, \{Q_{1}^j, \ldots, ~Q_{m_j}^j\}~
 \end{array}\right\rangle
\]
where ~$\delta_{i,i+1} = \delta_{Q_{1}^i, Q_{1}^{i+1}}$~ is the truncated time difference between the facts in class $i$ and class $i+1$.
\\
For such a $\delta$-configuration, $\Delta$, we define 
\[
\Delta({Q_i^l, Q_j^h}) = \left \{ \begin{array}{ll}
                             \ \  \sum\limits_{k = l}^{k = h - 1} \delta_{k, k+1} & \ \ \textrm{if ~$h \geq l$} \\
                               - \sum\limits_{k = h}^{k = l - 1} \delta_{k, k+1} &  \ \ \textrm{otherwise}
                              \end{array}\right. 
\]
which is the truncated time difference between any two facts of $\Delta$, $Q_i^l$ and $Q_j^h$ from the classes $l$ and $h$, respectively.
Here we assume $\infty$ is the addition absorbing element, \ie, \mbox{$\infty + D = \infty$} for any natural number $D$ and
$\infty + \infty = \infty$.

\vspace{2mm}
Notice that, for a given upper bound $\Dmax$, different configurations may have the same $\delta$-configuration. 
For example, with $\Dmax =4$, configurations
\begin{equation}\label{eq: S1}
\begin{array}{c}
 \Sscr_1 = \{~M@3.01, ~R@3.11, ~P@4.12, ~Time@11.12, ~Q@12.58, ~S@14~\}\quad
 \textrm{and}\\
 \Sscr_1' = \{~M@0.2, ~R@0.5, ~P@1.6, ~Time@6.57, ~Q@7.12, ~S@9.01~\}\\
\end{array}
\end{equation}
have both the following $\delta$-configuration: 
$$
 \Delta_{\Sscr_1} = \tup{~\{M,R\}, 1, \{P\}, \infty, \{Time\}, 1, \{Q\}, 2, \{S\}~} .
$$
This $\delta$-configuration specifies the truncated time differences between the facts from $\Sscr_1'$. For example, ~\mbox{$\Delta_{\Sscr_1}(R, P) = 1$}, that is, 
the integer part of the timestamp of the fact $P$ is ahead one unit 
with respect to the integer part of the timestamp of the fact $R$. Moreover, the timestamp of the fact $Time$ is 
more than $\Dmax$ units ahead with respect to the timestamp of $P$. This is indeed true
for both configurations $\Sscr_1$ and $\Sscr_1'$ given above. 

\bigskip
\paragraph{Unit Circle}
\\
In order to handle the decimal part of the timestamps, we use intervals instead of concrete values. 
These intervals are represented by a circle, called Unit Circle, which together with a $\delta$-configuration
composes a circle-configuration. The unit circle of a configuration \mbox{$\Sscr = \{F_1@t_1, \ldots, F_n@t_n, Time@t\}$} is constructed
by first ordering the facts from $\Sscr$  according to the \emph{decimal part} of their timestamps in the increasing order. In such a way we  obtain the sequence of facts \ $Q_1, \ldots, Q_{n+1}$, where \ $\{Q_1, \ldots, Q_{n+1}\} = \{F_1, \ldots, F_n, Time\}$. 
 Then the unit circle of the given configuration $\Sscr$ is obtained by aggregating facts 
that have the same \emph{decimal part} obtaining a sequence of classes:
$$
\Uscr = [~\{Q_{1}^0, \ldots, Q_{m_0}^0\}_\Zscr, ~\{Q_{1}^1, \ldots, Q_{m_1}^1\}, \ldots, ~\{Q_{1}^j, \ldots, Q_{m_j}^j\}~]
$$

\vspace{1mm}
\noindent
where   \ $\{  Q_{1}^1, \ldots,  Q_{m_j}^j \}  = \{ Q_1, \dots, Q_{n+1} \}$,
the facts in the same class have the same decimal part, \ie,
$\dec{Q_k^i}= \dec{Q_l^i}$, for all $1\leq k \leq m_i$, $1 \leq l \leq m_i$ and $1\leq i\leq j$, 
classes are ordered in the increasing order, \ie, $\dec{Q_k^i}<\dec{Q_l^{i'}}$ ~for all ~$i\neq i'$, where 
 $1\leq k \leq m_i$, $1 \leq l \leq m_{i'}$, $0\leq i\leq j$, $1 \leq i'\leq j$,
 and the first class \ $\{Q_{1}^0, \ldots, Q_{m_1}^0\}_\Zscr$, marked with the subscript $\Zscr$ contains all facts
whose timestamp's decimal part is zero, \ie, \ $\dec{Q_{i}^0} = 0$, for $1 \leq i \leq m_0$. 

We call the class \ $\{Q_{1}^0, \ldots, Q_{m_1}^0\}_\Zscr$  \ the \emph{Zero Point}. Notice that the zero point may be 
empty.

For a unit circle, $\Uscr$, we define:
 $\Uscr(Q_{j}^i) = i$
to denote the class in which the fact $Q_{j}^i$ appears in $\Uscr$. 

\vspace{2mm}
For example, the unit circle of configuration $\Sscr_1$ given in (\ref{eq: S1}) is the sequence: 
$$\Uscr_{\Sscr_1} = [ \, \{S\}_\Zscr, ~\{M\}, ~\{R\}, ~\{P, Time\},~\{Q\} \, ] \ .$$
Notice that $P$ and $Time$ are in the same class as the decimal parts of their timestamps are the same, 
namely $0.12$. Moreover, we have that ~$\Uscr_{\Sscr_1}(S) = 0 < 2 = \Uscr_{\Sscr_1}(R)$, specifying
that the decimal part of the timestamp of the 
fact $R$ is greater than the decimal part of the timestamp of the fact $S$.

\vspace{2mm}
We will graphically represent a unit circle as shown in Figure~\ref{fig:unit-graph}.
The (green) ellipse at the top of the circle marks the {zero point},
while the remaining classes are placed on the circle in 
the (red) squares ordered clockwise starting from the zero point. Thus, 
from the above graphical representation it can be seen that the decimal part of the timestamp of the fact $Q_{1}^1$ 
is smaller than the decimal of the timestamp of the fact $Q_{1}^2$, while the decimal part of the 
timestamps of the facts $Q_{1}^i$ and $Q_{2}^i$ are equal. The exact points where the squares are 
placed are not important, only their relative positions matter, \eg, the square for 
the class containing the fact $Q_{1}^1$ should be placed on the circle somewhere in between the zero point and the square for the 
class containing the fact $Q_{1}^2$, clockwise. 

\begin{figure}[t]
 \vspace{2mm}
\includegraphics[width=4.3cm]{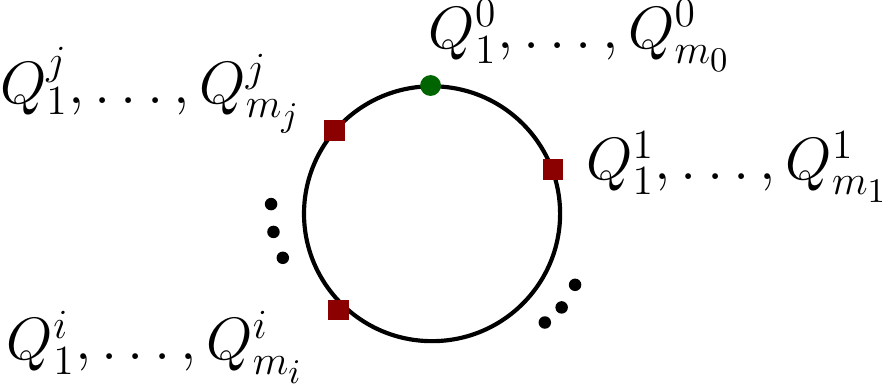}
\vspace{1mm}
\caption{Unit Circle}
\label{fig:unit-graph}
\vspace{1em}
\end{figure}

\begin{definition}
\label{def:circle-conf}
Let $\Tscr$ be a reachability problem,  $\Dmax$ an upper bound on the numeric values  appearing in $\Tscr$ and \ 
$\Sscr = \{~ F_1@t_1, ~F_2@t_2, \ldots, ~F_n@t_n, ~Time@t ~\}  . $
Let
\[
 \Delta_{\Sscr} = \left\langle
 \begin{array}{l}
 \{P_{1}^1, \ldots, P_{m_1}^1\}, \delta_{1, 2}, \{P_{1}^2, \ldots, P_{m_2}^2\}, \ldots 
  \{P_{1}^{j-1}, \ldots, P_{m_{j-1}}^{j-1}\}, \delta_{j-1, j}, \{P_{1}^j, \ldots, P_{m_j}^j\}
 \end{array}\right\rangle
\]
 where \
 $\{P_1^1, \ldots, P_{m_1}^1, P_1^2, \ldots, P_{m_j}^j\} = \{F_1, \ldots, F_n, Time\}$,
timestamps of facts \ $P_1^i, \ldots, P_{m_i}^i$ \ have the same integer part, $t^{\, i}$,  $\forall i = 1, \dots, j$ ,
and 
\ $\delta_{i,i+1}$ \
 is the truncated time difference (w.r.t. $\Dmax$) between a fact in class $i$ and a fact in class $i+1$.
Let
$$
\Uscr_{\Sscr} = [ \ \{Q_{1}^0, \ldots, Q_{m_0}^0\}_\Zscr, \{Q_{1}^1, \ldots, Q_{m_1}^1\}, \ldots, \{Q_{1}^k, \ldots, Q_{m_k}^k\} \ ]
$$
 where \
$\{Q_1^0, \ldots, Q_{m_0}^0, Q_1^1, \dots, Q_{m_k}^k\} = \{F_1, \ldots, F_n, Time\}$,  \
timestamps of facts \ $Q_1^i, \ldots, Q_{m_i}^i$ \ have the same decimal part,  $\forall i = 0, \dots, k$ ,
and timestamps of facts \ $Q_1^0, \ldots, Q_{m_0}^0$ \ are integers.
We say that $\Delta_{\Sscr}$  is the \emph{$\delta$-configuration} of $\Sscr$,
 \ $\Uscr_{\Sscr}$ is the \emph{unit-circle} of $\Sscr$, 
and \ $\Ascr_{\Sscr}=\tup{\Delta_{\Sscr}, \Uscr_{\Sscr}}$  \ is the  \emph{circle-configuration} of the configuration $\Sscr$ (or the circle-configuration corresponding to $\Sscr$). 
\end{definition}

\vspace{2mm}
For example, with $\Dmax =4$, configuration
\begin{equation}\label{eq: S1a}
\begin{array}{c}
 \Sscr_1 = \{ \ M@3.01, R@3.11, P@4.12, Time@11.12, Q@12.58, S@14\ \}
\end{array}
\end{equation}
has the  circle-configuration $\tup{\Delta_{\Sscr_1}, \Uscr_{\Sscr_1}}$, where
$$
\begin{array}{rl}
\Delta_{\Sscr_1} = &
\left\langle 
\  \{ M, R\} 1, \{ P\}, \infty, \{Time\}, 1, \{Q\}, 2, \{ S\}  \
\right\rangle
\\
 \Uscr_{\Sscr_1} = & [\ \{S\}_\Zscr, \{M\}, \{R\}, \{P, Time\},\{Q\} \ ] \ ,
\end{array} 
$$
as depicted in Figure~\ref{fig:circle-graph} .

For simplicity, we sometimes write $\Ascr$ and $\tup{\Delta, \Uscr}$ instead of  $\Ascr_{\Sscr}$ and $\tup{\Delta_{\Sscr}, \Uscr_{\Sscr}}$, when the corresponding configuration is clear from the context.

\vspace{2mm}

Notice that $\delta_{i,i+1}$ is well defined, since all the facts in the same class  \ $P_1^k, \ldots, P_{m_k}^k$ \ have the same integer part, $t^{\, k}$,  $\forall k = 1, \dots, j$. Namely,
$$ 
\delta_{i,i+1} = \delta_{Q_{1}^i, Q_{1}^{i+1}}= 
\left\{\begin{array}{cl}
{t^{\, i+1}} - t^{\, i}, & \ \textrm{ if }\ \ {t^{\, i+1}} - t^{\, i} \leq \Dmax\\
 \infty, & \ \textrm{ otherwise }
\end{array}\right.
\ , \qquad  i = 1, \dots, j-1 \ .
$$

\subsection{Constraint Satisfaction}
\label{subsec:constraints}

\begin{figure}[t]
\vspace{1mm}
\includegraphics[width=6cm]{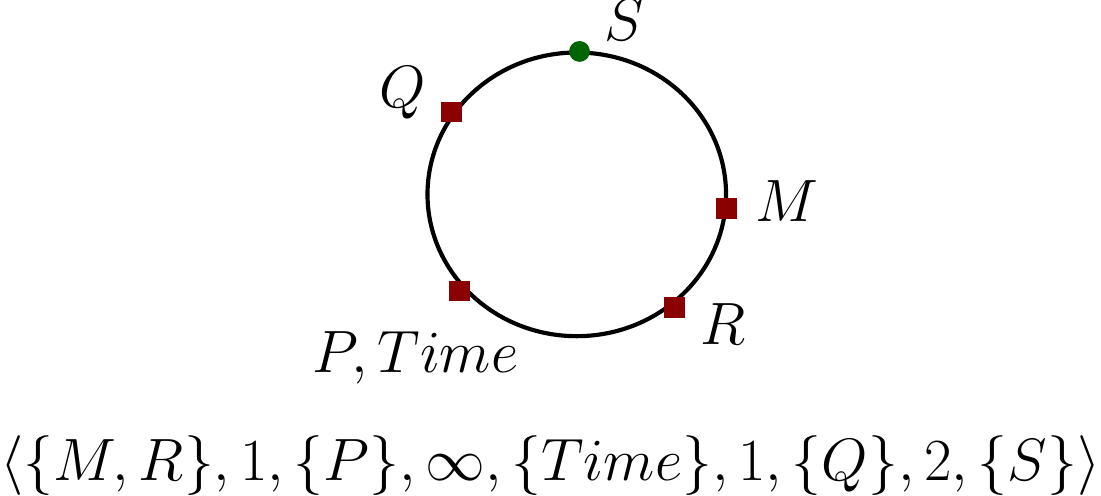}
\caption{Circle-Configuration}
 \label{fig:circle-graph}
\vspace{1em}
\end{figure} 
A circle-configuration $\tup{\Delta, \Uscr}$ contains all the information needed in order to determine whether 
a constraint of the form used in our model, \ie,~given in (\ref{eq:constraints}), is satisfied or not. 

Consider the circle-configuration in Figure~\ref{fig:circle-graph} which corresponds to configuration $\Sscr_1$ given in (\ref{eq: S1a}).
To determine, for instance, whether $t_Q > t_{Time} + 1$, we compute the integer difference between $t_Q$ and $t_{Time}$ from the $\delta$-configuration. This turns out to be 1 and means that we need to look at the decimal part of these timestamps to determine whether the constraint is satisfied or not. 
Since the decimal part of $t_Q$ is greater than the
decimal part of $t_{Time}$, as can be observed in the unit circle,
we can conclude that the constraint is satisfied. Similarly, one can 
also conclude that the constraint $t_Q > t_{Time} + 2$ is not satisfied as  
$\Int{t_Q} = \Int{t_{Time}} + 1$. 
The following results  formalize this intuition.

\vspace{2mm}
\begin{lemma}
\label{th:circle_constraints}
 Let \ $\tup{\Delta, \Uscr}$ \ be a circle-configuration of the configuration $\Sscr$. 
Then for  two arbitrary facts $P@t_p$ and $Q@t_Q$ in $\Sscr$
 and a natural number \mbox{$D<\Dmax$} the following holds:
 \begin{itemize}
  \item $t_P > t_Q + D$ \ \ \ iff \ \ \ $\Delta({Q,P}) > D$ \ or\   \ $ ( \ \Delta({Q,P}) = D$ and $\Uscr(P) > \Uscr(Q) \ )$  ;
  \item $t_P > t_Q - D$ \ \ \ iff \ \ \ $\Delta({P,Q}) < D$ \ or \  $( \ \Delta({P,Q}) = D$ and $\Uscr(P) > \Uscr(Q) \ ) $ ;
  \item $t_P = t_Q + D$ \ \ \ iff \ \ \  $\Delta({Q,P}) = D$ \ and \ $\Uscr(Q) = \Uscr(P)$;
  \item $t_P = t_Q - D$ \ \  \ iff \ \ \ $\Delta({P,Q}) = D$ \ and \ $\Uscr(Q) = \Uscr(P)$;
 \end{itemize}
\end{lemma}
\begin{proof} 
 Let $t_P \geq t_Q$. 
Recall that, as per definition of circle-configurations, there are $i$ and $j$ such that $Q= Q_i$ and $P= Q_j$, and that
$$
\begin{array}{rl}
\Delta({Q,P}) & =  \Delta({Q_i,Q_j})  =  \sum \delta_{k, k+1} \\
&  
=\Int{Q_j} - \Int{Q_{j-1}} +  \Int{Q_{j-1}}  - \Int{Q_{j-2}} + \dots - \Int{Q_{i}}\\
&  = \Int{Q_j} - \Int{Q_i} = \Int{P} - \Int{Q}
\end{array}
$$
Then, for a natural number ~$D \leq \Dmax$ \  it holds that
$$
\begin{array}{rl}
t_P = t_Q + D  \ \ \ \text{iff} \ \ \ & \Int{t_P-t_Q} = D \ \ \text{and} \ \  \dec{t_P-t_Q} = 0 \\
 \ \ \ \text{iff} \ \ \  &    \Delta({Q,P}) = D \ \ \text{and} \ \ \Uscr(Q) = \Uscr(P)
\end{array}
$$
because   \ $\dec{t_P} = \dec{t_Q} $ \ implies \ $\Int{t_P-t_Q} = D$ \ iff \  $\Int{t_P}-\Int{t_Q} = D$.
Similarly,
$$
\begin{array}{rcl}
t_P > t_Q + D  \ \ \text{iff} \ \ \ & \Int{t_P-t_Q} > D  \ \ & \text{or} \ \ \  \big( \  \Int{t_P-t_Q} = D  \ \ \text{and}  \ \ \dec{t_P- t_Q } > 0 \  \big)\\
 \ \ \ \text{iff} \ \ \  &  \Delta({Q,P}) > D  \ \ & \text{or} \ \ \ \big( \  \Delta({Q,P})= D  \ \ \text{and} \ \ \Uscr(P) > \Uscr(Q) \ \, \big)
\end{array}
$$
The remaining cases are proven similarly.
\end{proof}

Following the above result we say  that a circle-configuration  $\tup{\Delta, \Uscr}$ corresponding to a configuration $\Sscr$ \emph{satisfies the constraint} $c$ if the configuration $\Sscr$ satisfies the constraint $c$. We also say that  a  circle-configuration  is a \emph{goal circle-configuration} if it corresponds to a goal configuration. Furthermore, we also say that an \emph{action is applicable to a  circle-configuration} if that action is applicable to the corresponding configuration.

\subsection{Rewrite Rules and Plans with Circle-Configurations}

\vspace{-1mm} \noindent
This section shows that given a reachability problem with a set of rules, $\Rscr$, involving dense time, and an upper bound on the numbers appearing in the problem, $\Dmax$, we can compile a set of rewrite rules, $\Ascr$, over circle-configurations. Moreover, we 
show that any plan generated using the rules from $\Rscr$ can be soundly and faithfully 
represented by a plan using the set of rules $\Ascr$.
We first explain how we apply instantaneous rules to circle-configurations and then 
we explain how to handle the time advancement rule.

\paragraph{Instantaneous Actions}

\vspace{2mm}
Let $\Dmax$ be an upper bound on the numeric values in the given problem and let the following rule be an 
instantaneous rule (see Section~\ref{subsec:actions}) from  the set of actions $\Rscr$:
\[
\begin{array}{l}
 Time@T, ~W_1@T_1, \ldots, ~W_k@T_k, ~F_1@T_1', \ldots, ~F_n@T_n' \ \mid \ \Cscr ~\lra~ \\
 \quad \exists~\vec{X}.~[~Time@T, ~W_1@T_1, \ldots, ~W_k@T_k, ~Q_1@(T + D_1), \ldots, ~Q_m@(T + D_m)~]
\end{array}
\]
The above rule is compiled into a sequence of operations that may rewrite 
a given circle-configuration $\tup{\Delta, \Uscr}$ into another circle-configuration 
$\tup{\Delta_1, \Uscr_1}$ as follows:
\vspace{1mm}
\begin{enumerate}
 \item Check whether there are occurrences of $W_1, \ldots, W_k$ and $F_1, \ldots, F_n$ in $\tup{\Delta, \Uscr}$ 
 such that the guard $\Cscr$ is satisfied by $\tup{\Delta, \Uscr}$.  If it is
 the case, then continue to the next step; otherwise the rule is not applicable;
 
 \item We obtain the circle-configuration $\tup{\Delta', \Uscr'}$ by removing 
 a single occurrence of each of the facts $F_1, \ldots, F_n$ \ in \ $\tup{\Delta, \Uscr}$ used in step 1, and recomputing 
 the truncated time differences so that 
 for all the remaining facts $P$ and  $R$ in $\Delta$, we have \ $\Delta'(P,R) = \Delta(P, R)$, 
 \ie, the truncated time difference between $P$ and  $R$ is preserved;
 
 \item Create fresh values, $\vec{e}$, for the existentially quantified variables $\vec{X}$; 
 
 \item We obtain the circle-configuration $\tup{\Delta_1, \Uscr_1}$ by adding the 
 facts \ $Q_1[\vec{e}/\vec{X}], \ldots, Q_m[\vec{e}/\vec{X}]$ \ to $\Delta'$ so that 
 $\Delta_1(Time, Q_i) = D_i$, for $1 \leq i \leq m$, and that \ $\Delta_1(P, R) = \Delta'(P, R)$ \
 for all the remaining facts $P$ and  $R$ in $\Delta'$.
 We then obtain $\Uscr_1$ by adding the facts \ $Q_1, \ldots, Q_m$ \ to the class of the fact $Time$ in the unit circle $\Uscr'$;
 
 \item Return the circle-configuration $\tup{\Delta_1, \Uscr_1}$.
\end{enumerate}

\vspace{3mm}
The sequence of operations described above has the effect one would expect: replace the facts $F_1, \ldots, F_n$ in 
the pre-condition of the action with facts $Q_1, \ldots, Q_m$ appearing in the post-condition 
of the action but taking care to update the truncated time differences in the $\delta$-configuration. 
Moreover, all steps can be computed in polynomial time.

\vspace{2mm}
For example, consider the configuration $\Sscr_1$ given in (\ref{eq: S1}) and the rule:
$$
 Time@T, ~R@T_1, ~P@T_2  ~\lra~ 
Time@T, ~P@T_2, ~N@(T + 2)
$$
\noindent
If we apply this rule to $\Sscr_1$, we obtain the configuration 
$$
\Sscr_2 = \{~M@3.01, ~P@4.12, ~Time@11.12, ~Q@12.58, ~N@13.12, ~S@14~\}.
$$
\noindent

\begin{wrapfigure}{r}{0.4\textwidth}
\vspace{-3mm}
\quad
\includegraphics[width=6cm]{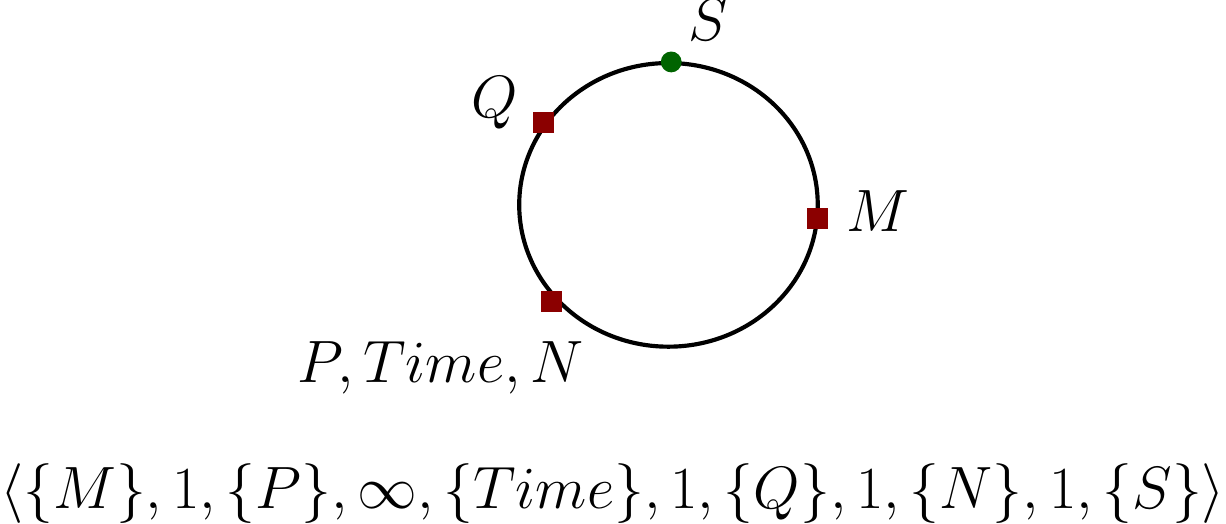}
\vspace{4mm}
\end{wrapfigure}
On the other hand, if we apply the above steps to the circle-configuration of $\Sscr_1$, shown 
in Figure~\ref{fig:circle-graph},
we obtain the circle-configuration shown
 to the right.
It is easy to check that this is indeed the circle-configuration of $\Sscr_2$.
The truncated time differences are updated and the fact $N$ 
is added to the class of $Time$ in the unit circle.

\newcommand\figspace{\vspace{4pt}}

\begin{figure}[t]
\textbullet\quad\textbf{\small Time in the zero point and not in the last class in the unit circle, where $n \geq 0$:}
\figspace

\includegraphics[width=8cm]{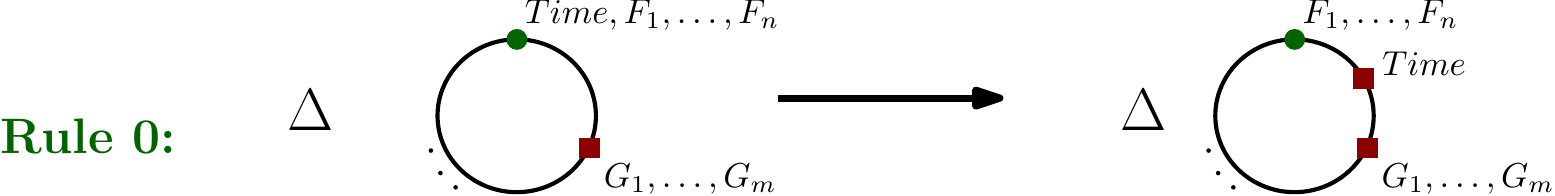}
\figspace

\textbullet\quad\textbf{\small Time alone and not in the zero point nor in the last class in the unit circle:}
\figspace

\includegraphics[width=8cm]{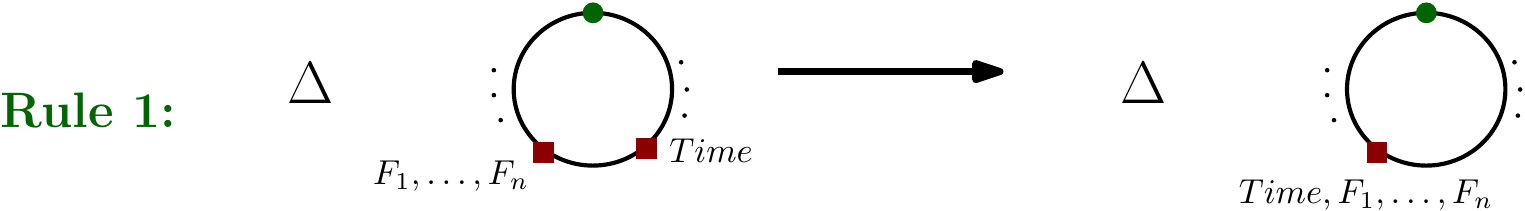}
\figspace

\textbullet\quad\textbf{\small Time not alone and not in the zero point nor in the last class in the unit circle:}
\figspace

\includegraphics[width=10cm]{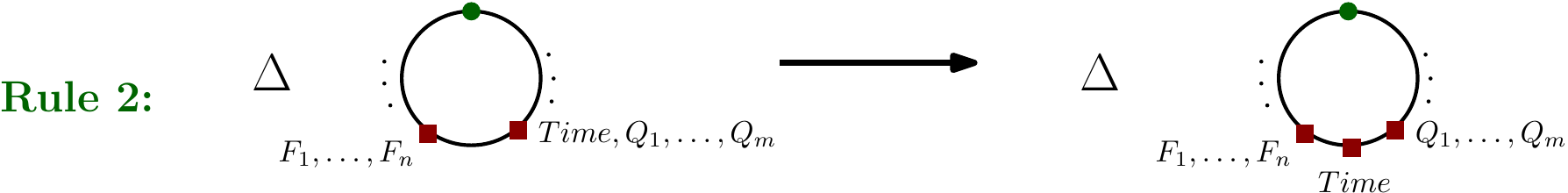}
\figspace

\textbullet\quad\textbf{\small Time not alone and in the last class in the unit circle which may be at the zero point:}
\figspace

\includegraphics[width=9cm]{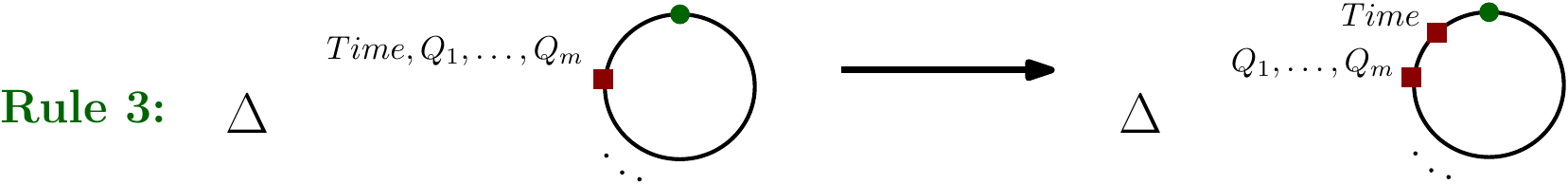}
\figspace

 \caption{Rewrite Rules for Time Advancement using Circle-Configurations.
 }
 \label{fig:time-advance}
 \vspace{1em}
\end{figure}

\begin{figure}[t]
\textbullet\quad\textbf{\small Time alone and in the last class in unit circle - Case 1: $m > 0, k \geq 0, n\geq 0$ and $\delta_1 > 1$:}
\figspace

\includegraphics[width=12cm]{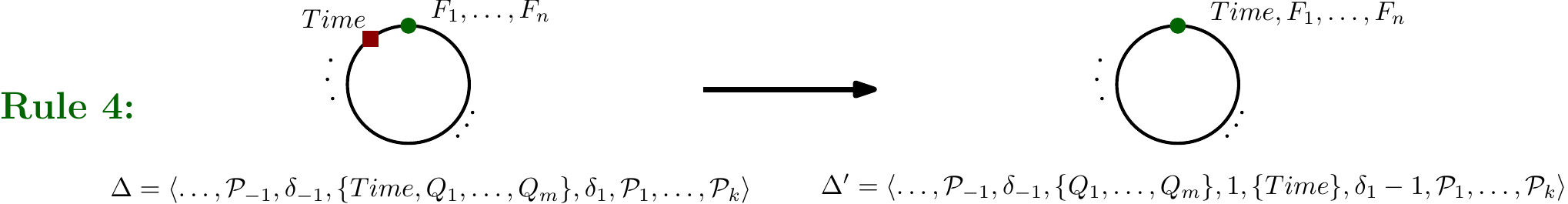}
\figspace

\textbullet\quad\textbf{\small Time alone and in the last class in unit circle - Case 2: $m > 0, k \geq 1$ and $n\geq 0$:}
\figspace

\includegraphics[width=11cm]{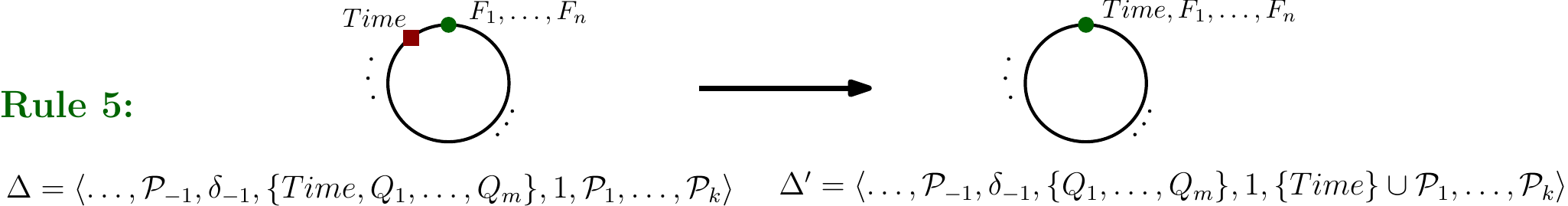}
\figspace

\textbullet\quad\textbf{\small Time alone and in the last class in unit circle - Case 3: $k \geq 0$ such that $\delta_{1} > 1$ when $k > 0$ and $\gamma_{-1}$
is the truncated time of $\delta_{-1} + 1$:}
\figspace

\includegraphics[width=11cm]{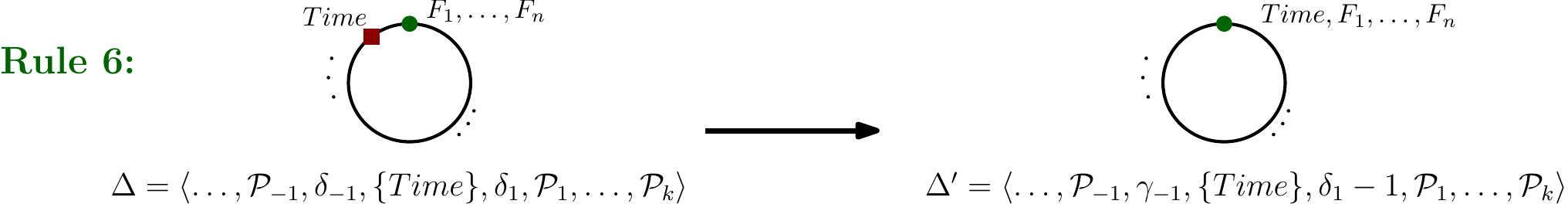}
\figspace

\textbullet\quad\textbf{\small Time alone and in the last class in unit circle - Case 4: $k \geq 1$ and $\gamma_{-1}$
is the truncated time of $\delta_{-1} + 1$:}
\figspace

\includegraphics[width=11cm]{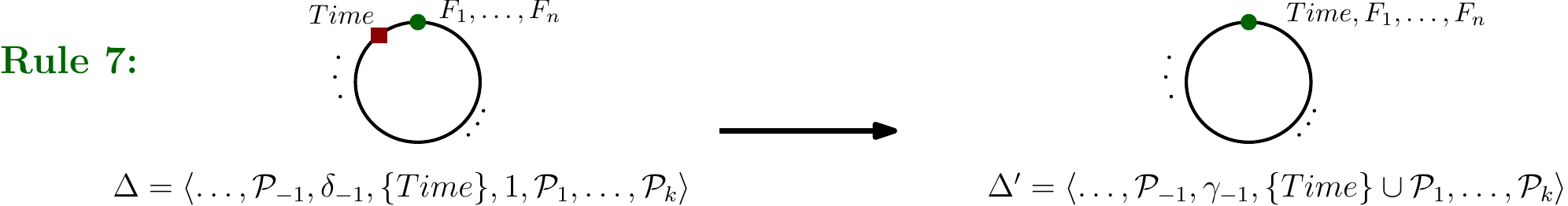}

\textbullet\quad\textbf{\small Time alone and in the last class in unit circle and in $\Delta$ - Case 5:  $\gamma_{-1}$
is the truncated time of $\delta_{-1} + 1$:}
\figspace

\vspace{-1mm}
 \caption{(Cont.) Rewrite Rules for Time Advancement using Circle-Configurations.
 }
 \label{fig:time-advance-cont}
\vspace{1em}
\end{figure}

\vspace{2mm}
\paragraph{Time Advancement Rule}

Specifying the time advancement rule 
$$Time@T \lra Time@(T + \varepsilon)$$
 over circle-configurations is more interesting. This action is translated into the rules depicted in Figures~\ref{fig:time-advance} and \ref{fig:time-advance-cont}. 
There are eight rules that rewrite a circle-configuration, $\tup{\Delta, \Uscr}$, depending on the position of 
the fact $Time$ in the unit circle $\Uscr$.  

Rule 0 specifies the case when the fact $Time$ appears in the {zero point} of $\Uscr$. Then $\Uscr$ is re-written so 
that a new class is created immediately after the zero point clockwise before any other class on the unit circle, and $Time$ is moved to that new class. This denotes that the decimal 
part of $Time$ is greater than zero and less than the decimal part of the facts in the following class ~$G_1, \ldots, G_n$.

Rule 1 specifies the case when  the fact 
$Time$ appears 
alone in a class on the unit circle and not in the last class. This means that there are some facts, $F_1, \ldots, F_n$, that 
appear in a class immediately after $Time$, \ie, $\Uscr(F_i) > \Uscr(Time)$ and for any other fact $G$, $G \notin\{F_1, \dots, F_n\}$, such that 
$\Uscr(G) > \Uscr(Time)$,  $\Uscr(G) > \Uscr(F_i)$ holds.  In this case, then time can advance so that 
it ends up in the same class as $F_i$, \ie,  time has advanced so much that its decimal part is the same as the decimal part
of the timestamps of $F_1, \ldots, F_n$. Therefore a constraint of the form $T_{F_i} > T_{Time} + D$ that was 
satisfied by $\tup{\Delta, \Uscr}$  might no longer be satisfied
by the resulting circle-configuration, depending on $D$ and the $\delta$-configuration~$\Delta$.  

Rule 2 is similar, but is only applicable when $Time$ is not alone in the 
unit circle class, \ie, there is at least one fact $Q_i$ such that $\Uscr(Time) = \Uscr(Q_i)$ and this class is not 
the last one, as in Rule 1.
Then, Rule~2 advances time enough so that its decimal part is greater
than the decimal part of the timestamps of $Q_i$, but not greater than the decimal part of the timestamps of the facts
in the class that immediately follows on the circle.

\begin{wrapfigure}{r}{0.45\textwidth}
\vspace{-2mm}
 \hspace{3mm}
\includegraphics[width=5cm]{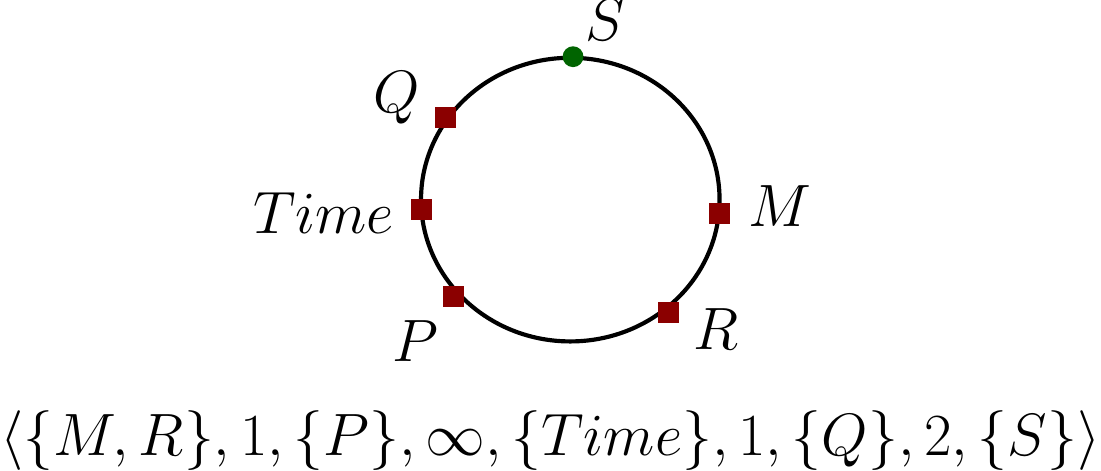}
 \caption{circle-configuration $\Ascr_{\Sscr_3}$  }
 \label{fig:circle-graph-3}
\vspace{-1mm}
\end{wrapfigure}
For example, Rule 2 could be applied to the circle-configuration $\Ascr_{\Sscr_1}$  shown in Figure~\ref{fig:circle-graph}.
We obtain the following circle-configuration, where the $\delta$-configuration does 
not change, but the fact $Time$  is moved to a new class on the  
unit circle, obtaining the  circle-configuration $\Ascr_{\Sscr_3}$ 
 shown to the right, Figure~\ref{fig:circle-graph-3}.

Rule 3 is similar to Rule 2, but it is applicable when $Time$ is in the last equivalence class, 
in which case a new class is created and placed clockwise immediately before the zero point of the circle. 

Notice that the $\delta$-configuration is not changed by \mbox{Rules 0-3}. The only rules that change the $\delta$-configuration are
the Rules 4, 5, 6 and 7, as in these cases $Time$ advances enough to complete the unit circle, \ie, reach the zero point.

Rules 4 and 5 handle the case when $Time$  initially has the same integer part as timestamps of other facts $Q_1, \ldots, Q_m$, in which case it might create a new class in the $\delta$-configuration (Rule 4) or merge with the following class $\Pscr_1$ (Rule 5). Rules 6 and 7 handle the case when $Time$ does not have the same integer part as the timestamp of any other fact, \ie, it appears alone in 
$\Delta$, in which case it might still remain alone in the same class (Rule 6) or merge with the following class $\Pscr_1$ (Rule 7).
Notice that the time difference, $\delta_{-1}$, to the class, $\Pscr_{-1}$, immediately before the class of $Time$
is incremented by one and truncated by the value of $\Dmax$ if necessary.

\begin{wrapfigure}{r}{0.4\textwidth}
\vspace{-1mm}
\includegraphics[width=4.5cm]{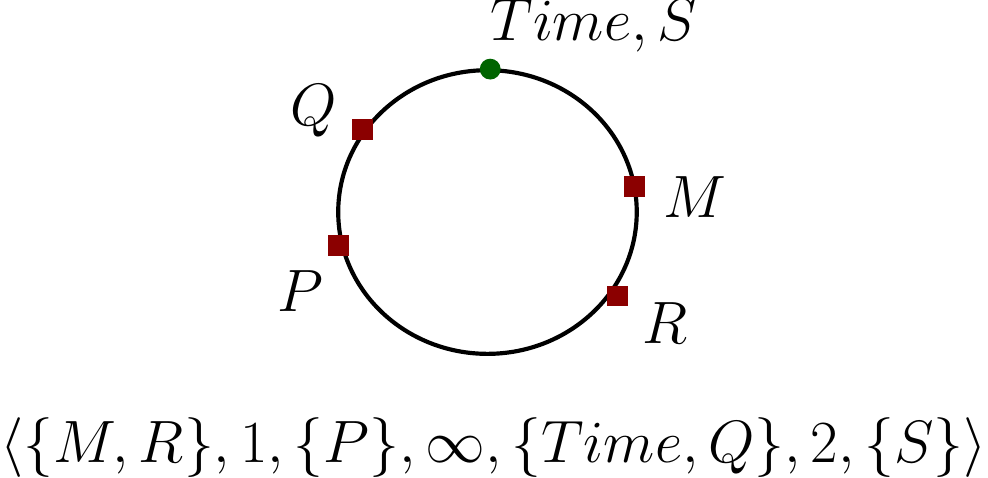} \quad
\vspace{-4mm}
\end{wrapfigure}

\vspace{2mm}
For example, consider the circle-configuration \ $\Ascr_{\Sscr_3}$  
 illustrated in Figure~\ref{fig:circle-graph-3}.
It is easy to check that by applying Rule 1, followed by Rule 3  to 
 $\Ascr_{\Sscr_3}$  
 we obtain a circle-configuration for which the Rule 7 is applicable.
 After applying Rule 7 we obtain the configuration $\Ascr_{\Sscr_4}$ shown to the right.

\medskip
Given a reachability problem $\Tscr$ and an upper bound $\Dmax$ on the numeric values of $\Tscr$ with the set of rules $\Rscr$ containing
an instantaneous rule $r$, we write 
$[r]$ for the corresponding rewrite rule of $r$ over circle-configurations as described above. Moreover, let $Next$ be the set of 8 time advancing rules shown in Figures~\ref{fig:time-advance} and  \ref{fig:time-advance-cont}. 
Notice that for a given circle-configuration only one of these rules is applicable. 

Therefore, we write
$$\widetilde{\Rscr} = \{ \ [r] : r \in \Rscr \ \} \cup Next \ $$
for the set of rules over circle-configurations corresponding to the set of rules $\Rscr$ over configurations.

We use  \ $\Ascr \lra_{a} \Ascr_1$ \ for the one-step reachability relation between circle-configurations using the rewrite rule $a$, \ie, 
the circle-configuration $\Ascr$ may be rewritten to the circle-configuration $\Ascr_1$ 
using the rewrite rule $a$. Finally, ~$\Ascr \lra^* \Ascr_1$ (respectively, ~$\Ascr \lra_{\Rscr'}^* \Ascr_1$)
denotes the reflexive transitive closure 
relation of the one-step relation (respectively, using only rules from the set \ $\Rscr' \subseteq \widetilde{\Rscr}$ ).

\medskip
\begin{lemma}
\label{th:circle_action}
 Let $\Tscr$ be a reachability problem and $\Dmax$ be an upper bound on the numeric values in $\Tscr$.
 Let   $\Ascr_1$ be the circle-configuration of the configuration $\Sscr_1$,
 and $r$ be an instantaneous action in $\Tscr$. 
Then \ $\Sscr_1 \lra_r \Sscr_2$ \  if and only  if \ \mbox{$\Ascr_1 \lra_{[r]} \Ascr_2$}, where $\Ascr_2$ is the circle-configuration of $\Sscr_2$. Moreover,  \  $\Sscr_1 \lra_{Tick} \Sscr_2$ \  if and only if \ $\Ascr_1 \lra_{Next}^* \Ascr_2$, where $\Ascr_2$ is the circle-configuration of $\Sscr_2$.
\end{lemma}
\begin{proof}
Following  Lemma~\ref{th:circle_constraints} we only need to prove the claim regarding $Tick$ and $Next$ rules as  $\Sscr_1$ satisfies the time constraints of rule $r$ if and only if its circle-configuration $\Ascr_1$ also satisfies the constraints in $[r]$. Moreover $\Ascr_2$ is the circle-configuration of $\Sscr_2$ by construction.

In order to prove that, given a transition $Next$ over circle-configurations,  there is the matching pair of configurations and the appropriate instance of $Tick$ rule, we first show how to extract a corresponding concrete configuration $\Sscr$ from a given circle-configuration ~$\Ascr= \tup{\Delta, \Uscr}$, where
\[
\begin{array}{rl}
 \Delta = & \left\langle
 \begin{array}{l}
~ \{P_{1}^1, \ldots, P_{m_1}^1\}, \delta_{1, 2}, \{P_{1}^2, \ldots, P_{m_2}^2\}, \ldots 
  \{P_{1}^{j-1}, \ldots, P_{m_{j-1}}^{j-1}\}, \delta_{j-1, j}, \{P_{1}^j, \ldots, P_{m_j}^j\} ~
 \end{array}\right\rangle \ ,
\\
\Uscr =&  [ \ \{Q_{1}^0, \ldots, Q_{m_0}^0\}_\Zscr, ~\{Q_{1}^1, \ldots, Q_{m_1}^1\}, \ldots, ~\{Q_{1}^k, \ldots, Q_{m_k}^k\} \ ] \ .
\end{array}
\]
Recall that \  $\{P_1^1, \ldots, P_{m_1}^1, P_1^2, \ldots, P_{m_j}^j\} = \{F_1, \ldots, F_n, Time\}$ so we easily obtain (untimed) facts of the configuration $\Sscr$. We timestamp these facts as follows. 
From $\Delta$ we can extract  natural numbers to be assigned as integer values of timestamps of facts:
We assign $I_1 =0$ as the integer part of the timestamp of each of the facts \ $P_{1}^1, \ldots, P_{m_1}^1$.
As $I_{i+1}$, the integer part of the timestamp of each of the facts \ $P_{1}^{i+1}, \ldots, P_{m_{i+1}}^{i+1}$, we set
 $$ 
I_{i+1} = 
\left\{\begin{array}{cl}
I_i + \delta_{i,i+1}\ , \ & \textrm{ if } \ \delta_{i,i+1} \leq \Dmax\\
 I_i + \Dmax+1\ , \ & \textrm{ if } \ \delta_{i,i+1} = \infty
\end{array}\right.
\ , \quad  i = 1, \dots, j-1 \ .
$$
Then, from $\Uscr$ we can extract values to be assigned as decimal values of timestamps. We set 0 as the decimal part of facts $Q_{1}^0, \ldots, Q_{m_0}^0$, and $\frac{i}{k+1}$ as the decimal part of facts $Q_{1}^i, \ldots, Q_{m_i}^i$, $\forall i= 1, \dots , k$.
Obtained timestamped facts form configuration $\Sscr$ corresponding to circle-configuration $\Ascr$.

\medskip
Let \ $\Ascr_1  \to_{Next}  \Ascr_2$. 
From $\Ascr_1$ we extract a corresponding concrete configuration $\Sscr_1$ as described above.
 Recall that  $Next$ is a set of 8 rules representing time advancement over circle-configurations, as given in Figures  \ref{fig:time-advance} and \ref{fig:time-advance-cont}. 
Depending on the position of the fact $Time$ in $\Ascr_1$ only one of the 8 rules from $Next$ applies.
Assume that $\Ascr$ fits the first case, \ie,~$Rule \ 0$ in Figure \ref{fig:time-advance}. 
Then, any value $\varepsilon \in \langle 0, \frac{1}{k+1} \rangle$, \eg,~$\varepsilon = \frac{1}{2(k+1)}$,  \ is suitable for an instance of the $Tick$ rule, 
$Time@T \lra Time@(T + \varepsilon)$,
 in  $\Sscr_1 \lra_{Tick} \Sscr_2$ so that $\Ascr_2$ corresponds to configuration $\Sscr_2$.

In the same way we can take $\varepsilon \in \langle 0, \frac{1}{k+1} \rangle$  \ in cases corresponding to $Rule \ 2$ and $Rule \ 3$ from Figure \ref{fig:time-advance}. 
For the remaining rules, \ie,~$Rule \ 1$ and $Rule \ 4$ to $Rule \ 7$, \  taking  \ $\varepsilon = \frac{1}{k+1}$\ results in $\Ascr_2$ that corresponds to configuration $\Sscr_2$. That is because the timestamp of $Time$, say $\frac{j}{k+1}$,  is increased exactly so that it matches the decimal part of next group of facts, that is $\frac{j+1}{k+1}$ or $0$.

\medskip
For the opposite direction let  \ $\Sscr_1 \lra_{Tick} \Sscr_2$ \  by means of the instance of $Tick$ using the actual value $\varepsilon$, written \ $ \Sscr_1 \lra_{Tick_{\varepsilon}} \Sscr_2$.
Configurations  $\Sscr_1 $ and $ \Sscr_2$ differ only in the timestamp of the fact $Time$. 
It may well be that \ $\Sscr_1$ and $\Sscr_2$ correspond to the same circle-configuration \ $\Ascr_1 =\Ascr_2$. Then, to the given $Tick$ action, \ $\Sscr_1 \lra_{Tick} \Sscr_2$, corresponds the empty sequence of $Next$ actions, \ $\Ascr_1 \lra_{Next}^* \Ascr_2$. 
In the rest of the proof we assume that ~$\Sscr_1$ and $\Sscr_2$ do not correspond to the same circle configuration.

 The effect of 
the $Tick_{\varepsilon}$ rule is the same as the effect of consecutively applying a series of $Tick$ rules using $\varepsilon_i$ where $\varepsilon_1 + \dots + \varepsilon_n = \varepsilon$ . We can choose a sequence of values  $\varepsilon_i$ so that to the $Tick$ for $\varepsilon_i$, for each $i$, corresponds a single application of the $Next$ rule:
\[
 \begin{array}{ccccccccc}
\Sscr_1 & \to_{Tick_{\varepsilon_1}}&   \Sscr^1 & \to_{Tick_{\varepsilon_2 }}&   \Sscr^2 & \to_{Tick_{\varepsilon_3}}&  \dots & \to_{Tick_{\varepsilon_n}}&   \Sscr_2\\[5pt]
\arr   &     &  \arr   &     &  \arr &  &   \dots   & & \arr
 \\[5pt]
\Ascr_1 & \to_{Next}& \ \Ascr^1 & \to_{{Next }}&   \Ascr^2 & \to_{{Next}}&  \dots & \to_{{Next}}&   \Ascr_2
\end{array}
\]
Consecutive application of such $Tick_{\varepsilon _i}$  corresponds to a series of the $Next$ rules as given in Figures  \ref{fig:time-advance} and \ref{fig:time-advance-cont}.
Namely, each $\varepsilon _i$ corresponds to one of the 8 Next rules.

Concrete values of timestamps in $\Sscr$ reflect in the position of the fact $Time$ in $\Ascr_1$.
Let the fact $Q@{T_Q}$ be one of the facts  in $\Sscr$  with the 
smallest  difference \ $\dec{T_Q} - \dec{ T}> 0$. 
If the global time $T$ has the decimal part equal to some other fact in the configuration, taking $0< \varepsilon < \dec{T_Q} - \dec{ T}$ corresponds to the application of either $Rule \ 0$, $Rule \ 2$ or $Rule \ 3$. Otherwise, if the decimal part of $T$ is different than the decimal part of all other facts in $\Sscr_1$, taking  \ $\varepsilon_i= \dec{T_Q} - \dec{ T}$ \ results in the new global time that matches the  timestamp $ {T_Q}$ in its decimal part.
\end{proof}

\vspace{2mm}
\begin{theorem}
\label{th:circle}
 Let $\Tscr$ be a reachability problem, $\Dmax$ be an upper bound on the numeric values in $\Tscr$. 
 Then \ $\Sscr_I \lra^* \Sscr_G$ \ for some initial and goal configurations $\Sscr_I$ and $\Sscr_G$
 in $\Tscr$ if and only if  \ $\Ascr_I \lra^* \Ascr_G$ \ where $\Ascr_I$ and $\Ascr_G$ are the circle-configurations
 of $\Sscr_I$ and $\Sscr_G$, respectively.
\end{theorem}
\begin{proof}

The proof is by induction on the length of the given plan.
Additionally, as in the proof of Lemma~\ref{th:circle_action}, we can assume that each instance of the time advancement rule in the plan over configurations uses a concrete value $\varepsilon$ that 
corresponds to a single application of one of the $Next$ rules in the matching plan over circle-configurations, 
or to no action at all (\ie,~to an empty sequence of $Next$ actions). 
We get the following bisimulation:
\[
\begin{array}{cccccccc}
\Sscr_I & \to_{r_1} \dots \to_{r_{i-1}} &  \Sscr_{i-1} & \to_{r_i} \  & \ \Sscr_{i}  & \to_{r_{i+1}}  \dots \to_{r_{n}} & \Sscr_G
\\
\arr & &   \arr &     &  \arr &     &  \arr
\\
\Ascr_I &\to_{r_1'} \dots \to_{r_{i-1}'} & \Ascr_{i-1} & \to_{r_i'}&   \Ascr_i & \to_{r_{i+1}'} \dots \to_{r_{n}'} & \Ascr_G
\end{array}
\]
where $r_i'$ is either the instantaneous action $[r_i]$  over circle-configurations, one of the $Next$ rule as given in Figures  \ref{fig:time-advance} and \ref{fig:time-advance-cont}, or an empty action (in which case $ \Ascr_{i-1} = \Ascr_i $ ).
\end{proof}

This theorem establishes that the set of plans over circle-configurations is a sound and complete representation 
of the set of plans with dense time. This means that we can search for solutions of problems symbolically, that is, 
without writing down the explicit values of the timestamps, \ie, the real numbers, in a plan.

\section{Complexity Results}
\label{sec:complex}

\vspace{-1.5mm}

This section details some of the complexity results for the reachability problem.

Reachability problem is a rather general problem that can have various applications. For example, the secrecy problem from the field of protocol security can be considered as an instance of the reachability problem.
Namely,  by interacting with the protocol run, the goal of the intruder is to learn a secret that is initially only known to another participant of the protocol.

In our previous work we have already considered such an application of our formal models in protocol security. In~\cite{kanovich13ic,kanovich14comlan} we consider bounded memory protocol theories and intruder theories and the corresponding secrecy problem. These formalizations bound the number of concurrent protocol sessions. However, the total number of protocol sessions and the total number of fresh values created is unbounded.

The model introduced in this paper can similarly  be applied in the protocol security analysis. As in~\cite{kanovich13ic,kanovich14comlan}, when using only balanced actions, one can necessarily consider only a bounded number of \emph{concurrent} sessions. Nevertheless, the analysis would consider an unbounded number of protocol sessions in \emph{total}.

In addition, by containing the dimension of time, our model presented in this paper can be used for the analysis of time-sensitive protocols, such as cyber-physical security protocols,  and their properties.

\medskip
\paragraph{Conditions for Decidability}

From the literature, we can infer some conditions for decidability of the reachability problem in general:

\begin{enumerate}
 \item \emph{Upper Bound on the Size of Facts}: In general, if we do not assume an upper bound 
 on the size of facts appearing in a plan, where the size of facts is the total number of predicate, function, constant and variable 
 symbols it contains (\eg,~the size of $P(f(a), x,a)$ is 5), then it is easy to encode the Post-Correspondence
 problem which is undecidable, see \cite{cervesato99csfw,durgin04jcs}.\footnote{We leave for Future Work the investigation 
 of specific cases, \eg, protocol
 with tagging mechanisms, where this upper bound may be lifted~\cite{ramanujam03fsttcs}.} Thus we will assume an upper bound on the 
 size of facts, denoted by the symbol $\uSize$.
 
 \item \emph{Balanced Actions}: An action is balanced if its pre-condition has the same number of facts as its post-condition~\cite{kanovich11jar}. 
 The reachability problem  is undecidable for (un-timed) systems with possibly unbalanced actions 
 even if the size of facts is bounded~\cite{cervesato99csfw,durgin04jcs}.
 In a balanced system, on the other hand, the number of facts in any configuration in a plan is 
 the same as the number of facts of the initial configuration, allowing one to recover decidability under some additional conditions.
 We denote the number of facts in the configuration by the symbol $\iSize$.

Both  undecidability results related to (un)balanced actions as well as the upper bound on the size of facts are time irrelevant, they carry over to systems with dense time.
 
\item\emph{Conditions for Timestamps and Time Constraints}:
Recall the form of instantaneous rules in our systems (Equation (\ref{eq:instantaneous}) in Section~\ref{sec:msr}) and the conditions the form of the rules imposes on timestamps and time constraints.

 In particular, the form of instantaneous rules imposes the restriction on using only natural numbers (for $D$s and $D_i$s) in time constraints and timestamps of created facts. This is not as restrictive as it may appear. Namely, we are able to capture systems that would allow rational numbers as well. These numbers are constants in any given concrete model, and they are used to specify concrete time requirements. 
By allowing  rational numbers we would not add to expressivity of the system, since these constants can be multiplied with a  common multiple of their denominators, to obtain natural numbers. The only difference between the original model (with rationals) and the obtained model (with only natural numbers) is the intended denotation of the time units used in the model. For example, one could use minutes instead of days or hours. The intended semantics of the system represented would not change.
In fact, to maintain all our results valid, it suffices to assume that all these numerical constants 
 mentioned within the above constraints and timestamps are commensurable, that is, these constants 
can  be given in terms of a common unit.

Furthermore, the form of  instantaneous rules restricts the form of timestamps of created facts and the form of time constraints.
Timestamps are necessarily of the form $T +D$, where $T$ is the current time of the enabling configuration and $D$ a natural number, and
 time constraints are relative, \ie, involve exactly two time variables from the pre-condition of the rule.
In~\cite{kanovich15mscs} we have shown that  relaxing either of the form of timestamps or the form of time constraints leads to the undecidability of the reachability
problem for multiset rewriting models with discrete time. The same would thus lead  to the undecidability of the reachability problems considered in this paper.
In particular, we  get undecidability  for systems with non-relative time constraints that involve three or more time variables. 
Similarly, we fall into undecidabilty if we allow timestamps of the created facts to be (already linear) polynomials of time variables from the pre-condition of the rule. See ~\cite{kanovich15mscs} for more details.

\end{enumerate}

\medskip
\begin{corollary}
 The reachability problem for our model is undecidable in general.
\end{corollary}
\begin{proof} 
To an instance of a reachability problem for  (un-timed) systems with possibly unbalanced actions, we associate a timed version of the same problem as follows. 

We attach arbitrary timestamps to facts in the initial configuration. Similarly, to each (untimed) fact in the goal configuration we attach a different time variable.  We put no constraints to the initial and goal configurations.
This way, timestamps of facts in the initial and the goal configuration play no particular role in the reachability problem.

Finally, we consider the following set of instantaneous rules: For every rule 
$$
\begin{array}{l}
 W_1 \ldots, ~W_k,~F_1, \ldots, ~F_n ~\lra~ 
  \exists~\vec{X}.~[~W_1, \ldots, ~W_k, ~Q_1, \ldots, ~Q_m~]
\end{array}
$$
in the un-timed system we take the corresponding rule
$$
\begin{array}{l}
 Time@T, ~W_1@T_1, \ldots, ~W_k@T_k, ~F_1@T_1', \ldots, ~F_n@T_n'  ~\lra~ \\
 \quad \exists~\vec{X}.~[~Time@T, ~W_1@T_1, \ldots, ~W_k@T_k, ~Q_1@T, \ldots, ~Q_m@T~]
\end{array}
$$
Notice that above rules have no guards attached and since facts have arbitrary timestamps,  the application of rules in the timed system is independent from the time aspect.

Clearly, the original (untimed) reachability problem has a solution if and only if its associated timed version (as described above) has a solution.
Since the reachability problem for  (untimed) systems with possibly unbalanced actions is undecidable, the reachability problem for timed systems with possibly unbalanced actions is undecidable as well.
\end{proof}

\paragraph{PSPACE-Completeness} We show that the reachability problem for our model with dense time and balanced actions is PSPACE-complete. Interestingly, the same problem is also PSPACE-complete when using models with discrete time~\cite{kanovich12rta}.

 Given the machinery in Section~\ref{sec:circle}, we can re-use many results in the literature to show 
 that the reachability problem is also PSPACE-complete for balanced systems with dense time  
 that can create fresh values, as given in Section~\ref{sec:msr}, assuming an upper bound on the 
 size of facts. For instance, we use the 
 machinery detailed in~\cite{kanovich13ic} to handle the fact that 
 a plan may contain an unbounded number of fresh values.
  
 The PSPACE lower bound can be inferred from~\cite{kanovich13ic}. The interesting bit is to 
 show PSPACE membership of the reachability problem. The following lemma establishes an upper bound
 on the number of different circle-configurations:

 \begin{lemma}\label{Lemma:number of states}
Given a reachability problem $\Tscr$ under a finite alphabet $\Sigma$, 
an upper bound on the size of facts, $\uSize$, and
an upper bound, $\Dmax$, on the numeric values appearing in $\Tscr$, then the
number of different circle-configurations, denoted
by $L_\Tscr(\iSize,\uSize,\Dmax)$, with $\iSize$ facts (counting 
repetitions) is
$$
 L_\Tscr(\iSize,\uSize,\Dmax) \leq J^\iSize  (D + 2 \iSize \uSize)^{\iSize
\uSize} \iSize^\iSize (\Dmax+2)^{(\iSize-1)} ,
$$
where $J$ and $D$ are, respectively, the number of predicate  
and the number of constant and function symbols 
in $\Sigma$.
\end{lemma}
\begin{proof} 
A circle-configuration consists of 
a $\delta$-configuration $\Delta$:
$$
 \Delta = \left\langle \
\{Q_{1}^1, \ldots, Q_{m_1}^1\}, \delta_{1, 2}, ~\{Q_{1}^2, \ldots, Q_{m_2}^2\}, \ldots ,    ~\delta_{j-1, j}, ~\{Q_{1}^j, \ldots, Q_{m_j}^j\} \ 
 \right\rangle
$$
 and unit circle $\Uscr$:
$$
\Uscr = [~\{Q_{1}^0, \ldots, Q_{m_0}^0\}_\Zscr, ~\{Q_{1}^1, \ldots, Q_{m_1}^1\}, \ldots, ~\{Q_{1}^j, \ldots, Q_{m_j}^j\}~] \ .
$$

In each component, $\Delta$ and $\Uscr$,  there are $m$ facts, therefore there are  $\iSize$ slots for predicate names and at most $\iSize \uSize$ slots for  constants and function symbols. Constants can be either constants in the initial alphabet 
$\Sigma$ or names for fresh values (nonces). Following \cite{kanovich13ic} and Definition \ref{def:equivalence}, we need to consider only $2\iSize \uSize$ names for fresh values (nonces). Finally, only
time differences up to $\Dmax$ have to be considered together with the symbol $\infty$, and there
are at most $\iSize-1$ slots for time
differences $\delta_{i,j}$ in 
$\Delta$. 

Finally, for each $\delta$-configuration, there are at most $\iSize^\iSize$ unit circles as for each 
fact $F$ we can assign a class, $\Uscr(F)$, and there are at most $\iSize$ classes. 
\end{proof}

Intuitively, our upper bound algorithm keeps track of
the length of the plan it is constructing and if its length exceeds
$L_\Tscr(\iSize,\uSize,\Dmax)$, then
it knows that it has reached the same circle-configuration twice.
Hence, there is a loop in the plan which can be avoided. If there is a solution plan of length that exceeds 
$L_\Tscr(\iSize,\uSize,\Dmax)$,
there is also a shorter plan that is the solution to the same reachability problem. Therefore, we can nondeterministically search for plans of bounded length.
 This search
is possible in PSPACE since the above number,
when stored in binary, occupies only polynomial  space with respect
to its parameters.

For the given reachability problem  $\Tscr$, where $m$ is the number of facts in the initial configuration and $k$ an upper bound on the size of facts, assume functions  $\Gscr$ and $\Rscr$ that return the value 1 in 
Turing space bounded by a polynomial in ~$m, k$ and $log_2(\Dmax)$ when given as input, respectively,  a circle-configuration that is a goal circle-configuration, and a pair of a circle-configuration and a transition that is valid, \ie,~an instance of an action in $\Tscr$ that is applicable to the given circle-configuration, and return 0 otherwise.

\vspace{2mm}
\begin{theorem}
\label{th:PSPACE}
  Let $\Tscr$ be a reachability problem with balanced actions, 
  $\Sscr_0$ be an initial configuration with exactly $m$ facts, 
  $\Dmax$ be an upper bound on the numeric values appearing in $\Tscr$,
  $k$ an upper bound on the size of facts and let 
  $\Gscr$ and $\Rscr$ be functions as described above.
  Then there is an algorithm that given an initial configuration $\Sscr_0$ decides
  the problem $\Tscr$ and the algorithm runs in polynomial space with respect to 
  $m,k$ and $log_2(\Dmax)$.
 \end{theorem}
\begin{proof}
Let $m$  be the number of facts in the initial configuration $\Sscr_I$,
$k$  an upper bound on the size of facts,  $\Dmax$ a natural number that is an upper bound on the numeric values appearing in the
reachability problem $\Tscr$, that is in the timestamps of the initial configuration, or the $D$s and $D_i$s in constraints or actions of the given reachability problem.

We modify the algorithms given in \cite{kanovich13ic} and in \cite{kanovich15mscs}  in order to  accommodate explicit dense time. 
The algorithm 
must accept whenever there is a sequence of actions from $\Tscr$ leading  from the
initial configuration $\Sscr_I$ to a goal configuration.
In order to do so, we construct an algorithm that searches non-deterministically whether such a plan exists. 
Then we apply Savitch's Theorem to determinize this algorithm. 

A crucial point here is that instead of searching for a plan using
concrete values, we rely on the equivalence among configurations given in Definition \ref{def:equivalence}
and use circle-configurations only. Theorem~\ref{th:circle}
guarantees that this abstraction is sound and faithful. 

Let $i$ be a natural number such that \ $0 \leq i \leq L_T(m,k ,\Dmax)$.

The algorithm begins with $\Ascr_0$ set to be the circle-configuration of $\Sscr_I$ and iterates the following sequence of operations:

\begin{enumerate}
 \item If $\Ascr_i$ is representing a goal configuration, \ie,
if $\Gscr(\Ascr_i) = 1$, then return ACCEPT; otherwise continue;
  \item If $i > L_T(m,k, \Dmax)$, then FAIL;
else continue;
 \item Guess non-deterministically an action, $r$, from $\Tscr$ applicable to $\Ascr_i$,
\ie, such an action $r$ that  $\Rscr(\Ascr_i, r) = 1$. If no such action exists, then
return FAIL. Otherwise 
replace $\Ascr_i$ with the  circle-configuration $\Ascr_{i+1}$
resulting from applying the action $r$ to the circle-configuration
$\Ascr_i$. 
  \item Set i = i + 1.
\end{enumerate}

We now show that this algorithm runs in polynomial space.
We start with the step-counter $i$:
The greatest number reached by this counter is $L_T(m,k,\Dmax)$. 
When stored in binary encoding, this number only takes 
space that is polynomial in the given inputs:
\[
\begin{array}[]{lcl}
 \log(L_T(m,k,\Dmax)) \leq 
m\log(J) +  mk\log(D + 2 mk) + m \log m +
(m-1)\log(\Dmax + 2) .
\end{array}
\]
Therefore, one only needs polynomial space to store the values of the 
step-counter.

We must also be careful to check 
that any circle-configuration, $W_i = \tup{\Delta, \Uscr}$, 
 where
$$
\begin{array}{l}
\, \Delta = \left\langle \
 \{Q_{1}^1, \ldots, Q_{m_1}^1\}, ~\delta_{1, 2}, ~\{Q_{1}^2, \ldots, Q_{m_2}^2\}, \ldots ,    ~\delta_{k-1, k}, ~\{Q_{1}^k, \ldots, Q_{m_k}^k\} \
 \right\rangle \ ,
\\
\Uscr = [ \ \{Q_{1}^0, \ldots, Q_{m_0}^0\}_\Zscr, ~\{Q_{1}^1, \ldots, Q_{m_1}^1\}, \ldots, ~\{Q_{1}^j, \ldots, Q_{m_j}^j\} \ ] 
\end{array}
$$
can be stored in space that is polynomial to the given inputs. 
Since our system is balanced, the size of facts is bounded, and the
values of the truncated time differences, $\delta_{i,j}$, are bounded, 
it follows that the size of any circle-configuration $\tup{\Delta, \Uscr}$
 in a plan is polynomially bounded with respect to \ $m,k$ and $\log(\Dmax + 2) .$

Finally, the algorithm needs to keep track of the action $r$ guessed when 
moving from one circle-configuration to another, and for the scheduling of a plan.
It has to store the action that has been used at the $i^{th}$ step. Since 
any action can be stored by remembering two circle-configurations, 
one can also store these actions in space polynomial to the inputs.
\end{proof}

\section{Related and Future Work}
\label{sec:related}

The formalization of timed models and their use in the analysis of cyber-physical security protocols has already been investigated.  
We review this literature.

Meadows \etal~\cite{meadows07booktitle} and Pavlovic and Meadows in \cite{pavlovic09spw} propose and use a logic called Protocol Derivation Logic (PDL) to formalize and prove the safety of a number of cyber-physical protocols. In particular, they specify the assumptions and protocol executions in the form of axioms, specifying the allowed order of events that can happen, and show that safety properties are implied by the axiomatization used. They do not formalize an intruder model. Another difference from our work is that their PDL specification is not an executable specification, while we have implemented our specification in  Maude~\cite{clavel-etal-07maudebook}. Finally, they do not investigate the complexity of protocol analysis nor investigate the expressiveness of formalizations using discrete and continuous time.

Another approach similar to \cite{meadows07booktitle} in the sense that it uses a theorem proving approach is given by Schaller \etal~\cite{basin11iss}. They formalize an intruder model and some cyber-physical security protocols  in Isabelle. They then prove the correctness of these protocols under some specific conditions and also identify attacks when some conditions are not satisfied.  Their work was a source of inspiration for our intruder model specified in~\cite{kanovich14fccfcs}, which uses the model described in Section~\ref{sec:msr}. Although their model includes time, their model is not refined enough to capture the attack in-between-ticks as they do not consider the discrete behavior of the verifier. 

Boureanu \etal~\cite{boureanu13iacr} proposed a 
discrete time model for formalizing distance bounding protocols and their security requirements. Thus they are more interested in the computational soundness of distance bounding protocols by considering 
an adversary model based on probabilistic Turing machines. They claim that their SKI protocol is secure against a number of attacks. However, their time model is discrete where all players are running at the same clock rate. Therefore, their model is not able to capture attacks that exploit the fact that players might run at different speeds.

Pavlovic and Meadows \cite{pavlovic10mfps} construct a probabilistic model for analysing distance bounding protocols against guessing attacks. The nature of the attack we consider and the attack considered in~\cite{pavlovic10mfps} are quite different leading to different probabilistic models with different goals.

The Timed Automata~\cite{alur04sfm} (TA) literature contains models for cyber-physical protocol analysis. Corin \etal~\cite{corin07jcs} formalize protocols and the standard Dolev-Yao intruder as timed automata and demonstrate that these can be used for the analysis. They are able to formalize the generation of nonces by using timed automata, but they need to assume that there is a bound on the number of nonces. This means that they assume a bound on the total number of protocol sessions. Our model based on rewrite theory, on the other hand, allows for an unbounded number of nonces, even in the case of balanced theories~\cite{kanovich13ic}.
Also they do not investigate the complexity  of the analysis problems nor the expressiveness difference between models with discrete 
and continuous time. Lanotte \etal~\cite{lanotte10actinf} specify cyber-physical protocols, but protocols where messages can be re-transmitted or alternatively a protocol session can be terminated, \ie, timeouts, in case a long time time elapses. They formalize the standard Dolev-Yao intruder. Finally, they also obtain a decidability result for their formalism and an EXPSPACE-hard lower bound for the reachability problem. It seems possible to specify features like timeouts and message re-transmission, in our rewriting
formalism. 

We also point out some important differences between our PSPACE-completeness proof and PSPACE-completeness proof for timed automata~\cite{alur04sfm}. A more detailed account can be found in the Related Work section of \cite{kanovich15mscs}. The first difference is that we do not impose any bounds on the number of nonces created, while the TA proof normally assumes a bound. The second difference is due to the first-order nature of rewrite rules. The encoding of a first-order system in TA leads to an exponential blow-up on the number of states of the automata as one needs take into account all instantiations of rules. Finally, the main abstractions that we use, namely circle-configurations, are one-dimensional, while regions used in the TA PSPACE proof are multidimensional.

Malladi \etal~\cite{malladi10corr} formalize distance bounding protocols in strand spaces. 
They then construct an automated tool for protocol analysis using a constraint solver. They did not take into account the fact that the verifier is running a clock in their analysis and therefore are not able to detect the attack in-between-ticks.

Cheval and Cortier~\cite{cheval15post} propose a way to prove the properties based on the observational equivalence of processes taking account the time processes execute with time by reducing this problem to the observational equivalence based on the length of inputs. They are able to automatically show that RFID protocols used by passports suffer a privacy attack. While the problems considered are different, it seems possible to use our machinery to prove complexity results on problems based on observational equivalence. We leave this as future work.

Nigam \etal~\cite{nigam16esorics} investigated further the properties of timed intruder models formalized in a language similar to our timed MSR framework. It is possible to infer an upper-bound on the number of time intruders to consider for cyber-physical security protocol verification. Based on this result, they implemented in Maude, also using SMT solvers, a prover that automates such a verification of protocol security including the attack in-between-ticks introduced here. 

Finally, \cite{cremers12oakland} introduces a taxonomy of attacks on distance bounding protocols, which include a new attack called Distance Hijacking Attack. This attack was caused by failures not in the time challenges phase of distance bounding protocols, but rather in the authentication phases. It would be interesting to understand how these attacks can be combined with the attack in-between-ticks to build more powerful attacks. 

As future work, we also intend to investigate the challenge-response approach in providing security of distance bounding protocols by expanding our probabilistic analysis to such a scenario. 
Repeating a number of challenge and response rounds in distance bounding protocols is generally  believed to mitigate the chances of an attack occurring. 
In particular, we are planning to analyze  whether the effects of the attack in-between-ticks can be reduced by repeated challenge-response rounds of protocols. Moreover, we are investigating improvements on our implementation in order to check for a wider number of attacks automatically.

\emph{Acknowledgments:}
Nigam is supported by the Brazilian
Research Agencies CNPq and Capes.
Talcott is partially supported by NSF grant CNS-1318848, by a grant from ONR and by a grant from Capes Science without Borders.
Scedrov is supported in part by the
AFOSR MURI ``Science of Cyber Security: Modeling, Composition,
and Measurement''.
Additional support for Scedrov from ONR. Part of
the work was done while Kanovich and Scedrov were
visiting the National Research University Higher School of
Economics, Moscow. They would like to thank Sergei O.
Kuznetsov for providing a very pleasant environment for work. 
This article was prepared within the framework of the Basic Research Program at the National Research University Higher School of Economics (HSE) and supported within the framework of a subsidy by the Russian Academic Excellence Project '5-100'.
We would like to thank Robin Pemantle for discussions, suggestions, and advice regarding the probabilistic analysis.

\bibliographystyle{abbrv}

\end{document}